\documentclass[12pt,a4paper]{article}

\usepackage[truedimen,margin=30mm]{geometry} 

\usepackage{amssymb}
\usepackage{amsmath}
\usepackage[hiresbb]{graphicx}
\usepackage{amsthm}   %This is necessary for "proof"
\usepackage{authblk}   %This is necessary for author descriptions
\usepackage{lscape}   %This is necessary for "landscape"
\usepackage{url}
\usepackage{multirow}
\usepackage{natbib}

\usepackage{setspace}

\usepackage{times}
\usepackage{dsfont}
\usepackage{subcaption}
\usepackage{here}
\usepackage{booktabs, threeparttable}

\makeatletter
\def\mojiparline#1{
    \newcounter{mpl}
    \setcounter{mpl}{#1}
    \@tempdima=\linewidth
    \advance\@tempdima by-\value{mpl}zw
    \addtocounter{mpl}{-1}
    \divide\@tempdima by \value{mpl}
    \advance\kanjiskip by\@tempdima
    \advance\parindent by\@tempdima
}
\makeatother
\def\linesparpage#1{
    \baselineskip=\textheight
    \divide\baselineskip by #1
}

\newtheorem{prop}{Proposition}

\allowdisplaybreaks[4]

\newcommand{\bld}{\boldsymbol}

\title{ {\bf Efficient Bayesian Inference in the Cox Model via Rank-Ordered Likelihood } }

\author[1]{Tomohiro Ohigashi}
\author[2]{Shunichiro Orihara}
\author[3]{Shonosuke Sugasawa}

\affil[1]{Department of Information and Computer Technology, Faculty of Engineering, Tokyo University of Science, Tokyo, Japan}
\affil[2]{Department of Health Data Science, Tokyo Medical University, Tokyo, Japan}
\affil[3]{Faculty of Economics, Keio University, Tokyo, Japan}

\date{}

\begin{document}
\linesparpage{25}
\allowdisplaybreaks[4]
\begin{singlespace}
\maketitle
\end{singlespace}

\vspace{-1cm}
\begin{center}
{\large\bf Abstract}
\end{center}
In Bayesian inference for the Cox proportional hazards model, modeling the baseline hazard function is challenging. 
Recently, direct Bayesian inference using the partial likelihood has been considered in the framework of general Bayesian inference. 
In terms of posterior computation, several studies have examined sampling algorithms under the Cox model. 
In this study, we develop two P\'{o}lya--Gamma-augmented sampling algorithms based on the Plackett--Luce and generalized Plackett--Luce rank-ordered likelihoods, referred to as PL-Cox and GPL-Cox, respectively. 
PL-Cox handles tied events through the Breslow approximation and relies on a negative-binomial approximation for posterior computation, whereas GPL-Cox models ties through a geometric latent ranking mechanism. 
Both methods can be implemented without Metropolis--Hastings correction and are extensible to shared frailty models.
In a simulation study, we considered multiple survival model settings, including continuous and discrete survival time models, as well as scenarios with varying degrees of ties, and found that the PL-Cox model exhibited relatively stable performance.
A large real-data application showed stable PL-Cox convergence but GPL-Cox difficulties with sparse ties.
A shared-frailty application further illustrated the flexibility of the proposed framework for hierarchical survival modeling.
The proposed methods are implemented in the publicly available R package \texttt{BayesPLCox}.


\section{Introduction}

For survival/time-to-event data analysis, the Cox proportional hazards model \citep{coxRegressionModelsLifeTables1972}, or the Cox model, is commonly used in various fields, including medical research, economics, and reliability engineering.
In the frequentist framework, the partial likelihood is usually employed to estimate hazard ratios for the Cox model, as it does not require the specification of the baseline hazard \citep{coxPartialLikelihood1975}.
For the Cox model, there are well-established methods for handling tied event times, which are frequently observed in survival data, including two approximation methods, the Breslow and Efron methods, and the exact method \citep{kalbfleischStatisticalAnalysisFailure2002}.
These methods within the frequentist framework are implemented using standard statistical software packages and are extensively utilized.

In Bayesian inference for the Cox model,
\citet{kalbfleischNonParametricBayesianAnalysis1978a} provided the first Bayesian interpretation of the Cox model by assigning a gamma process prior to the cumulative baseline hazard, showing that the partial likelihood arises as the marginal likelihood under this formulation.
\citet{sinhaBayesianJustificationCoxs2003} extended this justification, illustrating that the partial likelihood for the regression coefficient aligns with the marginal posterior for the regression coefficient derived from the gamma process prior for the cumulative baseline hazard, even in the presence of time-dependent covariates and tied event times.
More recently, direct Bayesian inference using the partial likelihood has been considered in the framework of general Bayesian inference \citep{bissiriGeneralFrameworkUpdating2016}.
Regarding posterior computation, several studies have examined sampling algorithms for the Cox model, although no standard approach has yet been established.
When the baseline hazard is modeled explicitly, posterior sampling must accommodate a nuisance component, whereas direct use of the partial likelihood generally does not yield closed-form full conditional distributions.
These difficulties become more pronounced in the presence of tied event times, frailty terms, or other multilevel structure.
\citet{sinhaBayesianJustificationCoxs2003} developed a Metropolis--Hastings algorithm with adaptive rejection for the marginal posterior of regression coefficients. 
Additionally, \citet{renCoxPolyaGammaAlgorithmFlexible2025} proposed a Gibbs sampling algorithm for the Cox model.
This method models the log-cumulative baseline hazard using a monotone spline approximation and applies P\'{o}lya--Gamma (PG) data augmentation \citep{polsonBayesianInferenceLogistic2013}, with a negative binomial approximation to the underlying Poisson process. 
Although Metropolis--Hastings corrections are implemented to address approximation bias, this method requires a trade-off between computational efficiency and approximation accuracy.
Furthermore, the treatment of tied event times is not explicitly included in the algorithmic formulation of this method.
More recently, \citet{tamanoEfficientGibbsSampling2025} introduced an efficient sampling procedure based on composite partial likelihood and PG data augmentation.
This method constructs a calibrated target distribution through an affine open-faced sandwich adjustment based on \citet{tamanoLocationScaleCalibrationGeneralized2025}, producing draws aligned with the partial likelihood benchmark.
This method naturally handles tied event times because of the composite partial likelihood.
However, as this method relies on a composite partial likelihood without specifying a full hierarchical data-generating mechanism and employs an estimating-equation-based calibration requiring explicit score and derivative information of the loss, its extension to multilevel survival models with latent frailty components using the proposed framework would be difficult.

In this study, we develop two P\'{o}lya--Gamma-augmented MCMC algorithms based on rank-ordered likelihoods related to the Cox model.
The central concept of the proposed framework is the Plackett--Luce (PL) model, which provides a probabilistic model for rank-ordering.
The PL model is tractable because it admits a generative representation based on the ordering of independent exponential latent variables \citep{bakerNewOrderstatisticsbasedRanking2020, bakerModifyingBradleyTerry2021}.
\citet{bakerModifyingBradleyTerry2021} introduced a model that replaces exponential latent variables with their discrete counterparts, specifically geometric latent variables.
\citet{hendersonModellingAnalysisRank2025} referred to this model as the generalized (or geometric) PL (GPL) model and derived an explicit form for the likelihood that facilitates maximum likelihood and Bayesian inference.
They also provided derivations of efficient Gibbs sampling algorithms.
Motivated by these developments, we develop PL-Cox and GPL-Cox, which combine likelihood formulations based on the PL and GPL models with PG data augmentation.
The two methods handle ties differently: PL-Cox uses the Breslow approximation, whereas GPL-Cox explicitly models tied event sets through geometric latent rankings; neither coincides with the exact Cox partial likelihood for tied events. 
Both admit conditionally Gaussian updates and can be implemented without posterior correction, although PL-Cox relies on a negative-binomial approximation. 
This conditionally Gaussian structure also facilitates extensions to shared log-normal frailty models.
Both methods are implemented in the
R package \texttt{BayesPLCox}, available at
\url{https://github.com/tom-ohigashi/BayesPLCox}.

The remainder of this paper is organized as follows.
In Section \ref{preliminaries}, we review the Cox model, the rank-ordered likelihood formulations based on the PL and GPL models, and the rank-ordered likelihood representation for the Cox model.
In Section \ref{method}, we develop the proposed PL-Cox and GPL-Cox models and describe posterior computation.
In Section \ref{simulation}, we present simulation studies to evaluate the performance of the proposed methods.
In Section \ref{application}, we illustrate the proposed methods through two applications to real datasets.
We conclude our paper in Section \ref{discussion} with further discussion.

\section{Preliminaries\label{preliminaries}}

\subsection{Cox model}
For each subject $i \in \{1,\ldots,n\}$, we observe $\boldsymbol{D} = \{D_i\}^n_{i=1} = \{(T_i, \delta_i, \boldsymbol{x}_i)\}^n_{i=1}$, where $T_i \in \mathbb{R}_{+}$ denotes the observed time defined as $T_i = \min (T^\ast_i, C^\ast_i)$, with $T_i^\ast$ and $C_i^\ast$ representing the event and censoring times, respectively.
Let $\delta_i$ represent the event indicator, where $\delta_i = 1$ if $T^\ast_i \leq C^\ast_i$; otherwise, $\delta_i = 0$.
Let $\boldsymbol{x}_i$ denote the $p$-dimensional covariate vector.
In the Cox model \citep{coxRegressionModelsLifeTables1972}, the hazard function for subject $i$ at time $t$ is given by $h(t \mid \boldsymbol{x}_i) = h_0 (t) \exp (\boldsymbol{x}^\top_i \boldsymbol{\beta})$ where $h_0(t)$ denotes the baseline hazard function and $\boldsymbol{\beta}$ is a $p$-dimensional vector of regression coefficients.
In frequentist inference based on the partial likelihood \citep{coxPartialLikelihood1975}, the specification of the baseline hazard function is not required.
Accordingly, the partial likelihood can be written as
\begin{equation}\label{partiallike}
    L (\boldsymbol{\beta} \mid \boldsymbol{D}) := \prod^n_{i=1} \left\{ \frac{\exp (\boldsymbol{x}^\top_i \boldsymbol{\beta})}{\sum_{\ell \in R(T_i)} \exp (\boldsymbol{x}^\top_\ell \boldsymbol{\beta})} \right\}^{\delta_i},
\end{equation}
where $R(t)$ is the risk set at time $t$.

The partial likelihood is formulated based on the assumption that event times are continuously distributed, which implies that the probability of ties occurring is zero.
In practice, tied event times may occur due to rounding, grouped observations, or inherently discrete time scales \citep{kalbfleischStatisticalAnalysisFailure2002}. 
Approximation methods, such as the Breslow and Efron approaches, are commonly used, which modify the denominator of the partial likelihood (\ref{partiallike}) within tied blocks.
Conversely, when event times are intrinsically discrete, the exact conditional likelihood based on all possible combinations within tied sets is theoretically appropriate \citep{kalbfleischStatisticalAnalysisFailure2002}.

\subsection{The PL and GPL models\label{sec:plgpl}}
The PL model provides a probabilistic model for rank-ordered data and is widely used for analyses of preferences, competitions, and choice data.
Let $n$ items be associated with positive
parameters $\boldsymbol{\lambda}
=(\lambda_1,\ldots,\lambda_n)^\top$,
which denote the relative strengths of these items.
Let $\boldsymbol{y} = (y_1, \ldots, y_n)$ represent a permutation that ranks the items, where $y_k$ indicates the item ranked at position $k$.
The PL model specifies the probability of observing the ranking $\boldsymbol{y}$ as
\begin{equation}\label{eq_pl}
P(\boldsymbol{Y}=\boldsymbol{y}\mid \boldsymbol{\lambda})
=
\prod_{k=1}^{n-1}
\frac{\lambda_{y_k}}
{\sum_{r=k}^{n} \lambda_{y_r}} .
\end{equation}
An important characteristic of the PL model is its generative representation, which is based on exponential latent variables.
Suppose that $W_i \sim \text{Exp}(\lambda_i)$ independently.
The ranking obtained by ordering the latent variables $W_i$ in ascending order adheres to the PL model.
This representation has been exploited to derive efficient inference algorithms \citep{bakerNewOrderstatisticsbasedRanking2020, bakerModifyingBradleyTerry2021}.

\citet{bakerModifyingBradleyTerry2021} introduced a discrete counterpart of the PL model by replacing exponential latent variables with geometric latent variables.
Let $W_i \sim \text{Geom}(\theta_i)$ represent independent geometric random variables.
The ranking induced by ordering $W_i$ in increasing order defines the GPL model.
\citet{hendersonModellingAnalysisRank2025} derived an explicit likelihood representation for this model and demonstrated that it facilitates both maximum likelihood and Bayesian inference.
Each success probability $\theta_i$ influences both the probability of an item appearing at higher ranks and the incidence of ties.
Larger values of $\theta_i$ imply smaller geometric latent variables on average, so that the corresponding item is more likely to be selected early. 
At the same time, when multiple $\theta_i$ are large, the associated latent variables are more likely to attain the same small value, thereby increasing the probability of tied outcomes.
Thus, unlike the PL model, the GPL model depends on the overall scale of the success probabilities as well as on their relative magnitudes.

\subsection{Rank-ordered likelihood for the Cox model}

The partial likelihood (\ref{partiallike}) can be naturally interpreted within the framework of rank-ordered outcomes. 
Each observed event can be viewed as the selection of a subject from the corresponding risk set.
Under the Cox model, the probability that subject $i$ experiences the event at time $t$ among the subjects in the risk set $R(t)$ is proportional to $\exp(\boldsymbol{x}_i^\top \boldsymbol{\beta})$.
Consequently, the partial likelihood can be interpreted as a sequence of selections from changing risk sets.
This interpretation highlights the close connection between survival models and probabilistic models for ranked data.
\citet{suConnectionBradleyTerry2006} showed that the partial likelihood aligns with the likelihood of the Bradley--Terry (BT) model for rank-order events under a stratified proportional hazards formulation.
Consider two subjects $i$ and $j$ with event times $T_i$ and $T_j$.
Under a common baseline hazard, the probability that subject $i$
fails before subject $j$ is given by
\[
P(T_i < T_j) =
\frac{\exp(\boldsymbol{x}_i^\top \boldsymbol{\beta})}
{\exp(\boldsymbol{x}_i^\top \boldsymbol{\beta}) + \exp(\boldsymbol{x}_j^\top \boldsymbol{\beta})},
\]
which coincides with the probability assigned by the BT model for paired comparisons.

More generally, consider $K$ subjects without censoring, whose event times satisfy
$T_{(1)} < T_{(2)} < \cdots < T_{(K)}$.
The probability of this ordered event is given by
\[
\prod_{k=1}^{K-1}
\frac{\exp(\boldsymbol{x}_{(k)}^\top \boldsymbol{\beta})}
{\sum_{\ell \in R(T_{(k)})} \exp(\boldsymbol{x}_{\ell}^\top \boldsymbol{\beta})}.
\]
This expression follows by sequentially conditioning on which subject experiences the next event. 
Under the Cox model, the common baseline hazard cancels from each conditional probability. 
The same sequential-selection principle yields the Cox partial likelihood when applied to the observed risk sets in
the presence of censoring.
The above representation shows that the partial likelihood can be interpreted as the likelihood of rank-ordered outcomes generated through a sequential selection mechanism.
The contribution of each event time has the same functional form as the likelihood component of the PL model applied to the corresponding risk set.
Therefore, the partial likelihood can be viewed as a special case of a rank-ordered likelihood derived from the PL model (\ref{eq_pl}).
This connection provides a useful perspective for extending the Cox model.
The PL and GPL formulations introduced in Subsection~\ref{sec:plgpl} therefore provide rank-ordered likelihood constructions based on the observed risk sets and event sets.
These formulations facilitate Bayesian inference using standard data augmentation techniques.

\section{Bayesian inference for the PL/GPL Cox models\label{method}}
\subsection{PL-Cox and GPL-Cox models}
We propose a rank-ordered likelihood framework for Bayesian inference based on two formulations, referred to as PL-Cox and GPL-Cox.
Let $t_{(1)} < \cdots < t_{(R)}$ denote the ordered distinct event times among $\{T_i : \delta_i = 1\}$.
For each distinct event time $t_{(r)}$, define $R_r := R\{t_{(r)}\}$ as the corresponding risk set, and let $E_r \subset R_r$ denote the set of subjects who experience the event at time $t_{(r)}$.
Let $d_r = |E_r|$ denote the number of events in the $r$th tied block.

\subsubsection{PL-Cox model}
For the PL-Cox model, the positive weight is defined as $\lambda_i = \exp(\boldsymbol{x}_i^\top \boldsymbol{\beta})$, where $\boldsymbol{\beta}$ is a $p$-dimensional regression coefficient.
The PL-Cox likelihood is invariant to a common shift in the linear predictor, as in the classical Cox partial likelihood. We therefore exclude the intercept from the default PL-Cox specification, while retaining it in GPL-Cox to model the overall propensity for tied events.
The likelihood is given by
\begin{equation}\label{like_PL}
L_{\mathrm{PL}}(\boldsymbol{\lambda})
=
\prod_{r=1}^{R}
\prod_{m=1}^{d_r}
\frac{\lambda_{y_{rm}}}{\sum_{j \in R_r} \lambda_j}
=
\prod_{r=1}^{R}
\frac{\displaystyle \prod_{i \in E_r} \lambda_i}
{\displaystyle \Big( \sum_{j \in R_r} \lambda_j \Big)^{d_r}},    
\end{equation}
where $(y_{r1},\ldots,y_{rd_r})$ denotes an arbitrary ordering of the subjects in $E_r$.
For tied events, the risk set $R_r$ is held fixed across the $d_r$ factors, yielding the Breslow approximation rather than an exact tied-event probability.
In the special case $d_r = 1$ for all $r$, this reduces to the usual Cox partial likelihood.
When $d_r > 1$, the above likelihood coincides with the Breslow approximation for handling ties in the Cox model.
Thus, the PL-Cox model can be interpreted as a ranking-based representation of the Cox model using the Breslow method.

\subsubsection{GPL-Cox model}
In the GPL model, each subject is associated with a success probability $\theta_i \in (0,1)$.
We write $\boldsymbol{\theta} =(\theta_1,\ldots,\theta_n)^\top$,
and the ranking mechanism is induced by independent geometric latent variables.
In the GPL-Cox model, these probabilities are linked to covariates through $\theta_i = \operatorname{expit}(\eta_i)$ where $\eta_i = \boldsymbol{x}_i^\top \boldsymbol{\beta}$, $\boldsymbol{x}_i$ includes the intercept term, and $\boldsymbol{\beta}$ is the corresponding regression coefficient vector.
In contrast to the PL-Cox model, the GPL-Cox model is not invariant to a common shift in the linear predictor.
For each non-intercept covariate, the exponentiated regression coefficient represents the odds ratio for the latent geometric success probability $\theta_i$ associated with a one-unit increase in that covariate, holding the others fixed, rather than a conventional Cox hazard ratio.

At each distinct event time $t_{(r)}$, the observed data consist of a top-$d_r$ ranking on subset $R_r$, in which the subjects in $E_r$ are tied in the top bucket and the remaining subjects in $R_r \setminus E_r$ are unranked below them.
Specializing Henderson's top-$m$ GPL likelihood to this setting yields the contribution at time $t_{(r)}$
\[
L_{\mathrm{GPL},r}(\boldsymbol{\theta})
=
\frac{
\displaystyle
\prod_{i \in E_r} \theta_i
\prod_{i \in R_r \setminus E_r} (1-\theta_i)
}{
\displaystyle
1 - \prod_{i \in R_r} (1-\theta_i)
}.
\]
Accordingly, the GPL-Cox likelihood is
\begin{equation}\label{like_GPL}
L_{\mathrm{GPL}}(\boldsymbol{\theta})
=
\prod_{r=1}^{R}
L_{\mathrm{GPL},r}(\boldsymbol{\theta})
=
\prod_{r=1}^{R}
\frac{
\displaystyle
\prod_{i \in E_r} \theta_i
\prod_{i \in R_r \setminus E_r} (1-\theta_i)
}{
\displaystyle
1 - \prod_{i \in R_r} (1-\theta_i)
}.
\end{equation}

This formulation is interpreted as a ranking-based likelihood that accounts for ties, rather than as a standard continuous-time hazard model.
Under the GPL-Cox model, each subject in a risk set is associated with a geometric latent variable, and the observed event set corresponds to the subjects sharing the smallest realized value.
Hence, $\theta_i = \operatorname{expit}(\eta_i)$ controls both the probability of subject $i$ belonging to the event set and the probability of tied events within the risk set.
Accordingly, the regression coefficients describe covariate effects on a latent ranking mechanism that determines which subjects enter the event set, and the model approaches PL-Cox when the success probabilities are uniformly small; see Subsection~\ref{relationship}.

\subsubsection{Latent variable representation}
The PL-Cox and GPL-Cox models allow convenient representation of latent variables.
In the PL-Cox model, the ranking is induced by exponential latent variables with rates $\lambda_i$.
In the GPL-Cox model, the ranking mechanism can be represented using geometric latent variables. For each risk set $R_r$, a latent variable $Z_r$ is introduced, whose distribution depends on the success probabilities of all subjects in the risk set.
The observed tied event set $E_r$ corresponds to the top bucket determined by the smallest realized geometric latent values in the $r$th risk set.
This representation is useful for posterior computation because the denominator term $1 - \prod_{i \in R_r} (1-\theta_i)$ can be handled through standard data augmentation, leading to an efficient Gibbs sampling algorithm in the next subsection.

\subsection{Posterior computation}

For PL-Cox, a direct Metropolis-type sampler could alternatively be applied to the likelihood in (\ref{like_PL}). 
Our use of P\'{o}lya--Gamma augmentation yields conditionally Gaussian updates without proposal tuning and also facilitates extensions involving Gaussian latent effects. 
For GPL-Cox, the latent-variable formulation arises naturally from the GPL construction and leads directly to a Gibbs sampling scheme. 
The trade-off is that these samplers introduce additional latent variables and are therefore more involved internally than a direct Metropolis-type implementation; these steps are handled in the \texttt{BayesPLCox} package.

\subsubsection{PL-Cox model\label{post_pl}}
To facilitate posterior computation for the likelihood in (\ref{like_PL}), we introduce the latent variables $Z_r > 0, r=1,\ldots,R$, using the identity
\[
\frac{1}{A_r^{d_r}} = \frac{1}{\Gamma(d_r)} \int_0^\infty z_r^{d_r-1}\exp(-A_r z_r)\,dz_r, A_r>0,
\]
which follows directly from the definition of the gamma function.
Applying this identity with $A_r = \sum_{j \in R_r} \lambda_j$, the augmented likelihood can be written as
\[
L_{\mathrm{PL}}^{\ast}(\boldsymbol{\beta},\boldsymbol{Z})
\propto
\prod_{r=1}^{R}
\left[
z_r^{d_r-1}
\prod_{i \in E_r} \lambda_i
\exp\!\left\{
-z_r \sum_{j \in R_r} \lambda_j
\right\}
\right].
\]
Integrating out $\boldsymbol{Z}$ recovers the original PL-Cox likelihood in (3) exactly. 
The variables $Z_r$ are computational auxiliaries introduced through the Gamma integral identity, rather than observed event times.
This expression can be rearranged into an individual-specific form.
Define $c_i = \sum_{r=1}^{R} \mathds{1}(i \in E_r), \zeta_i = \sum_{r=1}^{R} \mathds{1}(i \in R_r)\, z_r$.
Then
\[
L_{\mathrm{PL}}^{\ast}(\boldsymbol{\beta},\boldsymbol{Z})
\propto
\prod_{i=1}^{n}
\lambda_i^{c_i}\exp(-\zeta_i \lambda_i)
\prod_{r=1}^{R} z_r^{d_r-1}.
\]
Hence, conditional on $\boldsymbol{Z}$, the likelihood contribution for subject $i$ has the kernel of a Poisson likelihood with the mean $\zeta_i \lambda_i$.

However, the resulting conditional posterior distribution for $\boldsymbol{\beta}$ is not available in a suitable form for Gibbs sampling.
To address this issue, we employ a Poisson--Gamma construction similar to that, as used by \citet{d2023efficient} and \citet{hamuraRobustBayesianModeling2025}.
We introduce latent variables $\xi_i \sim \mathrm{Gamma}(\delta,\delta)$ independently, where $\delta>0$ is a fixed approximation parameter, and consider the augmented model $c_i \mid \boldsymbol{\beta}, \boldsymbol{Z}, \xi_i \sim \mathrm{Poisson}(\xi_i\, \zeta_i \lambda_i)$.
Since $\mathrm{E}(\xi_i)=1$ and $\mathrm{Var}(\xi_i)=1/\delta$, the distribution of $\xi_i$ concentrates around 1 as $\delta$ increases.
Marginalizing over $\xi_i$ yields the following negative binomial approximation to the original Poisson kernel:
\[
c_i \mid \boldsymbol{\beta}, \boldsymbol{Z}
\approx
\mathrm{NegBinom}
\left(
\delta,\,
\frac{\zeta_i \lambda_i}{\delta+\zeta_i \lambda_i}
\right).
\]

This representation leads to a logistic-type likelihood and allows the use of the PG augmentation of \citet{polsonBayesianInferenceLogistic2013}.
Let $\psi_i = \boldsymbol{x}_i^\top \boldsymbol{\beta} + \log(\zeta_i/\delta)$.
Introducing PG latent variables $\omega_i \mid \boldsymbol{\beta},\boldsymbol{Z},\boldsymbol{D} \sim \mathrm{PG}(c_i+\delta,\psi_i)$ yields a Gaussian full conditional distribution for $\boldsymbol{\beta}$ under a normal prior $\boldsymbol{\beta} \sim \mathrm{N}(\boldsymbol{b}_0,V_0)$.
Let $\kappa_i=(c_i-\delta)/2$, $\kappa=(\kappa_1,\ldots,\kappa_n)^\top$, 
$\Omega=\mathrm{diag}(\omega_1,\ldots,\omega_n)$, and 
$\boldsymbol{o}=(\log(\zeta_1/\delta),\ldots,\log(\zeta_n/\delta))^\top$.
The conditional posterior distribution is
$\boldsymbol{\beta} \mid \boldsymbol{\omega},\boldsymbol{Z},\boldsymbol{D} \sim \mathrm{N}(B^{-1} \boldsymbol{g},B^{-1})$
where $B=X^\top\Omega X+V_0^{-1}$ and $\boldsymbol{g}=X^\top(\kappa-\Omega\boldsymbol{o})+V_0^{-1}\boldsymbol{b}_0$.

The PL-Cox sampler alternates the exact gamma update of $\boldsymbol{Z}$ with a P\'{o}lya--Gamma-based update of $\boldsymbol{\beta}$ using the negative-binomial approximation with finite $\delta$.
The resulting updates are given by
\begin{align*}
Z_r \mid \boldsymbol{\beta}, \boldsymbol{D}
&\sim
\mathrm{Gamma}\!\left(d_r,\sum_{j \in R_r}\lambda_j\right),
\qquad r=1,\ldots,R,\\
\omega_i \mid \boldsymbol{\beta}, \boldsymbol{Z}, \boldsymbol{D}
&\sim
\mathrm{PG}(c_i+\delta,\psi_i),
\qquad i=1,\ldots,n,\\
\boldsymbol{\beta} \mid \boldsymbol{\omega}, \boldsymbol{Z}, \boldsymbol{D}
&\sim
\mathrm{N}(B^{-1}\boldsymbol{g},\,B^{-1}).
\end{align*}
For finite $\delta$, these updates do not constitute an exact Gibbs sampler for the original PL-Cox posterior, since the update for $\boldsymbol{Z}$ is exact whereas that for $\boldsymbol{\beta}$ relies on a negative-binomial approximation. Consequently, the uncorrected sampler is not guaranteed to target the original PL-Cox posterior. Detailed derivations are provided in Appendix~A, and the effect of the approximation and a comparison with a Metropolis--Hastings-corrected sampler are examined in Appendix~D.1 in the Supplementary Material.

\subsubsection{GPL-Cox model}
To facilitate posterior computation for the likelihood in (\ref{like_GPL}), we introduce the geometric latent-variable representation described in the previous subsection.
We introduce $c_i$ and $\zeta_i$ in the same manner as in the PL-Cox model.
Under this representation, the augmented likelihood has the kernel
\[
L_{\mathrm{GPL}}^{\ast}(\boldsymbol{\beta},\boldsymbol{Z})
\propto
\prod_{i=1}^{n}
\theta_i^{c_i}(1-\theta_i)^{\zeta_i-c_i}.
\]
Summing over the discrete latent variables $\boldsymbol{Z}$ recovers the GPL-Cox likelihood in (4) exactly. 
Here, $Z_r$ represents the first latent trial in which at least one subject in the risk set $R_r$ succeeds.

Using the logistic parameterization introduced above, the augmented likelihood can be written as
\[
L_{\mathrm{GPL}}^{\ast}(\boldsymbol{\beta},\boldsymbol{Z})
\propto
\prod_{i=1}^{n}
\frac{\exp(c_i\eta_i)}
{(1+\exp(\eta_i))^{\zeta_i}}.
\]
The resulting augmented likelihood has an exact logistic form, so that
the P\'{o}lya--Gamma augmentation of
\citet{polsonBayesianInferenceLogistic2013} can be applied directly,
yielding conditionally Gaussian updates for $\boldsymbol{\beta}$.

Introducing latent variables $\omega_i \mid \boldsymbol{\beta},\boldsymbol{Z},\boldsymbol{D} \sim \mathrm{PG}(\zeta_i,\eta_i)$ yields a Gaussian full conditional distribution for $\boldsymbol{\beta}$ under a normal prior $\boldsymbol{\beta} \sim \mathrm{N}(\boldsymbol{b}_0,V_0)$.
Let $\kappa_i=c_i-\zeta_i/2$, $\boldsymbol{\kappa}=(\kappa_1,\ldots,\kappa_n)^\top$, and
$\Omega=\mathrm{diag}(\omega_1,\ldots,\omega_n)$.
Then $\boldsymbol{\beta}\mid \boldsymbol{\omega},\boldsymbol{Z},\boldsymbol{D} \sim \mathrm{N}(B^{-1}\boldsymbol{g},B^{-1})$, where $B=X^\top\Omega X+V_0^{-1}, \boldsymbol{g}=X^\top\boldsymbol{\kappa}+V_0^{-1}\boldsymbol{b}_0$.

Therefore, posterior computation for the GPL-Cox model proceeds via Gibbs sampling with updates:
\begin{align*}
Z_r \mid \boldsymbol{\beta},\boldsymbol{D}
&\sim
\mathrm{Geom}\!\left(
1-\prod_{j \in R_r}(1-\theta_j)
\right),
\\
\omega_i \mid \boldsymbol{\beta},\boldsymbol{Z},\boldsymbol{D}
&\sim
\mathrm{PG}(\zeta_i,\eta_i),
\\
\boldsymbol{\beta} \mid \boldsymbol{\omega},\boldsymbol{Z},\boldsymbol{D}
&\sim
\mathrm{N}(B^{-1}g,B^{-1}).
\end{align*}
Detailed derivations of these full conditional distributions are provided in Appendix B in the Supplementary Material.

For computational efficiency, we employ the moment-matched Gaussian approximation of \citet{glynnBayesianAnalysisDynamic} to generate P\'{o}lya--Gamma random variables with large shape parameters in both PL-Cox and GPL-Cox.
This computational approximation is distinct from the finite-$\delta$ approximation used in PL-Cox.
Thus, although the latent-variable augmentation for GPL-Cox is theoretically exact, its numerical implementation also involves approximation.

\subsubsection{Extensions to hierarchical models}
The conditional Gaussian structure induced by PG augmentation, as in \citet{polsonBayesianInferenceLogistic2013}, also makes hierarchical extensions natural.
Therefore, the PL-Cox and GPL-Cox models can be extended to shared frailty models.
For example, consider a log-normal frailty term introduced at the cluster level: under the PL-Cox and GPL-Cox models, the conditional distribution of the frailty effects remains Gaussian, which allows them to be updated within the Gibbs sampler without requiring additional Metropolis--Hastings steps.
Therefore, both the PL-Cox and GPL-Cox models can be easily extended to multilevel survival data settings.

In addition, the proposed framework naturally allows extensions to priors for $\boldsymbol{\beta}$ beyond the normal prior.
For example, priors that can be expressed as scale mixtures of normals, including Student-$t$ and horseshoe priors \citep{carvalhoHorseshoeEstimatorSparse2010}, would be feasible through the introduction of additional latent scale parameters while retaining the Gibbs sampling framework.
Detailed derivations and the corresponding full conditional distributions for the shared log-normal frailty model, Student-$t$ priors, and horseshoe priors are provided in Appendix~C in the Supplementary Materials.

\subsection{Relationship between the PL and GPL models\label{relationship}}

The PL and GPL models provide two probabilistic mechanisms for generating rank-ordered outcomes.
Although the PL model is based on exponential latent variables, the GPL model relies on geometric latent variables.
Despite these differences, the two formulations are closely related.
\citet{hendersonModellingAnalysisRank2025} derived an explicit likelihood representation of the GPL model and showed that the PL model arises as a limiting special case of the GPL model when the success probabilities are parameterized as $\theta_i = \lambda_i \theta_1$ with fixed $\lambda_i > 0$ and $\theta_1 \to 0$, in the absence of ties.

A related limiting relationship also holds for the GPL-Cox model under the logistic link.
Let the linear predictor be expressed as $\eta_i=\alpha+\boldsymbol{x}_i^{\ast \top}\boldsymbol{\beta}$, where $\alpha\in\mathbb{R}$ is the intercept and $\boldsymbol{x}_i^{\ast}$ excludes the intercept term, so that $\theta_i=\operatorname{expit}(\eta_i)$.

\begin{prop}
Under the above parameterization, as $\alpha\to -\infty$, the GPL-Cox likelihood approaches the PL-Cox likelihood, provided that each event time corresponds to a single failure.
\end{prop}

\begin{proof}
As $\alpha\to -\infty$,
\[
\theta_i
=\frac{\exp(\alpha+\boldsymbol{x}_i^{\ast \top}\boldsymbol{\beta})}{1+\exp(\alpha+\boldsymbol{x}_i^{\ast \top}\boldsymbol{\beta})}
\sim \exp(\alpha)\exp(\boldsymbol{x}_i^{\ast \top}\boldsymbol{\beta}),
\]
because
\[
\frac{
\theta_i
}{
\exp(\alpha)\exp(\boldsymbol{x}_i^{\ast \top}\boldsymbol{\beta})
}
=
\frac{1}{1+\exp(\alpha+\boldsymbol{x}_i^{\ast \top}\boldsymbol{\beta})}
\to 1 .
\]
In this limit, the success probabilities become proportional to $\exp(\boldsymbol{x}_i^{\ast \top}\boldsymbol{\beta})$ with a common vanishing scale factor.
Under Henderson's limiting result for the GPL model, the GPL likelihood approaches the PL likelihood.
Applying this argument to each risk set in the Cox construction yields the PL-Cox likelihood.
\end{proof}

This result highlights the role of the intercept in the GPL-Cox model.
Under the logistic link, the intercept controls the overall scale of the success probabilities and prevalence of tied events.
As $\alpha \to -\infty$, all success probabilities become small, and the model retains only relative differences through $\boldsymbol{x}_i^{\ast \top} \boldsymbol{\beta}$, leading to the PL-Cox formulation.

\section{Simulation study\label{simulation}}
\subsection{Settings}
We considered several survival time generating mechanisms, including continuous and discrete survival time models.
The sample size was set to $n=300$.
The covariates consisted of four independent variables generated from the standard normal distribution.
The regression coefficients were set to $\boldsymbol{\beta}^\top_{\text{true}}=(0.10,0.05,-0.15,0.30)$.
Right censoring was introduced independently of the failure time.
For continuous survival time scenarios, censoring times were generated from $\mathrm{Unif}(0.5,30)$.
To induce tied event times, the observed survival times were coarsened using prespecified rounding units of 0.01, 0.1, 0.5, where larger values correspond to coarser discretization of the time scale.
For discrete survival time scenarios, right censoring times were independently generated from a discrete uniform distribution over $\{1,\dots,T_{\max}\}$.
Unlike the continuous-time scenarios, ties arise intrinsically in the discrete-time scenarios. 
To examine additional temporal aggregation, we considered observation intervals of 1, 7, 14, and 28 time units. 
When the interval width is 1, no additional rounding is applied, but tied event times can still occur.

We compared the PL-Cox and GPL-Cox models with several existing approaches, including the Breslow, Efron, and Exact methods for handling ties in the frequentist Cox model, along with the methods of \citet{renCoxPolyaGammaAlgorithmFlexible2025} (Ren) and \citet{tamanoEfficientGibbsSampling2025} (Tamano).
The implementation details for these methods are described in Appendix~E.1 in the Supplementary Material.
For the PL-Cox and GPL-Cox models, independent normal priors $\mathrm{N}(0,100)$ were assigned to the regression coefficients.
For PL-Cox, we fixed the approximation parameter at $\delta=100$.
The choice of this value, sensitivity to the intercept specification, and approximation accuracy are examined in Appendix D.1 in the Supplementary Material.
Posterior inference was carried out using Markov chain Monte Carlo with 3,000 iterations, and the first 1,000 iterations were discarded as burn-in.
For PL-Cox and GPL-Cox, all regression coefficients, including the GPL-Cox intercept, were initialized at zero. 
The competing Bayesian methods used the initialization procedures provided by their respective implementations.
Each scenario was replicated 10,000 times.
The performance of the competing methods was evaluated using the empirical bias (Bias), the Monte Carlo standard deviation (SD) of the point estimates, root mean squared error (RMSE), coverage probability (CP) of the nominal 95\% interval, and average interval width (AW).
Further details and simulation results for $\beta_4$ are provided in Appendix~E in the Supplementary Material. 
Complete results for all four regression coefficients are available in the accompanying Supplementary Excel file.

\subsection{Results}

Table~\ref{res_cont} shows the results for $\beta_4$ under continuous survival times generated from an exponential model.
Across all four scenarios, the PL-Cox method yielded RMSE and AW values comparable to those of the frequentist Cox methods, with CP slightly below the nominal level.
In the scenario with no rounding, the GPL-Cox method exhibited a shorter AW, resulting in a CP below the nominal level.
This may be partly because, as shown in Equation (\ref{like_GPL}), the numerator of the likelihood includes the product of $1-\theta_i$ for subjects who did not experience the event.
The explicit inclusion of these factors may affect posterior concentration, potentially leading to shorter credible intervals and contributing to the reduction in CP.
This pattern attenuated as the frequency of ties increased, and in the scenario where rounding unit was 0.1, its performance was comparable to that of the other methods.
Across all four scenarios, the Ren method yielded a slightly higher RMSE, and the Tamano method consistently produced results comparable to those obtained from the three frequentist approaches.

\begin{table}[htbp]
\centering
\caption{Simulation results for the regression coefficient $\beta_4$ under
continuous survival times generated from an exponential model ($n=300$).
``None'' indicates no rounding.}
\label{res_cont}
\renewcommand{\arraystretch}{0.75}
\setlength{\tabcolsep}{4pt}
\begin{tabular}{ccrccccc}
\hline
Sce & Method & \multicolumn{1}{c}{Bias} & SD & RMSE & CP & AW & Time (s) \\
\hline
None & Breslow & 0.005 & 0.075 & 0.075 & 94.52 & 0.286 & 0.002 \\
  & Efron & 0.005 & 0.075 & 0.075 & 94.52 & 0.286 & 0.002 \\
  & Exact & 0.005 & 0.075 & 0.075 & 94.52 & 0.286 & 0.003 \\
  & PL-Cox & 0.007 & 0.076 & 0.076 & 93.80 & 0.286 & 0.414 \\
  & GPL-Cox & $-$0.011 & 0.075 & 0.076 & 88.37 & 0.252 & 0.609 \\
  & Ren & 0.012 & 0.077 & 0.078 & 92.76 & 0.283 & 8.59 \\
  & Tamano & 0.005 & 0.075 & 0.075 & 94.46 & 0.286 & 11.04 \\
\hline
0.01 & Breslow & 0.005 & 0.075 & 0.075 & 94.53 & 0.286 & 0.002 \\
  & Efron & 0.005 & 0.075 & 0.075 & 94.51 & 0.286 & 0.002 \\
  & Exact & 0.005 & 0.075 & 0.075 & 94.52 & 0.286 & 0.003 \\
  & PL-Cox & 0.007 & 0.076 & 0.076 & 93.82 & 0.286 & 0.410 \\
  & GPL-Cox & 0.002 & 0.077 & 0.077 & 90.81 & 0.270 & 0.671 \\
  & Ren & 0.012 & 0.077 & 0.078 & 92.78 & 0.283 & 8.66 \\
  & Tamano & 0.005 & 0.075 & 0.075 & 94.41 & 0.286 & 11.08 \\
\hline
0.1 & Breslow & 0.003 & 0.074 & 0.074 & 94.67 & 0.286 & 0.002 \\
  & Efron & 0.005 & 0.075 & 0.075 & 94.48 & 0.286 & 0.002 \\
  & Exact & 0.006 & 0.075 & 0.075 & 94.45 & 0.288 & 0.004 \\
  & PL-Cox & 0.005 & 0.075 & 0.075 & 93.86 & 0.285 & 0.379 \\
  & GPL-Cox & 0.006 & 0.075 & 0.075 & 93.77 & 0.284 & 1.41 \\
  & Ren & 0.013 & 0.077 & 0.078 & 92.67 & 0.283 & 8.65 \\
  & Tamano & 0.003 & 0.074 & 0.074 & 94.70 & 0.286 & 11.16 \\
\hline
0.5 & Breslow & $-$0.004 & 0.072 & 0.073 & 95.12 & 0.286 & 0.002 \\
  & Efron & 0.005 & 0.075 & 0.075 & 94.48 & 0.286 & 0.002 \\
  & Exact & 0.013 & 0.077 & 0.078 & 94.23 & 0.295 & 0.008 \\
  & PL-Cox & $-$0.002 & 0.073 & 0.073 & 94.49 & 0.285 & 0.343 \\
  & GPL-Cox & 0.006 & 0.075 & 0.075 & 94.67 & 0.290 & 2.22 \\
  & Ren & 0.015 & 0.078 & 0.079 & 92.19 & 0.283 & 8.75 \\
  & Tamano & $-$0.004 & 0.072 & 0.072 & 95.16 & 0.285 & 11.28 \\
\hline
\end{tabular}
\par\smallskip{\footnotesize Note: Bias, SD, RMSE, CP, and AW are based on 10,000 simulated datasets per scenario (successful fits only). Time is the mean elapsed time per dataset from the same 10,000-replicate experiment, including model fitting and posterior summaries; the individual data generation and file I/O are excluded.}
\end{table}

Table~\ref{res_discrete} shows the results for $\beta_4$ under discrete survival times generated from a logistic hazard model with a constant hazard.
When the coarsening unit was 1, the PL-Cox method exhibited an AW comparable to those of the competing methods, with a CP slightly below the nominal level.
As the coarsening unit size increased, the PL-Cox method exhibited a negative bias.
When the coarsening unit size was 28, the point estimation performance of the PL-Cox method was similar to that of the Breslow and Tamano methods.
Across all four scenarios, the GPL-Cox method yielded bias and RMSE values similar to those of the Exact method, consistent with its geometric ranking formulation under a constant hazard.

\begin{table}[htbp]
\centering
\caption{Simulation results for the regression coefficient $\beta_4$ under
discrete survival times generated from a logistic hazard model with constant
hazard ($n=300$).}
\label{res_discrete}
\renewcommand{\arraystretch}{0.75}
\setlength{\tabcolsep}{4pt}
\begin{tabular}{ccrccccc}
\hline
Sce & Method & \multicolumn{1}{c}{Bias} & SD & RMSE & CP & AW & Time (s) \\
\hline
1 & Breslow & 0.002 & 0.081 & 0.081 & 94.72 & 0.313 & 0.002 \\
  & Efron & 0.004 & 0.081 & 0.082 & 94.61 & 0.313 & 0.002 \\
  & Exact & 0.005 & 0.082 & 0.082 & 94.63 & 0.314 & 0.004 \\
  & PL-Cox & 0.004 & 0.082 & 0.082 & 93.92 & 0.311 & 0.377 \\
  & GPL-Cox & 0.005 & 0.082 & 0.082 & 93.68 & 0.310 & 1.25 \\
  & Ren & 0.010 & 0.084 & 0.084 & 93.34 & 0.309 & 8.67 \\
  & Tamano & 0.002 & 0.081 & 0.081 & 94.59 & 0.312 & 9.38 \\
\hline
7 & Breslow & $-$0.005 & 0.080 & 0.080 & 94.87 & 0.315 & 0.002 \\
  & Efron & 0.004 & 0.082 & 0.082 & 94.42 & 0.315 & 0.002 \\
  & Exact & 0.012 & 0.085 & 0.086 & 94.47 & 0.325 & 0.007 \\
  & PL-Cox & $-$0.003 & 0.081 & 0.081 & 94.15 & 0.313 & 0.345 \\
  & GPL-Cox & 0.012 & 0.084 & 0.085 & 94.22 & 0.321 & 2.10 \\
  & Ren & 0.013 & 0.086 & 0.087 & 92.61 & 0.312 & 8.99 \\
  & Tamano & $-$0.005 & 0.080 & 0.080 & 94.82 & 0.314 & 9.30 \\
\hline
14 & Breslow & $-$0.013 & 0.078 & 0.079 & 95.11 & 0.317 & 0.002 \\
  & Efron & 0.003 & 0.083 & 0.083 & 94.62 & 0.318 & 0.002 \\
  & Exact & 0.020 & 0.088 & 0.090 & 94.14 & 0.337 & 0.010 \\
  & PL-Cox & $-$0.011 & 0.079 & 0.080 & 94.67 & 0.316 & 0.334 \\
  & GPL-Cox & 0.021 & 0.088 & 0.090 & 93.87 & 0.334 & 1.79 \\
  & Ren & 0.015 & 0.088 & 0.089 & 92.43 & 0.315 & 10.27 \\
  & Tamano & $-$0.013 & 0.078 & 0.079 & 95.02 & 0.316 & 9.25 \\
\hline
28 & Breslow & $-$0.028 & 0.075 & 0.080 & 95.28 & 0.322 & 0.002 \\
  & Efron & 0.002 & 0.084 & 0.084 & 94.78 & 0.324 & 0.002 \\
  & Exact & 0.037 & 0.095 & 0.102 & 93.55 & 0.364 & 0.015 \\
  & PL-Cox & $-$0.027 & 0.076 & 0.081 & 94.69 & 0.321 & 0.330 \\
  & GPL-Cox & 0.039 & 0.095 & 0.102 & 92.72 & 0.362 & 1.14 \\
  & Ren & 0.019 & 0.091 & 0.093 & 91.70 & 0.321 & 17.83 \\
  & Tamano & $-$0.028 & 0.075 & 0.080 & 95.31 & 0.322 & 9.05 \\
\hline
\end{tabular}
\par\smallskip{\footnotesize Note: Bias, SD, RMSE, CP, and AW are based on 10,000 simulated datasets per scenario (successful fits only). Time is the mean elapsed time per dataset from the same 10,000-replicate experiment, including model fitting and posterior summaries; the individual data generation and file I/O are excluded.}
\end{table}

The results for the other scenarios are provided in Appendix~E.2 in the Supplementary Material.
The PL-Cox method maintained stable performance across continuous survival time scenarios beyond the constant hazard setting.
In discrete survival time scenarios with non-constant hazards, the PL-Cox method exhibited bias, with behavior comparable to that of the Breslow approximation.
Under non-proportional hazards, systematic bias was observed in certain settings, particularly when early events dominated; this pattern was consistent with that in the Breslow approximation.
The GPL-Cox method sometimes led to increased bias and SD in continuous survival time scenarios beyond the constant hazard setting.
Even in discrete survival time scenarios, when the hazard was not constant, its behavior deviated from that of the Exact method.
Under non-proportional hazards, the bias tended to be larger when the number of ties was small.
Overall, these results do not suggest a consistent estimation advantage of GPL-Cox over PL-Cox, even for intrinsically discrete event times.

\section{Applications\label{application}}

Subsection \ref{Vital} examines computation and MCMC convergence in a large dataset with varying degrees of tied event times, while Subsection \ref{readmission} presents an extension of the methods to a shared frailty model.
Furthermore, Appendix~F.3 in the Supplementary Material presents an analysis of real data characterized by a high frequency of ties.

Following \citet{vehtariRankNormalizationFoldingLocalization}, we assess MCMC convergence using the rank-normalized split and folded $\widehat{R}$ diagnostic.
The effective sample size (ESS) quantifies the equivalent number of independent draws.
Bulk ESS is computed from rank-normalized split chains, whereas tail ESS is the minimum of the ESS for the 5th and 95th percentiles.
We also report ESS per second (ESS/sec) as a measure of computational efficiency.

\subsection{Vital Data\label{Vital}}
We analyzed data from the VITamin D and OmegA-3 TriaL (VITAL), a large randomized primary-prevention trial of vitamin $\text{D}_3$ and marine n-3 fatty acids \citep{mansonVitaminSupplementsPrevention2019, macfarlaneEffectsVitaminMarine2020}.
Individual-level data were obtained from Project Data Sphere (\url{https://www.projectdatasphere.org/}).
The outcome was time to incident cancer. After complete-case selection for the treatment assignments and baseline covariates, the analysis comprised 23,972 participants and 1,507 observed events.
To examine the effect of tied event times, we analyzed the original follow-up times and versions rounded to 0.05, 0.10, and 0.25 years.
The numbers of distinct event times were 1,012, 120, 60, and 25, respectively; the corresponding maximum tie-block sizes were 6, 26, 44, and 88.

We compare the computational performance of PL-Cox, GPL-Cox, and Tamano across the four coarsening conditions using elapsed time and the minimum Bulk ESS/sec.
Table~\ref{vital_computation} summarizes the results from four independent chains for each method.
PL-Cox and GPL-Cox generally require substantially less computation time than Tamano.
Because the sampling budgets and computational procedures differ across methods, the reported ESS/sec values should be interpreted as descriptive measures of computational efficiency.

\begin{table}[htbp]
\centering
\caption{Computational performance in the VITAL analysis
based on four-chain posterior sampling.
\label{vital_computation}}
\setlength{\tabcolsep}{4pt}
\begin{tabular}{lrrrr}
\hline
Method & Original & 0.05 & 0.10 & 0.25 \\
\hline
\multicolumn{5}{l}{\textit{Elapsed time (min)}} \\
PL-Cox & 55.0 & 99.9 & 99.7 & 53.9 \\
GPL-Cox & 69.4 & 79.5 & 139.5 & 374.5 \\
Tamano & 3683.6 & 3669.0 & 3677.4 & 3705.4 \\
\hline
\multicolumn{5}{l}{\textit{Minimum Bulk ESS/sec}} \\
PL-Cox & 0.239 & 0.245 & 0.248 & 0.215 \\
GPL-Cox & $-$ & 0.137 & 0.156 & 0.104 \\
Tamano & 0.068 & 0.068 & 0.068 & 0.067 \\
\hline
\end{tabular}

\smallskip
\begin{minipage}{0.97\linewidth}
\footnotesize
\textit{Notes:}
Elapsed times are the sum of computation times over four serial chains, including burn-in.
Sampling budgets differ across methods and coarsening conditions.
The minimum Bulk ESS is taken over the parameters included in the convergence diagnostics.
The Bulk ESS/sec value for GPL-Cox under the original-time condition is omitted because the convergence criteria were not satisfied.
For Tamano, the reported diagnostics characterize the calibrated samples and do not by themselves establish the validity of the posterior approximation.
\end{minipage}
\end{table}

For posterior inference, the four-chain runs were assessed using the criteria $\widehat R\leq1.01$, Bulk ESS $\geq400$, and Tail ESS $\geq400$.
PL-Cox used 10,000 burn-in iterations and retained 80,000 iterations per chain for the original and 0.25-year data, or 160,000 iterations for the 0.05- and 0.10-year data.
All four PL-Cox analyses satisfied the convergence criteria.
GPL-Cox retained 80,000 iterations per chain after 20,000 burn-in iterations and satisfied the criteria for the three rounded-time conditions, but not for the original follow-up times ($\max\widehat R=1.1387$; minimum Bulk ESS $=19.7$).
An extended GPL-Cox run retaining 200,000 iterations per chain also failed the criteria.
Figure~\ref{vital_forest} presents the exponentiated
coefficient summaries across the four coarsening
conditions; GPL-Cox results for the original
follow-up times are shown for descriptive comparison
only because convergence was not established.
Further convergence diagnostics are provided in the Supplementary Material.

\begin{figure}[htbp]
\centering
\includegraphics[width=0.98\linewidth,height=0.77\textheight,keepaspectratio]{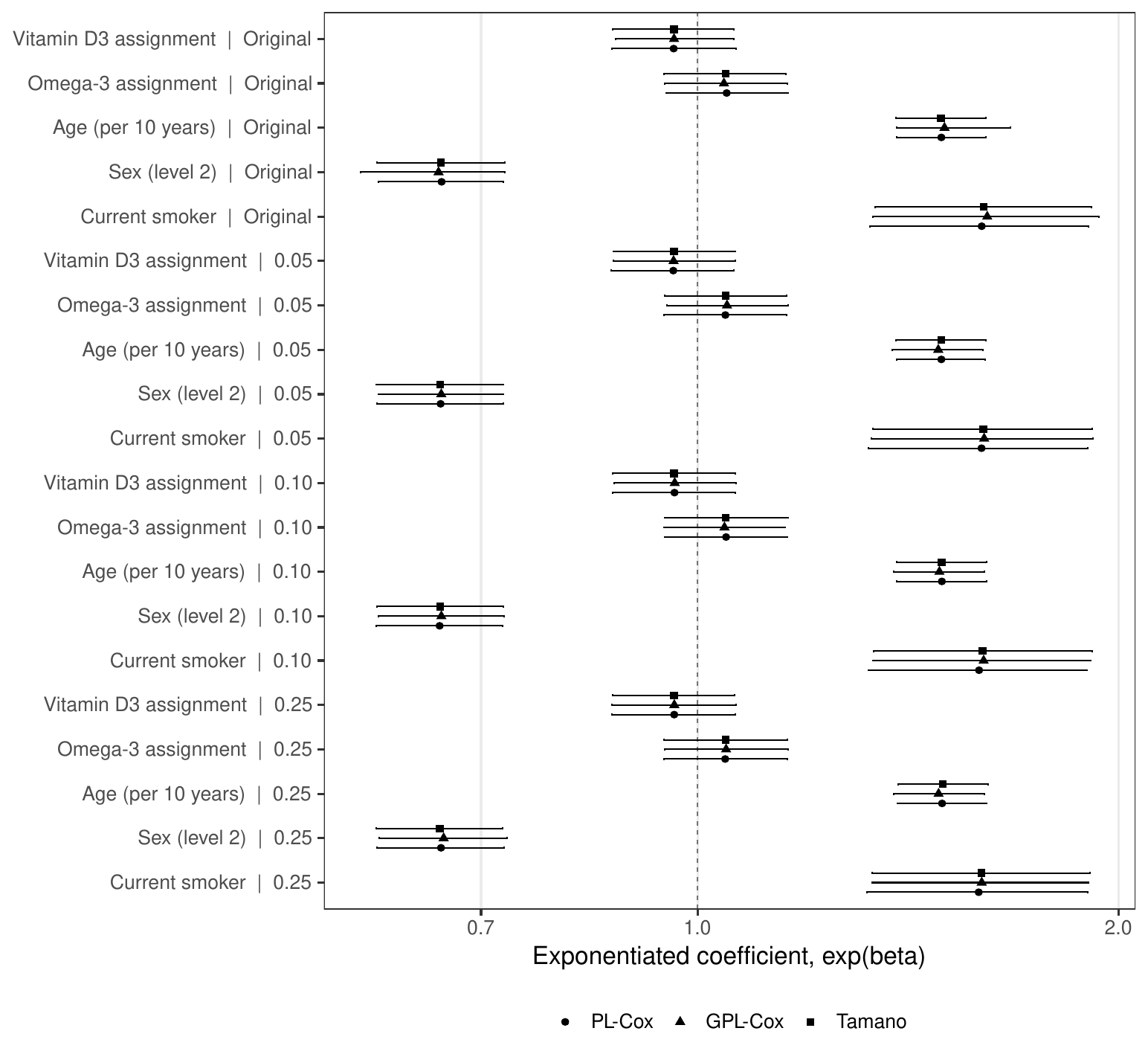}
\caption{Exponentiated regression-coefficient summaries (posterior point estimates and 95\% credible intervals) for selected covariates in VITAL. PL-Cox (circles) and Tamano (squares) are on the conventional Cox hazard-ratio scale; GPL-Cox (triangles) shows $\exp(\beta)$ from the latent ranking model, which is not a conventional Cox hazard ratio. The same filled symbols are used for all methods and coarsening conditions. GPL-Cox summaries for original follow-up times are displayed for visual completeness only; their nominal intervals are not considered reliable for inference because MCMC convergence was not established. PL-Cox satisfied the prespecified convergence criteria in all four conditions.\label{vital_forest}}
\end{figure}

\subsection{Readmissions Data\label{readmission}}

We analyzed the recurrent-event colorectal cancer readmissions data in the \texttt{frailtypack} package \citep{rondeauFrailtypackPackageAnalysis2012}. The analysis comprised 403 subjects, 861 observation rows, and 458 observed events. We considered chemotherapy treatment, sex, Dukes stage, and the Charlson comorbidity index, with subject-specific shared log-normal frailties $u_i\sim\mathrm{N}(0,\sigma_u^2)$.

The PL-Cox and GPL-Cox methods were compared with the method of \citet{renCoxPolyaGammaAlgorithmFlexible2025} (Ren) and a Hamiltonian Monte Carlo (HMC) implementation of the Cox partial likelihood with the Breslow approximation (Stan-Breslow).
All four methods were fitted using four independent chains.
The final posterior summaries were calculated from runs retaining 20,000, 80,000, 40,000, and 5,000 iterations per chain for PL-Cox, GPL-Cox, Ren, and Stan-Breslow, respectively.
All four methods satisfied the prespecified convergence criteria of $\widehat R\leq1.01$, Bulk ESS $\geq400$, and Tail ESS $\geq400$ for the monitored regression coefficients and frailty variance.

Table~\ref{tab:readmission_estimates} presents the exponentiated regression-coefficient summaries, frailty-variance estimates, and computational performance.
The fixed-effect estimates obtained from PL-Cox, Ren, and Stan-Breslow are similar, whereas GPL-Cox produces larger estimates for some coefficients and for the frailty variance.

PL-Cox achieves a minimum Bulk ESS/sec of 12.50, compared with 6.62 for Stan-Breslow, corresponding to an approximately 1.89-fold higher minimum Bulk ESS/sec under the implemented sampling configurations.
The corresponding elapsed computation times are 38.7 and 1048.0 seconds, respectively.
Although the sampling budgets and posterior targets differ across methods, these results demonstrate that PL-Cox can achieve favorable computational efficiency relative to HMC in this shared-frailty setting.
Additional convergence diagnostics are provided in Appendix~F.2 in the Supplementary Material.

\begin{table}[htbp]
\centering
\caption{Exponentiated regression-coefficient summaries,
shared-frailty variance estimates, and computational
performance for the readmissions data.
\label{tab:readmission_estimates}}
\setlength{\tabcolsep}{3pt}
\begin{tabular}{lcccc}
\hline
 & PL-Cox & GPL-Cox & Ren & Stan-Breslow \\
\hline
Chemo (treated)
& 0.82 [0.61, 1.09]
& 0.83 [0.61, 1.14]
& 0.83 [0.62, 1.09]
& 0.82 [0.61, 1.08] \\
Sex (female)
& 0.62 [0.46, 0.82]
& 0.58 [0.43, 0.79]
& 0.63 [0.48, 0.83]
& 0.62 [0.47, 0.81] \\
Dukes C
& 1.35 [0.97, 1.90]
& 1.38 [0.97, 1.98]
& 1.33 [0.97, 1.83]
& 1.35 [0.98, 1.85] \\
Dukes D
& 3.03 [2.04, 4.54]
& 3.66 [2.39, 5.63]
& 2.92 [2.00, 4.32]
& 2.98 [2.03, 4.40] \\
Charlson (1--2)
& 1.49 [0.87, 2.47]
& 1.55 [0.87, 2.68]
& 1.47 [0.88, 2.40]
& 1.49 [0.89, 2.46] \\
Charlson ($\geq 3$)
& 1.39 [1.05, 1.82]
& 1.53 [1.15, 2.03]
& 1.37 [1.05, 1.77]
& 1.39 [1.07, 1.81] \\
Frailty variance
& 0.62 [0.37, 0.93]
& 0.87 [0.58, 1.23]
& 0.52 [0.30, 0.81]
& 0.56 [0.34, 0.83] \\
\hline
Elapsed time (sec)
& 38.7 & 299.0 & 63142.9 & 1048.0 \\
Min Bulk ESS
& 483 & 1238 & 869 & 6936 \\
Min Bulk ESS/sec
& 12.50 & 4.14 & 0.014 & 6.62 \\
\hline
\end{tabular}

\smallskip
\begin{minipage}{0.98\linewidth}
\footnotesize
\textit{Notes:}
Fixed-effect entries are exponentiated posterior
means and endpoints of 95\% equal-tail credible
intervals; they are not posterior means of
the exponentiated coefficients.
PL-Cox, Ren, and Stan-Breslow have the
conventional Cox hazard-ratio interpretation,
whereas GPL-Cox reports exponentiated
latent-ranking coefficients.
Frailty variance is reported on its original scale.
Elapsed times include four serial chains and burn-in.
Minimum Bulk ESS/sec is defined as the minimum
Bulk ESS across the monitored parameters divided
by the total elapsed time.
Sampling budgets differ across methods.
\end{minipage}
\end{table}

\section{Discussion\label{discussion}}
In this study, we developed two P\'{o}lya--Gamma-augmented sampling algorithms based on PL and GPL rank-ordered likelihood formulations.
We have developed an R package, \texttt{BayesPLCox}, for implementing the PL-Cox and GPL-Cox methods, which is publicly available on GitHub at \url{https://github.com/tom-ohigashi/BayesPLCox}.
Across many scenarios in the simulation study, the PL-Cox model demonstrated stable performance, suggesting that it is a reasonable first choice.
The GPL-Cox model directly incorporates information about tie sizes, as illustrated by the RHC analysis in the Supplementary Material. However, in VITAL, the GPL-Cox sampler did not attain satisfactory convergence for the original, comparatively sparsely tied event times, even in an extended run; it met the diagnostic criteria after coarsening, while its computation time increased with the degree of coarsening. Thus, neither general computational superiority nor reliable posterior inference in all tie regimes can be claimed from this analysis.
An advantage of the GPL-Cox model is that it allows posterior sampling to be performed with ease, even in cases where the frequentist Exact method fails to provide estimates.
A limitation of the GPL-Cox model is that its performance may deteriorate in settings with sparse ties, where the additional information incorporated through the likelihood may lead to over-concentration of the posterior.
In supplementary simulation comparisons with HMC implemented in Stan using the Cox partial likelihood with the Breslow approximation, PL-Cox required less wall-clock time in shared-frailty settings, whereas HMC generally achieved higher ESS per second (see Appendix D.2 in the Supplementary Material). 
Thus, the main contribution of the proposed methods lies in their flexible rank-ordered likelihood framework and computational tractability rather than uniformly superior computational efficiency.

Additional simulations showed deterioration in performance under more complex hazard structures, particularly undercoverage for GPL-Cox (see Appendix~E in the Supplementary Material). 
Improving robustness under such settings remains an important direction for future research.

%---------------------------------------%
%          Acknowledgement              %
%---------------------------------------%
\section*{Acknowledgement}
This work was partially supported by JSPS KAKENHI Grant Numbers 24K20739, 24K21420, 25H00546, and 25K21166.
RHC data obtained from \url{http://hbiostat.org/data} courtesy of the Vanderbilt University Department of Biostatistics.
This study was carried out under the Cooperative Use Registration (2025-ISMCRP-0001).

%---------------------------------------%
%       Data Availability Statement     %
%---------------------------------------%
\section*{Data Availability Statement}
The R package \texttt{BayesPLCox}, which implements the proposed methods, is publicly available on GitHub at \url{https://github.com/tom-ohigashi/BayesPLCox}.

%---------------------------------------------%
%                Reference                    %
%---------------------------------------------%

\newpage
% \pdfoutput=1
% \documentclass[11pt]{article}
% \usepackage[utf8]{inputenc}
% \usepackage{amsmath,amssymb,bbm}
% \usepackage{amsthm}
% \usepackage{authblk}
% \usepackage{cprotect}

% \usepackage{stmaryrd}
% \usepackage{colortbl,xcolor}

% \usepackage{subcaption}

% \usepackage{natbib}
% \usepackage{booktabs, threeparttable}

% \usepackage{geometry}
% \usepackage{lipsum}

% \usepackage{color}
% \usepackage{soul}
% \usepackage{url}
% % One column, one-inch margins
% \geometry{letterpaper, margin=1in}

% % Double-spaced throughout
% \usepackage{setspace}
% \doublespacing

% \usepackage[figuresright]{rotating}
% \usepackage{ dsfont }

% \usepackage{here}

% \newtheorem{theorem}{Theorem}

%%%%%%%%%% Merge with supplemental materials %%%%%%%%%%
%%%%%%%%%% Prefix a "S" to all equations, figures, tables and reset the counter %%%%%%%%%%
\setcounter{equation}{0}
\setcounter{figure}{0}
\setcounter{table}{0}
\setcounter{page}{1}
\setcounter{section}{0}
\makeatletter

\renewcommand{\thesection}{Appendix \Alph{section}}
\renewcommand{\thesubsection}{\Alph{section}.\arabic{subsection}}
\renewcommand{\theequation}{\Alph{section}.\arabic{equation}}

\renewcommand{\thetable}{\arabic{table}}
\renewcommand{\tablename}{Table}

\renewcommand{\thefigure}{\arabic{figure}}
\renewcommand{\figurename}{Figure}

\renewcommand{\bibnumfmt}[1]{[S#1]}
\renewcommand{\citenumfont}[1]{S#1}
%%%%%%%%%% Prefix a "S" to all equations, figures, tables and reset the counter %%%%%%%%%%

% \newtheorem{defi}{Definition}
% \newtheorem{theo}{Theorem}
% \newtheorem{assu}{Assumption}
% \newtheorem{esti}{Estimation}
% \newtheorem{lemm}{Lemma}
% \newtheorem{prop}{Proposition}
% \newtheorem{note}{Note}

% \allowdisplaybreaks[4]

% \newcommand{\indep}{\mathop{\perp\!\!\!\perp}}
% \newcommand{\notindep}{\mathop{\not\!\perp\!\!\!\perp}}
% \newcommand{\bld}{\boldsymbol}

% \newcommand{\argmin}{\mathop{\rm arg~min}\limits}
% \newcommand{\argmax}{\mathop{\rm arg~max}\limits}

% \usepackage[top=1in,bottom=1in,left=1in,right=1in]{geometry}

% \begin{document}
% \title{Supplementary Materials for a Full-Likelihood Framework for Bayesian Inference in the Cox Model via Rank-Ordered Data Modeling by Tomohiro Ohigashi, Shunichiro Orihara and Shonosuke Sugasawa.}

% \author{\empty}
% \date{\empty}
% \date{\empty}
% \maketitle
\clearpage
\begin{center}
{\Large \bfseries Supplementary Materials for Efficient Bayesian Inference in the Cox Model via Rank-Ordered Likelihood by Tomohiro Ohigashi, Shunichiro Orihara and Shonosuke Sugasawa.\par}
\end{center}
\vspace{1em}

% \author{\empty}
% \date{\empty}
% \date{\empty}
% \maketitle

\section{Step-by-step sampling procedures for PL-Cox}

For posterior computation under the PL-Cox model, we use the latent-variable representation introduced in the main text.
Let \(R\) be the number of distinct event times, let \(R_r\) denote the risk set at time \(t_{(r)}\), and let \(E_r\) denote the corresponding event set, with \(d_r = |E_r|\).
We define $\lambda_i = \exp(\eta_i)$, where $\eta_i = \bld{x}_i^\top \bld{\beta}$, and the design vector \(\bld{x}_i\) may include an intercept term.
We assume the prior $\bld{\beta} \sim \mathrm{N}(\bld{b}_0, V_0)$.
By default, the \texttt{BayesPLCox} package initializes the regression coefficients at zero for PL-Cox.

Using the gamma integral identity, the augmented likelihood can be written as
\[
L_{\mathrm{PL}}^{\ast}(\bld{\beta},\bld{Z})
\propto
\prod_{r=1}^{R}
\left[
z_r^{d_r-1}
\prod_{i \in E_r} \lambda_i
\exp\!\left\{
-z_r \sum_{j \in R_r} \lambda_j
\right\}
\right].
\]
Writing $A_r=\sum_{j\in R_r}\lambda_j$, integrating out $Z_r$ yields
\[
\int_0^\infty
\frac{z_r^{d_r-1}}{\Gamma(d_r)}
\left(\prod_{i\in E_r}\lambda_i\right)
\exp(-A_r z_r)\,dz_r
=
\frac{\prod_{i\in E_r}\lambda_i}{A_r^{d_r}}.
\]
Thus, marginalizing over all $Z_r$ recovers the original PL-Cox likelihood (3) exactly, before introducing the negative-binomial approximation with finite $\delta$.

Equivalently, defining
\[
c_i = \sum_{r=1}^{R} \mathds{1}(i \in E_r),
\qquad
\zeta_i = \sum_{r=1}^{R} \mathds{1}(i \in R_r)\, z_r,
\]
we obtain
\[
L_{\mathrm{PL}}^{\ast}(\bld{\beta},\bld{Z})
\propto
\prod_{i=1}^{n}
\lambda_i^{c_i}\exp(-\zeta_i \lambda_i)
\prod_{r=1}^{R} z_r^{d_r-1}.
\]

To obtain conditionally Gaussian updates for \(\bld{\beta}\), we employ the Poisson--Gamma construction described in the main text.
Let \(\delta>0\) be a fixed approximation parameter and define
\[
\psi_i = \bld{x}_i^\top \bld{\beta} + \log(\zeta_i/\delta),
\qquad
\kappa_i = \frac{c_i-\delta}{2},
\qquad
\bld{o} = \bigl(\log(\zeta_1/\delta),\ldots,\log(\zeta_n/\delta)\bigr)^\top .
\]

The sampling updates are described below.
The update for $Z_r$ is exact, whereas the updates for $\omega_i$ and $\boldsymbol{\beta}$ are derived from the negative-binomial approximation.

\begin{itemize}
\item \textbf{(Sampling of \(Z_r\))} For \(r=1,\ldots,R\),
\[
Z_r \mid \bld{\beta}, \bld{D}
\sim
\mathrm{Gamma}\!\left(d_r,\sum_{j \in R_r}\lambda_j\right).
\]

\item \textbf{(Sampling of \(\omega_i\))} For \(i=1,\ldots,n\),
\[
\omega_i \mid \bld{\beta}, \bld{Z}, \bld{D}
\sim
\mathrm{PG}(c_i+\delta,\psi_i).
\]

\item \textbf{(Sampling of \(\bld{\beta}\))} Let
\[
\Omega = \mathrm{diag}(\omega_1,\ldots,\omega_n),
\qquad
B = X^\top \Omega X + V_0^{-1},
\]
\[
\bld{g}
=
X^\top(\bld{\kappa}-\Omega \bld{o}) + V_0^{-1}\bld{b}_0,
\qquad
\bld{\kappa}=(\kappa_1,\ldots,\kappa_n)^\top,
\]
where \(X\) is the design matrix with \(i\)th row \(\bld{x}_i^\top\).
Then
\[
\bld{\beta} \mid \bld{\omega}, \bld{Z}, \bld{D}
\sim
\mathrm{N}(B^{-1}\bld{g},\,B^{-1}).
\]
\end{itemize}

\newpage

\section{Step-by-step sampling procedures for GPL-Cox}

For posterior computation under the GPL-Cox model, we use the geometric latent-variable representation.
Let \(R\) be the number of distinct event times, let \(R_r\) denote the risk set at time \(t_{(r)}\), and let \(E_r\) denote the corresponding event set.
We write $\theta_i = \operatorname{expit}(\eta_i)$ where $\eta_i = \bld{x}_i^\top \bld{\beta}$, and the design vector \(\bld{x}_i\) may include an intercept term.
We assume the prior $\bld{\beta} \sim \mathrm{N}(\bld{b}_0, V_0)$.
By default, the \texttt{BayesPLCox} package initializes the regression coefficients, including the intercept, at zero for GPL-Cox.

For each distinct event time \(t_{(r)}\), introduce a latent geometric variable
\[
Z_r \mid \bld{\beta}, \bld{D}
\sim
\mathrm{Geom}\!\left(
1-\prod_{j \in R_r}(1-\theta_j)
\right),
\qquad r=1,\ldots,R,
\]
where we use the convention that \(\mathrm{Geom}(p)\) has support \(\{1,2,\ldots\}\).

Define
\[
c_i = \sum_{r=1}^{R} \mathds{1}(i \in E_r),
\qquad
\zeta_i = \sum_{r=1}^{R} \mathds{1}(i \in R_r)\, Z_r,
\qquad
\kappa_i = c_i - \frac{\zeta_i}{2}.
\]
Under this augmentation, the complete-data likelihood has kernel
\[
L_{\mathrm{GPL}}^{\ast}(\bld{\beta},\bld{Z})
\propto
\prod_{i=1}^{n}
\theta_i^{c_i}(1-\theta_i)^{\zeta_i-c_i}.
\]
Let $a_r=\prod_{j\in R_r}(1-\theta_j)$ and $b_r=\prod_{i\in E_r}\theta_i \prod_{j\in R_r\setminus E_r}(1-\theta_j)$.
Marginalizing over $Z_r$ gives
\[
\sum_{z_r=1}^{\infty}
b_r a_r^{z_r-1}
=
\frac{b_r}{1-a_r}
=
L_{\mathrm{GPL},r}(\boldsymbol{\theta}).
\]
Thus, summing over all $Z_r$ recovers the GPL-Cox likelihood (4) exactly.
Using \(\theta_i=\exp(\eta_i)/(1+\exp(\eta_i))\), this can be rewritten as
\[
L_{\mathrm{GPL}}^{\ast}(\bld{\beta},\bld{Z})
\propto
\prod_{i=1}^{n}
2^{-\zeta_i}
\exp(\kappa_i \eta_i)
\int_0^\infty
\exp\!\left(-\frac{\omega_i \eta_i^2}{2}\right)
p_{\mathrm{PG}}(\omega_i \mid \zeta_i,0)\,d\omega_i,
\]
which yields conditionally Gaussian updates for \(\bld{\beta}\).

The full conditional distributions are described as follows:

\begin{itemize}
\item \textbf{(Sampling of \(Z_r\))} For \(r=1,\ldots,R\),
\[
Z_r \mid \bld{\beta}, \bld{D}
\sim
\mathrm{Geom}\!\left(
1-\prod_{j \in R_r}(1-\theta_j)
\right).
\]

\item \textbf{(Sampling of \(\omega_i\))} For \(i=1,\ldots,n\),
\[
\omega_i \mid \bld{\beta}, \bld{Z}, \bld{D}
\sim
\mathrm{PG}(\zeta_i,\eta_i).
\]
If \(\zeta_i=0\), then \(\omega_i\) is degenerate at \(0\).

\item \textbf{(Sampling of \(\bld{\beta}\))} Let
\[
\Omega = \mathrm{diag}(\omega_1,\ldots,\omega_n),
\qquad
B = X^\top \Omega X + V_0^{-1},
\]
\[
\bld{b}
=
X^\top \bld{\kappa} + V_0^{-1}\bld{b}_0,
\qquad
\bld{\kappa}=(\kappa_1,\ldots,\kappa_n)^\top,
\]
where \(X\) is the design matrix with \(i\)th row \(\bld{x}_i^\top\).
Then
\[
\bld{\beta} \mid \bld{\omega}, \bld{Z}, \bld{D}
\sim
\mathrm{N}(B^{-1}\bld{b},\,B^{-1}).
\]
\end{itemize}

\newpage

\section{Extensions of the Gibbs sampling framework\label{app:gibbs_extensions}}

The PL-Cox and GPL-Cox samplers in Appendices A and B are based on different latent-variable representations, but both yield a conditionally Gaussian kernel for the linear predictor after P\'{o}lya--Gamma augmentation.
This common structure makes it possible to incorporate several hierarchical specifications and non-Gaussian priors while retaining closed-form conditional updates.
We illustrate this point for a shared log-normal frailty model and for Student-$t$ and horseshoe priors on the regression coefficients.

For PL-Cox, the Gaussian updates derived below are based on the finite-$\delta$ approximation and should not be interpreted as exact full conditional updates under the original PL-Cox posterior. 

\subsection{A common conditionally Gaussian representation\label{app:common_gaussian}}

Let $\bld{\eta}=(\eta_1,\ldots,\eta_n)^\top$ denote the vector of linear predictors and let $\Omega=\mathrm{diag}(\omega_1,\ldots,\omega_n)$.
Conditional on the latent variables, the P\'{o}lya--Gamma part of both samplers can be written in the common form
\[
\prod_{i=1}^n
\exp\left\{
\kappa_i(\eta_i+o_i)
-\frac{\omega_i}{2}(\eta_i+o_i)^2
\right\},
\]
where, for PL-Cox,
\[
o_i=\log(\zeta_i/\delta),
\qquad
\kappa_i=\frac{c_i-\delta}{2},
\]
and, for GPL-Cox,
\[
o_i=0,
\qquad
\kappa_i=c_i-\frac{\zeta_i}{2}.
\]
Define
\[
\bld{r}=\bld{\kappa}-\Omega\bld{o},
\qquad
\bld{\kappa}=(\kappa_1,\ldots,\kappa_n)^\top,
\qquad
\bld{o}=(o_1,\ldots,o_n)^\top.
\]
Ignoring terms that do not depend on $\bld{\eta}$, the common kernel is then
\[
\exp\left\{
-\frac{1}{2}\bld{\eta}^\top\Omega\bld{\eta}
+\bld{\eta}^\top\bld{r}
\right\}.
\]
For PL-Cox, this expression corresponds to the negative-binomial approximation described in Appendix A, whereas for GPL-Cox it follows directly from the geometric augmentation in Appendix B. Hence, the extensions below do not introduce any additional approximation beyond that already used in the PL-Cox sampler.

When $\bld{\eta}=X\bld{\beta}$ and $\bld{\beta}\sim\mathrm{N}(\bld{b}_0,V_0)$, this representation immediately gives the Gaussian updates reported in Appendices A and B. The same representation also permits the following extensions.

\subsection{Shared log-normal frailty\label{app:frailty_extension}}

Suppose that subject $i$ belongs to cluster $g(i)\in\{1,\ldots,G\}$ and define
\[
\eta_i=\bld{x}_i^\top\bld{\beta}+u_{g(i)},
\qquad
u_g\mid\sigma_u^2\sim\mathrm{N}(0,\sigma_u^2),
\quad g=1,\ldots,G.
\]
The multiplicative frailty is therefore $\exp(u_g)$, which is log-normally distributed. Let $H$ be the $n\times G$ cluster-incidence matrix with $H_{ig}=\mathds{1}\{g(i)=g\}$, and write
\[
M=(X,H),
\qquad
\bld{\alpha}=
\begin{pmatrix}
\bld{\beta}\\
\bld{u}
\end{pmatrix},
\qquad
\bld{u}=(u_1,\ldots,u_G)^\top.
\]
Assume
\[
\bld{\beta}\sim\mathrm{N}(\bld{b}_0,V_0),
\qquad
\sigma_u^2\sim\mathrm{IG}(a_u,b_u),
\]
where $\mathrm{IG}(a,b)$ denotes the inverse-gamma distribution with density proportional to
$x^{-a-1}\exp(-b/x)$ for $x>0$. Conditional on $\sigma_u^2$, define
\[
P_0=
\begin{pmatrix}
V_0^{-1} & 0\\
0 & \sigma_u^{-2}I_G
\end{pmatrix},
\qquad
\bld{m}_0=
\begin{pmatrix}
\bld{b}_0\\
\bld{0}
\end{pmatrix}.
\]
Since $\bld{\eta}=M\bld{\alpha}$, the full conditional distribution of the regression coefficients and log-frailties is
\[
\bld{\alpha}\mid\bld{\omega},\bld{Z},\sigma_u^2,\bld{D}
\sim
\mathrm{N}(B_\alpha^{-1}\bld{g}_\alpha,B_\alpha^{-1}),
\]
where
\[
B_\alpha=M^\top\Omega M+P_0,
\qquad
\bld{g}_\alpha=M^\top\bld{r}+P_0\bld{m}_0.
\]
To implement this Gaussian update efficiently, partition
the precision matrix and linear term as
\[
B_\alpha=
\begin{pmatrix}
A & C\\
C^\top & D_u
\end{pmatrix},
\qquad
\bld{g}_\alpha=
\begin{pmatrix}
\bld{g}_\beta\\
\bld{g}_u
\end{pmatrix},
\]
where
\[
\begin{aligned}
A &= X^\top\Omega X+V_0^{-1},\\
C &= X^\top\Omega H,\\
D_u &= H^\top\Omega H+\sigma_u^{-2}I_G,\\
\bld{g}_\beta &= X^\top\bld{r}+V_0^{-1}\bld{b}_0,\\
\bld{g}_u &= H^\top\bld{r}.
\end{aligned}
\]
Because each subject belongs to a single cluster, $D_u$ is diagonal.
Defining the Schur complement
\[
S=A-CD_u^{-1}C^\top,
\]
we can first sample the regression coefficients from
\[
\bld{\beta}\mid\bld{\omega},\bld{Z},
\sigma_u^2,\bld{D}
\sim
\mathrm{N}\left(
S^{-1}(\bld{g}_\beta-CD_u^{-1}\bld{g}_u),
S^{-1}
\right),
\]
and then sample the log-frailties conditionally from
\[
\bld{u}\mid\bld{\beta},\bld{\omega},\bld{Z},
\sigma_u^2,\bld{D}
\sim
\mathrm{N}\left(
D_u^{-1}(\bld{g}_u-C^\top\bld{\beta}),
D_u^{-1}
\right).
\]
This procedure generates draws from the same joint Gaussian distribution without factorizing the full $(p+G)$-dimensional precision matrix. 
Instead, it requires factorization of only the $p\times p$ Schur complement, followed by conditionally independent Gaussian updates for the $G$ log-frailties.
The frailty variance also has a closed-form update,
\[
\sigma_u^2\mid\bld{u},\cdots
\sim
\mathrm{IG}\left(
 a_u+\frac{G}{2},
 b_u+\frac{1}{2}\bld{u}^\top\bld{u}
\right).
\]
Thus, the shared log-normal frailty extension retains closed-form Gaussian and inverse-gamma updates, subject to the finite-$\delta$ approximation for PL-Cox.

\subsection{Student-$t$ priors for regression coefficients\label{app:student_t_extension}}

A Student-$t$ prior can be written as a Gaussian scale mixture. Specifically, for $j=1,\ldots,p$, let
\[
\beta_j\mid\phi_j
\sim
\mathrm{N}\left(b_{0j},\frac{s_j^2}{\phi_j}\right),
\qquad
\phi_j\sim
\mathrm{Gamma}\left(\frac{\nu_j}{2},\frac{\nu_j}{2}\right),
\]
where the Gamma distribution is parameterized by shape and rate. Marginally,
\[
\beta_j\sim t_{\nu_j}(b_{0j},s_j).
\]
Let
\[
D_t=\mathrm{diag}\left(
\frac{\phi_1}{s_1^2},\ldots,
\frac{\phi_p}{s_p^2}
\right).
\]
Conditional on $\bld{\phi}=(\phi_1,\ldots,\phi_p)^\top$, the full conditional distribution of $\bld{\beta}$ remains Gaussian,
\[
\bld{\beta}\mid\bld{\omega},\bld{Z},\bld{\phi},\bld{D}
\sim
\mathrm{N}(B_t^{-1}\bld{g}_t,B_t^{-1}),
\]
where
\[
B_t=X^\top\Omega X+D_t,
\qquad
\bld{g}_t=X^\top\bld{r}+D_t\bld{b}_0.
\]
The mixing precisions have independent full conditional distributions
\[
\phi_j\mid\beta_j,\cdots
\sim
\mathrm{Gamma}\left(
\frac{\nu_j+1}{2},
\frac{1}{2}
\left\{
\nu_j+\frac{(\beta_j-b_{0j})^2}{s_j^2}
\right\}
\right).
\]
Hence, replacing the Gaussian prior by independent Student-$t$ priors only requires additional Gamma updates while preserving the Gaussian update for $\bld{\beta}$.

\subsection{Horseshoe priors for regression coefficients\label{app:horseshoe_extension}}

As a further example, we consider the horseshoe prior proposed by \citet{carvalhoHorseshoeEstimatorSparse2010}, which combines a global shrinkage parameter with coefficient-specific local shrinkage parameters. Specifically,

$$
\beta_j\mid\ell_j^2,\tau^2
\sim
\text{N}(0,\tau^2\ell_j^2),
\qquad
\ell_j\sim C^+(0,1),
\qquad
\tau\sim C^+(0,s_\tau).
$$

Following the inverse-gamma augmentation of \citet{makalicSimpleSamplerHorseshoe2016}, the half-Cauchy priors can be represented using conditionally conjugate inverse-gamma distributions, which yields closed-form Gibbs updates for the local and global scale parameters. Specifically,

$$
\ell_j^2\mid\nu_j\sim IG(1/2,1/\nu_j),
\qquad
\nu_j\sim IG(1/2,1),
$$

and

$$
\tau^2\mid\xi\sim IG(1/2,1/\xi),
\qquad
\xi\sim IG(1/2,1/s_\tau^2).
$$

Conditional on the local and global scales, define
\[
D_{\mathrm{HS}}
=
\mathrm{diag}\left(
\frac{1}{\tau^2\ell_1^2},\ldots,
\frac{1}{\tau^2\ell_p^2}
\right).
\]
The regression coefficients then have the Gaussian full conditional distribution
\[
\bld{\beta}\mid\bld{\omega},\bld{Z},\bld{\ell},\tau^2,\bld{D}
\sim
\mathrm{N}(B_{\mathrm{HS}}^{-1}\bld{g}_{\mathrm{HS}},B_{\mathrm{HS}}^{-1}),
\]
where
\[
B_{\mathrm{HS}}=X^\top\Omega X+D_{\mathrm{HS}},
\qquad
\bld{g}_{\mathrm{HS}}=X^\top\bld{r}.
\]
The remaining full conditional distributions are
\[
\ell_j^2\mid\cdots
\sim
\mathrm{IG}\left(
1,
\frac{1}{\nu_j}+\frac{\beta_j^2}{2\tau^2}
\right),
\]
\[
\nu_j\mid\cdots
\sim
\mathrm{IG}\left(
1,
1+\frac{1}{\ell_j^2}
\right),
\]
\[
\tau^2\mid\cdots
\sim
\mathrm{IG}\left(
\frac{p+1}{2},
\frac{1}{\xi}
+\frac{1}{2}\sum_{j=1}^p\frac{\beta_j^2}{\ell_j^2}
\right),
\]
and
\[
\xi\mid\cdots
\sim
\mathrm{IG}\left(
1,
\frac{1}{s_\tau^2}+\frac{1}{\tau^2}
\right).
\]
Therefore, the horseshoe prior can also be incorporated while retaining closed-form conditional updates.
If an intercept is included in the model, it can be assigned a separate Gaussian prior and excluded from the shrinkage hierarchy.

\newpage

\section{Additional simulation studies}

\subsection{Practical specification of the PL-Cox sampler: intercept and approximation parameter $\delta$}\label{app:delta}

The PL-Cox sampler requires two practical specifications that are not parameters of the original Cox partial likelihood: whether an intercept is included in the log-linear representation of the Plackett--Luce parameters and the value of the approximation parameter $\delta$ used in the negative binomial representation. Because the PL likelihood is invariant to a common multiplicative rescaling of all Plackett--Luce parameters, the intercept is not identified from the rank likelihood in the same way as ordinary regression coefficients. In addition, a finite value of $\delta$ introduces an approximation when deriving the P\'{o}lya--Gamma Gibbs update for $\bld{\beta}$.

We therefore investigated the inclusion or exclusion of an intercept jointly with several values of $\delta$. This section first describes the approximation induced by finite $\delta$ and an MH correction that removes this approximation. We then examine the empirical sensitivity of posterior inference and MCMC performance to both the intercept specification and $\delta$, and use these results to motivate the default PL-Cox implementation. Finally, we compare the selected implementation with the MH-corrected sampler to assess the practical magnitude of the remaining approximation error.

\subsubsection{Approximation induced by finite $\delta$ and MH correction}

Conditional on the latent variables $\bld{Z}$, define $\mu_i = \zeta_i\lambda_i = \zeta_i\exp(\bld{x}_i^\top\bld{\beta})$.
Ignoring factors that do not depend on $\bld{\beta}$, the original augmented PL-Cox conditional posterior has kernel
\[
\pi(\bld{\beta}\mid\bld{Z},\bld{D})
\propto
\pi_0(\bld{\beta})
\prod_{i=1}^n
\mu_i^{c_i}\exp(-\mu_i),
\]
where $\pi_0(\bld{\beta})$ denotes the prior density.

Under the Poisson--Gamma construction with $\xi_i\sim\mathrm{Gamma}(\delta,\delta)$, marginalization over $\xi_i$ gives the negative binomial approximation
\[
\pi_\delta(\bld{\beta}\mid\bld{Z},\bld{D})
\propto
\pi_0(\bld{\beta})
\prod_{i=1}^n
\mu_i^{c_i}
\left(
1+\frac{\mu_i}{\delta}
\right)^{-(c_i+\delta)}.
\]
Consequently, the log ratio between the original and approximate conditional kernels, up to an additive constant independent of $\bld{\beta}$, is
\[
H_\delta(\bld{\beta};\bld{Z})
=
\sum_{i=1}^n
h_{i,\delta}(\mu_i),
\]
where
\[
h_{i,\delta}(\mu)
=
-\mu
+
(c_i+\delta)
\log\left(
1+\frac{\mu}{\delta}
\right).
\]
For fixed $c_i$ and $\mu$, a Taylor expansion gives
\[
h_{i,\delta}(\mu)
=
\frac{c_i\mu-\mu^2/2}{\delta}
+
O(\delta^{-2}),
\qquad
\delta\to\infty.
\]
Thus, for fixed $\bld{Z}$ and $\bld{\beta}$, the negative binomial conditional kernel converges pointwise to the original Poisson conditional kernel as $\delta$ increases.

Although this pointwise convergence does not by itself establish convergence of the invariant distribution of the uncorrected sampler, it shows that the conditional approximation becomes more accurate as $\delta$ increases for fixed $\boldsymbol{Z}$ and $\boldsymbol{\beta}$.
We therefore assess the practical impact of the approximation by comparing the uncorrected sampler with an MH-corrected sampler targeting the original PL-Cox posterior.
Let $K_\delta(\bld{\beta}'\mid \bld{\beta},\bld{Z})$ denote one marginal P\'{o}lya--Gamma data-augmentation transition obtained by first sampling the P\'{o}lya--Gamma latent variables under the negative binomial approximation and then sampling a proposed value $\bld{\beta}'$ from the resulting Gaussian full conditional distribution. 
% This transition is reversible with respect to $\pi_\delta(\bld{\beta}\mid\bld{Z},\bld{D})$.
This reversibility result assumes exact P\'{o}lya--Gamma sampling, without the Gaussian approximation used for large shape parameters in the default implementation.
Under this condition, the transition is reversible with respect to $\pi_\delta(\bld{\beta}\mid\bld{Z},\bld{D})$.
Therefore, conditional on $\bld{Z}$, the proposed value $\bld{\beta}'$ can be accepted with probability
\[
\alpha(\bld{\beta}, \bld{\beta}'; \bld{Z}) =
\min \left\{1, \exp\left[H_\delta(\bld{\beta}';\bld{Z}) - H_\delta(\bld{\beta}; \bld{Z}) \right] \right\}. 
\]

Specifically, after updating $\bld{Z}$ from its exact gamma full conditional distribution, we generate $\bld{\beta}'$ using the exact P\'{o}lya--Gamma data-augmentation transition $K_\delta$ and then apply the above acceptance step.
Hence, the MH-corrected $\bld{\beta}$ update has the original conditional posterior $\pi(\bld{\beta}\mid\bld{Z},\bld{D})$ as its invariant distribution. 
Combined with the exact update of $\bld{Z}$, the resulting Metropolis-within-Gibbs sampler targets the original augmented PL-Cox posterior. 
This correction follows the same general strategy of using an approximate P\'{o}lya--Gamma transition as an MH proposal to remove approximation error as considered by \citet{renCoxPolyaGammaAlgorithmFlexible2025}.

\subsubsection{Sensitivity to the intercept specification and $\delta$}

We considered $n=300$ under the three Weibull baseline hazard shapes described in Appendix~\ref{app:sim_settings}, with either no rounding or a rounding width of $\Delta=0.1$, and used independent $\text{N}(0,100)$ priors for the regression coefficients.
We conducted a simulation study to investigate two implementation choices jointly: the inclusion of an intercept and the value of the approximation parameter $\delta$. 
For each simulation setting, the PL-Cox sampler was fitted both with and without an intercept, and, for each specification, we considered $\delta\in\{10,100,1000,3000\}$. 
All prior specifications and data-generation settings were held fixed across the comparisons. 
The MCMC length was varied with $\delta$, as described below, while the same MCMC settings were used for the specifications with and without an intercept for a given value of $\delta$.
Each setting was evaluated using 1,000 simulated datasets.
For this sensitivity analysis, the P\'{o}lya--Gamma sampler used a moment-matched Gaussian approximation when the shape parameter exceeded 50.
Since the shape parameter is $c_i+\delta$, with $c_i\in\{0,1\}$, exact P\'{o}lya--Gamma draws were used for $\delta=10$, whereas the Gaussian approximation was used for $\delta\in\{100,1000,3000\}$.
Consequently, comparisons involving $\delta=10$ reflect both the effect of $\delta$ and the change in the P\'{o}lya--Gamma sampling scheme.

The inclusion of an intercept requires particular consideration in the PL-Cox formulation because multiplying all Plackett--Luce parameters by a common positive constant leaves the PL likelihood unchanged. 
Consequently, a common intercept corresponds to a redundant overall scale in the original rank likelihood. 
Nevertheless, under the negative binomial approximation, its inclusion could potentially affect posterior computation. 
We therefore examined the intercept and $\delta$ jointly rather than excluding the intercept solely on identifiability grounds.

For each combination of the intercept specification and $\delta$, we evaluated estimation accuracy using bias, root mean squared error (RMSE), empirical coverage probability, average interval width (AW), and posterior standard deviation. 
We additionally examined $\widehat{R}$ and effective sample sizes for each regression coefficient, following \citet{vehtariRankNormalizationFoldingLocalization}.

\begin{table}[!htbp]
\centering
\begin{threeparttable}
\caption{Sensitivity of the PL-Cox sampler to the approximation parameter
$\delta$ under the intercept-free specification across three baseline
hazard shapes and two rounding settings.}
\label{tab:delta_no_intercept}
% \small
\setlength{\tabcolsep}{3.5pt}
\renewcommand{\arraystretch}{1.05}
\begin{tabular}{lcrrrrrr}
\toprule
Baseline hazard & $\Delta$ & $\delta$ & Bias & RMSE & Coverage & AW & Post. SD \\
\midrule
Constant
 & None & 10 & 0.0043 & 0.078 & 0.948 & 0.306 & 0.078 \\
 &  & 100 & 0.0010 & 0.074 & 0.943 & 0.285 & 0.073 \\
 &  & 1000 & 0.0006 & 0.073 & 0.943 & 0.282 & 0.072 \\
 &  & 3000 & 0.0005 & 0.073 & 0.945 & 0.282 & 0.072 \\
\cmidrule(lr){2-8}
 & 0.1 & 10 & 0.0038 & 0.077 & 0.951 & 0.305 & 0.078 \\
 &  & 100 & 0.0005 & 0.073 & 0.947 & 0.284 & 0.073 \\
 &  & 1000 & 0.0001 & 0.073 & 0.945 & 0.282 & 0.072 \\
 &  & 3000 & 0.0001 & 0.073 & 0.946 & 0.282 & 0.072 \\
\midrule
Decreasing
 & None & 10 & 0.0047 & 0.077 & 0.953 & 0.298 & 0.076 \\
 &  & 100 & 0.0010 & 0.072 & 0.947 & 0.277 & 0.071 \\
 &  & 1000 & 0.0007 & 0.072 & 0.947 & 0.275 & 0.070 \\
 &  & 3000 & 0.0006 & 0.072 & 0.946 & 0.275 & 0.070 \\
\cmidrule(lr){2-8}
 & 0.1 & 10 & 0.0038 & 0.075 & 0.956 & 0.297 & 0.076 \\
 &  & 100 & 0.0003 & 0.072 & 0.950 & 0.277 & 0.071 \\
 &  & 1000 & $-$0.0001 & 0.071 & 0.949 & 0.275 & 0.070 \\
 &  & 3000 & $-$0.0001 & 0.071 & 0.949 & 0.275 & 0.070 \\
\midrule
Increasing
 & None & 10 & 0.0031 & 0.078 & 0.966 & 0.318 & 0.081 \\
 &  & 100 & 0.0008 & 0.075 & 0.955 & 0.294 & 0.075 \\
 &  & 1000 & 0.0005 & 0.075 & 0.955 & 0.291 & 0.074 \\
 &  & 3000 & 0.0006 & 0.075 & 0.955 & 0.291 & 0.074 \\
\cmidrule(lr){2-8}
 & 0.1 & 10 & 0.0025 & 0.077 & 0.968 & 0.318 & 0.081 \\
 &  & 100 & 0.0003 & 0.074 & 0.957 & 0.294 & 0.075 \\
 &  & 1000 & 0.0001 & 0.074 & 0.956 & 0.291 & 0.074 \\
 &  & 3000 & 0.0001 & 0.074 & 0.956 & 0.291 & 0.074 \\
\bottomrule
\end{tabular}
\begin{tablenotes}
\small
\item Event times were generated under the three baseline hazard shapes
described in Appendix~\ref{app:sim_settings}. Here, $\Delta$ denotes the rounding width;
``None'' indicates that no rounding was applied. Each setting is based on
1,000 simulated datasets. The summaries are calculated over the four slope
coefficients. AW denotes the average width of the 95\% credible intervals;
Post. SD denotes the mean posterior standard deviation.
\end{tablenotes}
\end{threeparttable}
\end{table}

\begin{table}[!htbp]
\centering
\begin{threeparttable}
\caption{Sensitivity of the PL-Cox sampler to the approximation parameter
$\delta$ when an intercept is included across three baseline hazard shapes
and two rounding settings.}
\label{tab:delta_with_intercept}
% \small
\setlength{\tabcolsep}{3.5pt}
\renewcommand{\arraystretch}{1.05}
\begin{tabular}{lcrrrrrr}
\toprule
Baseline hazard & $\Delta$ & $\delta$ & Bias & RMSE & Coverage & AW & Post. SD \\
\midrule
Constant
 & None & 10 & $-$0.0040 & 0.069 & 0.967 & 0.296 & 0.076 \\
 &  & 100 & 0.0009 & 0.074 & 0.944 & 0.284 & 0.073 \\
 &  & 1000 & 0.0006 & 0.073 & 0.943 & 0.282 & 0.072 \\
 &  & 3000 & 0.0005 & 0.073 & 0.943 & 0.282 & 0.072 \\
\cmidrule(lr){2-8}
 & 0.1 & 10 & $-$0.0044 & 0.069 & 0.966 & 0.296 & 0.076 \\
 &  & 100 & 0.0005 & 0.073 & 0.946 & 0.284 & 0.073 \\
 &  & 1000 & 0.0001 & 0.073 & 0.946 & 0.282 & 0.072 \\
 &  & 3000 & 0.0001 & 0.073 & 0.943 & 0.282 & 0.072 \\
\midrule
Decreasing
 & None & 10 & $-$0.0041 & 0.068 & 0.965 & 0.289 & 0.074 \\
 &  & 100 & 0.0010 & 0.072 & 0.948 & 0.277 & 0.071 \\
 &  & 1000 & 0.0006 & 0.072 & 0.947 & 0.275 & 0.070 \\
 &  & 3000 & 0.0006 & 0.072 & 0.948 & 0.275 & 0.070 \\
\cmidrule(lr){2-8}
 & 0.1 & 10 & $-$0.0047 & 0.068 & 0.966 & 0.289 & 0.074 \\
 &  & 100 & 0.0003 & 0.071 & 0.950 & 0.277 & 0.071 \\
 &  & 1000 & $-$0.0001 & 0.071 & 0.949 & 0.275 & 0.070 \\
 &  & 3000 & $-$0.0001 & 0.071 & 0.948 & 0.275 & 0.070 \\
\midrule
Increasing
 & None & 10 & $-$0.0046 & 0.070 & 0.976 & 0.307 & 0.078 \\
 &  & 100 & 0.0008 & 0.075 & 0.955 & 0.294 & 0.075 \\
 &  & 1000 & 0.0005 & 0.075 & 0.956 & 0.292 & 0.074 \\
 &  & 3000 & 0.0005 & 0.075 & 0.954 & 0.291 & 0.074 \\
\cmidrule(lr){2-8}
 & 0.1 & 10 & $-$0.0050 & 0.070 & 0.976 & 0.306 & 0.078 \\
 &  & 100 & 0.0003 & 0.074 & 0.957 & 0.294 & 0.075 \\
 &  & 1000 & 0.0001 & 0.074 & 0.955 & 0.291 & 0.074 \\
 &  & 3000 & 0.0000 & 0.074 & 0.955 & 0.291 & 0.074 \\
\bottomrule
\end{tabular}
\begin{tablenotes}
\small
\item Event times were generated under the three baseline hazard shapes
described in Appendix~\ref{app:sim_settings}. Here, $\Delta$ denotes the rounding width;
``None'' indicates that no rounding was applied. Each setting is based on
1,000 simulated datasets. The summaries are calculated over the four slope
coefficients. AW denotes the average width of the 95\% credible intervals;
Post. SD denotes the mean posterior standard deviation.
\end{tablenotes}
\end{threeparttable}
\end{table}

Tables~\ref{tab:delta_no_intercept} and \ref{tab:delta_with_intercept} summarize the simulation results for the intercept-free and intercept-inclusive specifications, respectively.
A clear dependence on the intercept specification was observed when $\delta=10$. 
With an intercept, the mean bias was negative in all six simulation settings, ranging from $-0.00502$ to $-0.00402$, whereas without an intercept it was positive in all six settings, ranging from $0.00248$ to $0.00470$. 
The corresponding RMSEs and interval widths also differed noticeably between the two intercept specifications.
Thus, when $\delta=10$, posterior inference was affected by the inclusion of an intercept even though the common scale is not identifiable from the original PL likelihood.

In contrast, the two intercept specifications produced nearly identical results once $\delta$ reached 100. At $\delta=100$, the RMSE, interval width, and posterior standard deviation were nearly identical between the two intercept specifications in all six simulation settings.
Similar agreement was observed for $\delta=1000$ and $\delta=3000$.
Since the intercept corresponds to a redundant common scale in the original PL likelihood and provided no material difference in inference once $\delta$ was sufficiently large, we adopted the intercept-free specification as the default PL-Cox model.

Under the intercept-free specification, the results at $\delta=100$ were already close to those obtained with larger values of $\delta$.
Across the six individual simulation settings, increasing $\delta$ from 100 to 1000 changed the RMSE by only $0.00023$--$0.00043$, while increasing $\delta$ further from 1000 to 3000 changed the RMSE by at most $7.1\times10^{-5}$.
The corresponding changes in empirical coverage probability were also small: the absolute difference between $\delta=100$ and $\delta=1000$ was at most 0.0015 across the six settings.
Posterior standard deviations showed a similarly small change.
These results indicate that $\delta=100$ was already sufficiently close to the results obtained with substantially larger values of $\delta$. 
We therefore selected $\delta=100$ as the default value for the intercept-free PL-Cox sampler.

We also examined the effect of $\delta$ on MCMC efficiency under the selected intercept-free specification. 
For each simulated dataset, four independent MCMC chains were run without thinning. 
The number of iterations per chain was varied with $\delta$: for $\delta=10$, each chain was run for 10,000 iterations, with the first 5,000 discarded as burn-in; for $\delta=100$, 25,000 iterations with 5,000 burn-in iterations; for $\delta=1000$, 70,000 iterations with 10,000 burn-in iterations; and for $\delta=3000$, 220,000 iterations with 20,000 burn-in iterations. 
Thus, the numbers of retained draws per chain were 5,000, 20,000, 60,000, and 200,000, respectively. 
The same MCMC settings were used for the specifications with and without an intercept.
Because the number of retained draws differed across values of $\delta$, raw ESS values are not directly comparable as measures of computational efficiency. 
We therefore also report ESS per second, which accounts for both sampling efficiency and computation time.

\begin{table}[!htbp]
\centering
\caption{MCMC efficiency of the intercept-free PL-Cox sampler for different values of $\delta$.}
\label{tab:delta_mcmc}
\setlength{\tabcolsep}{3pt}
\begin{tabular}{llrrrrrrr}
\toprule
Baseline hazard & $\Delta$ & $\delta$ & Med. max $\widehat R$ & Bulk ESS & Tail ESS & Time (s) & Bulk ESS/s & Tail ESS/s \\
\midrule
Constant & None & 10 & 1.001 & 3872 & 7489 & 45 & 85.78 & 166.15 \\
 &  & 100 & 1.002 & 2571 & 5561 & 13 & 197.41 & 426.92 \\
 &  & 1000 & 1.005 & 1073 & 2317 & 36 & 29.60 & 64.00 \\
 &  & 3000 & 1.004 & 1365 & 2956 & 113 & 12.05 & 26.13 \\
\addlinespace
Constant & 0.1 & 10 & 1.001 & 3860 & 7472 & 45 & 86.55 & 167.50 \\
 &  & 100 & 1.002 & 2576 & 5555 & 12 & 217.52 & 470.74 \\
 &  & 1000 & 1.005 & 1073 & 2309 & 33 & 32.56 & 69.94 \\
 &  & 3000 & 1.004 & 1356 & 2973 & 103 & 13.22 & 28.90 \\
\midrule
Decreasing & None & 10 & 1.001 & 4047 & 7753 & 45 & 89.34 & 171.20 \\
 &  & 100 & 1.002 & 2678 & 5848 & 13 & 204.17 & 444.59 \\
 &  & 1000 & 1.005 & 1124 & 2422 & 37 & 30.74 & 66.22 \\
 &  & 3000 & 1.004 & 1438 & 3099 & 115 & 12.56 & 27.07 \\
\addlinespace
Decreasing & 0.1 & 10 & 1.001 & 4034 & 7756 & 45 & 90.55 & 174.02 \\
 &  & 100 & 1.002 & 2680 & 5824 & 12 & 231.77 & 503.08 \\
 &  & 1000 & 1.005 & 1122 & 2409 & 32 & 34.75 & 74.84 \\
 &  & 3000 & 1.004 & 1422 & 3094 & 101 & 14.16 & 30.74 \\
\midrule
Increasing & None & 10 & 1.001 & 3694 & 7223 & 45 & 82.13 & 160.31 \\
 &  & 100 & 1.002 & 2464 & 5344 & 13 & 191.00 & 414.92 \\
 &  & 1000 & 1.005 & 1020 & 2198 & 36 & 28.36 & 61.09 \\
 &  & 3000 & 1.004 & 1302 & 2808 & 112 & 11.62 & 25.16 \\
\addlinespace
Increasing & 0.1 & 10 & 1.001 & 3697 & 7222 & 45 & 82.83 & 162.06 \\
 &  & 100 & 1.002 & 2469 & 5336 & 12 & 207.13 & 448.13 \\
 &  & 1000 & 1.005 & 1025 & 2207 & 33 & 30.88 & 66.16 \\
 &  & 3000 & 1.004 & 1298 & 2811 & 104 & 12.54 & 27.17 \\
\bottomrule
\end{tabular}
\begin{minipage}{0.98\linewidth}
\footnotesize Event times were generated under three baseline hazard shapes. $\Delta$ is the rounding width; None denotes no rounding. Each setting uses 1,000 simulated datasets. Med. max $\widehat R$ is the median across datasets of the maximum among the four slope coefficients; Bulk ESS and Tail ESS are medians across datasets of the minima across slopes. Time is the median elapsed time per dataset. Each ESS/s value is the median of the dataset-level minimum slope ESS divided by its elapsed time.
\end{minipage}
\end{table}

Table~\ref{tab:delta_mcmc} shows that $\delta=100$ also provided favorable MCMC performance across all six simulation settings. 
The median maximum $\widehat{R}$ remained close to 1, while $\delta=100$ consistently yielded the highest ESS per second. 
Although larger values of $\delta$ produced similar inferential results, they substantially increased computation time and reduced sampling efficiency. 
Taken together with the sensitivity results in Tables~\ref{tab:delta_no_intercept} and \ref{tab:delta_with_intercept}, these findings support the use of $\delta=100$ as the default value.
No simulated dataset exceeded the $\widehat{R}=1.01$ threshold at $\delta=100$, whereas a small proportion did so at $\delta=1000$.

\subsubsection{Comparison with the MH-corrected sampler}

Finally, we compared the uncorrected PL-Cox sampler with $\delta=100$ against the MH-corrected sampler described above. 
Both samplers used exact P\'{o}lya--Gamma draws, isolating the effect of the finite-$\delta$ approximation. 
The MH-corrected sampler targets the original PL-Cox posterior and provides a benchmark. 
Table~\ref{tab:delta_mh} summarizes the comparison across the six simulation settings.

\begin{table}[!htbp]
\centering
\begin{threeparttable}
\caption{Comparison of the uncorrected PL-Cox sampler with $\delta=100$ and the MH-corrected sampler targeting the original PL-Cox posterior (100 simulated datasets per scenario).}
\label{tab:delta_mh}
\setlength{\tabcolsep}{4pt}
\renewcommand{\arraystretch}{1.05}
\begin{tabular}{lcrrrrr}
\toprule
Baseline hazard & $\Delta$
& Mean abs. diff.
& RMSE diff.
& SD diff.
& AW diff.
& MH accept. \\
\midrule
Constant & None & 0.00232 & 0.00278 & 0.00072 & 0.00292 & 0.9926 \\
 & 0.1 & 0.00224 & 0.00266 & 0.00069 & 0.00272 & 0.9928 \\
\midrule
Decreasing & None & 0.00224 & 0.00271 & 0.00072 & 0.00307 & 0.9925 \\
 & 0.1 & 0.00216 & 0.00261 & 0.00074 & 0.00291 & 0.9928 \\
\midrule
Increasing & None & 0.00235 & 0.00303 & 0.00080 & 0.00328 & 0.9922 \\
 & 0.1 & 0.00221 & 0.00263 & 0.00085 & 0.00342 & 0.9923 \\
\bottomrule
\end{tabular}
\begin{tablenotes}
\small
\item Both samplers used $\delta=100$ without an intercept, with four chains, 5,000 burn-in iterations and 20,000 retained draws per chain, on the same 100 simulated datasets per scenario. P\'{o}lya--Gamma random variables were generated without the Gaussian approximation in both samplers; only the second sampler used MH correction. Mean abs. diff. and RMSE diff. are the maxima over four slope coefficients of the mean absolute and root mean squared paired differences in posterior means, respectively. SD diff. and AW diff. are the maxima of the absolute mean paired differences in posterior standard deviations and 95\% credible interval widths, respectively. MH accept. denotes the mean post-burn-in MH acceptance proportion.
\end{tablenotes}
\end{threeparttable}
\end{table}

Across all six simulation settings, the maximum mean absolute difference in posterior means over the four slope coefficients was 0.00235, and the largest corresponding RMSE was 0.00303. 
Differences in posterior standard deviations and credible interval widths were also small. 
The mean MH acceptance proportions ranged from 0.9922 to 0.9928, indicating that proposals from the uncorrected finite-$\delta$ transition were rarely rejected.

Thus, although the uncorrected sampler with finite $\delta$ is not guaranteed to target the original PL-Cox posterior, the observed discrepancies at $\delta=100$ were small in the six settings examined.
Together with the sensitivity and computational-efficiency results above, these findings support $\delta=100$ as a practical default for the intercept-free PL-Cox sampler. 
However, this empirical agreement does not establish exact posterior targeting or provide a general bound on approximation error. 
This comparison isolates the finite-$\delta$ approximation and does not assess the additional Gaussian approximation used for P\'{o}lya--Gamma draws in the default implementation.

\subsection{Computational comparison with Hamiltonian Monte Carlo}

To assess the computational efficiency of the proposed Gibbs samplers relative to a general-purpose Hamiltonian Monte Carlo (HMC) implementation, we conducted additional simulation studies using Stan. 
We considered both a standard fixed-effects Cox model and a shared-frailty extension. 
The fixed-effects setting examines scalability with respect to the sample size and the number of covariates, whereas the shared-frailty setting examines whether the relative computational performance changes when the dimension of the latent random effects increases with the sample size.

\subsubsection{Simulation settings}

For the fixed-effects setting, covariates were independently generated as $x_{ij} \sim \text{N}(0,1)$. The first four regression coefficients were set to
\[
(\beta_1,\beta_2,\beta_3,\beta_4) = (0.10,0.05,-0.15,0.30),
\]
and, when $p=20$, the remaining 16 coefficients were set to zero.
We considered
\[
n \in \{300,1000,5000,10000\},
\qquad
p \in \{4,20\}.
\]
Event times were generated from a Weibull proportional hazards model with shape parameter 1 and scale parameter $10\exp(-\boldsymbol{x}_i^\top\boldsymbol{\beta})$.
Censoring times were independently generated from $\mathrm{Unif}(0.5,30)$.
The observed times were rounded to the nearest 0.1 to allow tied event times.
For each value of $n$, the settings with $p=4$ and $p=20$ shared the same first four covariates and survival and censoring outcomes, with the additional covariates in the $p=20$ setting being noise variables.

For the shared-frailty setting, we used $p=4$ and the same regression coefficients, baseline distribution, censoring distribution, and rounding rule.
Subjects were divided into clusters of size 10, and the cluster-specific random effects were generated as
\[
u_g \sim \text{N}(0,\sigma_u^2), \qquad \sigma_u^2=1.
\]
Thus, the linear predictor was
$\boldsymbol{x}_i^\top\boldsymbol{\beta}+u_{g(i)}$.
We considered
\[
(n,G) = (300,30),\ (1000,100),\ (5000,500),\ (10000,1000),
\]
where $G$ denotes the number of clusters.

In both settings, we compared PL-Cox, GPL-Cox, and an HMC implementation based on the Breslow partial likelihood in Stan (Stan-Breslow). 
For PL-Cox, we used $\delta=100$ and excluded the intercept, following the results in Section~D.1.2. 
The prior variance for each regression coefficient was 100 for all methods. 
In the shared-frailty setting, we additionally used
\[
\sigma_u^2 \sim \mathrm{IG}(0.01,0.01)
\]
for all three methods. 
The Stan implementation evaluated the Breslow partial likelihood using cumulative log-sum-exp calculations over the ordered risk sets to avoid repeated construction of the full risk sets.

The comparison between PL-Cox and Stan-Breslow is of primary interest because Stan-Breslow directly targets the posterior based on the Breslow partial likelihood, whereas PL-Cox uses finite-$\delta$ approximate conditional updates without MH correction. 
GPL-Cox is included as a secondary comparison based on a distinct probabilistic model for tied event sets.

For the fixed-effects comparisons, PL-Cox and GPL-Cox were implemented using BayesPLCox version 0.5.0. 
For each scenario, 100 independently generated datasets were analyzed.
Four chains were run for each method. For PL-Cox and GPL-Cox, 5,000 iterations per chain were discarded as burn-in and 20,000 posterior draws per chain were retained. 
For Stan-Breslow, 1,000 warm-up iterations and 20,000 post-warm-up draws per chain were used. 
The four chains were run sequentially for all methods. 
Stan compilation and the preliminary Cox fit used only to construct initial values were excluded from the timing, whereas method-specific data preparation was included. Computational efficiency was evaluated using wall-clock time, bulk and tail effective sample sizes (ESSs), and ESS per second.
We regarded $\widehat{R}\leq1.01$ and ESS $\geq400$ as indicating satisfactory mixing.

\begin{table}[htbp]
\centering
\caption{Computational efficiency of PL-Cox, GPL-Cox, and Stan-Breslow in the fixed-effects setting.}
\label{tab:comp_hmc_fixed}
% \small
\begin{tabular}{rrlrrrr}
\hline
$n$ & $p$ & Method & Med. max $\widehat R$ & Time (s) & Bulk ESS/s & Tail ESS/s \\
\hline
300 & 4 & PL-Cox & 1.002 & 13 & 226.40 & 490.44 \\
300 & 4 & GPL-Cox & 1.002 & 121 & 29.66 & 65.08 \\
300 & 4 & Stan-Breslow & 1.000 & 16 & 5742.08 & 3889.82 \\
300 & 20 & PL-Cox & 1.003 & 23 & 116.81 & 254.67 \\
300 & 20 & GPL-Cox & 1.002 & 133 & 26.86 & 58.67 \\
300 & 20 & Stan-Breslow & 1.000 & 22 & 4337.73 & 2933.88 \\
1000 & 4 & PL-Cox & 1.002 & 29 & 101.55 & 220.89 \\
1000 & 4 & GPL-Cox & 1.002 & 388 & 9.63 & 20.97 \\
1000 & 4 & Stan-Breslow & 1.000 & 44 & 2108.64 & 1387.08 \\
1000 & 20 & PL-Cox & 1.003 & 66 & 43.38 & 94.80 \\
1000 & 20 & GPL-Cox & 1.002 & 430 & 8.83 & 19.14 \\
1000 & 20 & Stan-Breslow & 1.000 & 56 & 2160.75 & 1134.54 \\
5000 & 4 & PL-Cox & 1.002 & 114 & 25.70 & 57.01 \\
5000 & 4 & GPL-Cox & 1.002 & 1888 & 1.97 & 4.27 \\
5000 & 4 & Stan-Breslow & 1.000 & 197 & 471.35 & 306.90 \\
5000 & 20 & PL-Cox & 1.003 & 301 & 9.69 & 21.13 \\
5000 & 20 & GPL-Cox & 1.002 & 2096 & 1.82 & 3.94 \\
5000 & 20 & Stan-Breslow & 1.000 & 251 & 537.82 & 245.40 \\
10000 & 4 & PL-Cox & 1.002 & 222 & 13.12 & 28.76 \\
10000 & 4 & GPL-Cox & 1.002 & 3760 & 0.99 & 2.16 \\
10000 & 4 & Stan-Breslow & 1.000 & 388 & 241.36 & 156.06 \\
10000 & 20 & PL-Cox & 1.003 & 606 & 4.86 & 10.54 \\
10000 & 20 & GPL-Cox & 1.002 & 4145 & 0.92 & 1.99 \\
10000 & 20 & Stan-Breslow & 1.000 & 500 & 274.58 & 122.71 \\
\hline
\end{tabular}
\par\vspace{2pt}
\begin{minipage}{0.98\linewidth}
\footnotesize Values are medians across 100 simulated datasets. For each dataset, Bulk ESS/s and Tail ESS/s are medians across the slope coefficients. Med. max $\widehat R$ is the median of within-dataset maximum slope $\widehat R$.
\end{minipage}
\end{table}

Table ~\ref{tab:comp_hmc_fixed} shows that PL-Cox required less wall-clock time than Stan-Breslow in the four settings with $p=4$, whereas Stan-Breslow was faster in the four settings with $p=20$. 
GPL-Cox required more time than either method in all eight settings. 
In contrast, Stan-Breslow achieved substantially higher bulk and tail ESS per second than both Gibbs samplers in every setting. For example, at $(n,p)=(10000,4)$, the median times were 222, 388, and 3760 seconds for PL-Cox, Stan-Breslow, and GPL-Cox, respectively. 
At $(n,p)=(10000,20)$, the corresponding values were 606, 500, and 4145 seconds. 
In the latter setting, the median bulk ESS per second was 4.86, 274.58, and 0.92, respectively.

% -----------------------------------------------------------------------------
% Shared frailty: computational efficiency
% -----------------------------------------------------------------------------

The shared-frailty results reported in Table~\ref{tab:comp_hmc_frailty} were obtained using \texttt{BayesPLCox} version 0.5.0 for PL-Cox and GPL-Cox.
All methods used four serial MCMC chains with 20,000 retained draws per chain, with 5,000 burn-in iterations for PL-Cox and GPL-Cox and 1,000 warmup iterations for Stan-Breslow.
The regression-coefficient prior was $\text{N}(0,100)$ and the frailty-variance prior was $\mathrm{IG}(0.01,0.01)$.
PL-Cox used $\delta=100$ without MH correction and with PG normal-approximation threshold 50.
PL-Cox and Stan-Breslow are based on the same Breslow partial- likelihood formulation and priors. 
However, Stan-Breslow directly targets the corresponding posterior, whereas PL-Cox uses finite-$\delta$ approximate conditional updates without MH correction.
GPL-Cox was included as an additional, distinct model comparison. 
All runtime comparisons exclude Stan model compilation.

\begin{table}[htbp]
\centering
\caption{Computational efficiency of PL-Cox, GPL-Cox, and Stan-Breslow in the shared-frailty setting.}
\label{tab:comp_hmc_frailty}
\small
\begin{tabular}{rrlrrrrr}
\hline
$n$ & $G$ & Method & Med. max $\widehat R$ & Time (s) & Bulk ESS/s & Tail ESS/s & $\sigma_u^2$ Bulk ESS/s \\
\hline
300 & 30 & PL-Cox & 1.003 & 18 & 138.54 & 314.97 & 127.83 \\
300 & 30 & GPL-Cox & 1.002 & 54 & 63.43 & 143.78 & 52.63 \\
300 & 30 & Stan-Breslow & 1.000 & 39 & 2885.16 & 1582.94 & 617.77 \\
1000 & 100 & PL-Cox & 1.003 & 35 & 73.71 & 164.04 & 72.63 \\
1000 & 100 & GPL-Cox & 1.002 & 150 & 23.93 & 54.52 & 21.54 \\
1000 & 100 & Stan-Breslow & 1.000 & 122 & 1064.49 & 526.73 & 196.62 \\
5000 & 500 & PL-Cox & 1.002 & 118 & 21.58 & 48.21 & 22.02 \\
5000 & 500 & GPL-Cox & 1.002 & 668 & 5.49 & 12.26 & 5.43 \\
5000 & 500 & Stan-Breslow & 1.000 & 755 & 119.57 & 82.09 & 24.95 \\
10000 & 1000 & PL-Cox & 1.002 & 218 & 11.77 & 26.28 & 12.52 \\
10000 & 1000 & GPL-Cox & 1.002 & 1315 & 2.84 & 6.33 & 2.77 \\
10000 & 1000 & Stan-Breslow & 1.000 & 2182 & 69.48 & 29.82 & 11.54 \\
\hline
\end{tabular}
\par\vspace{2pt}
\begin{minipage}{0.98\linewidth}
\footnotesize Values are medians over 100 simulated datasets, with cluster size 10 ($G=n/10$). Med. max $\widehat R$, Bulk ESS/s, and Tail ESS/s refer to the four slopes and frailty variance $\sigma_u^2$; the last column is for $\sigma_u^2$ alone. All 100 fits per method and scenario passed the prespecified key-parameter convergence criteria.
Bulk ESS/s and Tail ESS/s are the medians across datasets of the within-dataset medians over the four regression coefficients and $\sigma_u^2$.
Med.\ max $\widehat R$ is the median of the within-dataset maxima over the same parameters.
\end{minipage}
\end{table}

Table~\ref{tab:comp_hmc_frailty} shows that PL-Cox was faster than Stan-Breslow in all four shared-frailty settings. 
The Stan-to-PL ratio of median computation times increased from 2.21 at $(n,G)=(300,30)$ to 9.99 at $(n,G)=(10000,1000)$. 
At the largest setting, the median times were 218, 2182, and 1315 seconds for PL-Cox, Stan-Breslow, and GPL-Cox, respectively. 
Because sample size and group count increased together, these comparisons do not isolate the effect of random-effect dimension. 
Nevertheless, Stan-Breslow had greater effective-sample efficiency. 
At $(n,G)=(10000,1000)$, its median bulk ESS per second for the key parameters was 69.48, compared with 11.77 for PL-Cox; for the frailty variance alone, the values were 11.54 and 12.52, respectively. 
All three methods passed the prespecified key-parameter convergence criteria in every replication, and Stan-Breslow produced no divergent transitions or maximum-treedepth transitions.

% -----------------------------------------------------------------------------
% Fixed-effects: estimation performance
% -----------------------------------------------------------------------------
\begin{table}[!htbp]
\centering
\begin{threeparttable}
\caption{Estimation performance of PL-Cox, GPL-Cox, and Stan-Breslow in the fixed-effects setting.}
\label{tab:perf_hmc_fixed}
\setlength{\tabcolsep}{4pt}
\renewcommand{\arraystretch}{1.05}
\begin{tabular}{rrlrrrr}
\toprule
$n$ & $p$ & Method & Bias & RMSE & Coverage & AW \\
\midrule
300 & 4 & PL-Cox       & $-$0.0041 & 0.079 & 0.935 & 0.285 \\
 &  & GPL-Cox      & $-$0.0040 & 0.079 & 0.938 & 0.284 \\
 &  & Stan-Breslow & $-$0.0044 & 0.078 & 0.935 & 0.283 \\
\cmidrule(lr){2-7}
 & 20 & PL-Cox       & 0.0009 & 0.080 & 0.942 & 0.297 \\
 &  & GPL-Cox      & 0.0007 & 0.080 & 0.941 & 0.296 \\
 &  & Stan-Breslow & 0.0008 & 0.080 & 0.940 & 0.294 \\
\midrule
1000 & 4 & PL-Cox       & $-$0.0003 & 0.038 & 0.953 & 0.153 \\
 &  & GPL-Cox      & $-$0.0001 & 0.037 & 0.955 & 0.152 \\
 &  & Stan-Breslow & $-$0.0006 & 0.037 & 0.953 & 0.152 \\
\cmidrule(lr){2-7}
 & 20 & PL-Cox       & $-$0.0005 & 0.039 & 0.948 & 0.154 \\
 &  & GPL-Cox      & $-$0.0006 & 0.039 & 0.948 & 0.154 \\
 &  & Stan-Breslow & $-$0.0006 & 0.039 & 0.948 & 0.153 \\
\midrule
5000 & 4 & PL-Cox       & 0.0006 & 0.017 & 0.968 & 0.068 \\
 &  & GPL-Cox      & 0.0006 & 0.016 & 0.973 & 0.068 \\
 &  & Stan-Breslow & 0.0002 & 0.016 & 0.970 & 0.068 \\
\cmidrule(lr){2-7}
 & 20 & PL-Cox       & 0.0000 & 0.017 & 0.955 & 0.068 \\
 &  & GPL-Cox      & 0.0000 & 0.017 & 0.953 & 0.068 \\
 &  & Stan-Breslow & 0.0000 & 0.017 & 0.954 & 0.067 \\
\midrule
10000 & 4 & PL-Cox       & 0.0003 & 0.012 & 0.955 & 0.048 \\
 &  & GPL-Cox      & 0.0004 & 0.012 & 0.955 & 0.048 \\
 &  & Stan-Breslow & $-$0.0001 & 0.012 & 0.963 & 0.048 \\
\cmidrule(lr){2-7}
 & 20 & PL-Cox       & $-$0.0003 & 0.012 & 0.951 & 0.048 \\
 &  & GPL-Cox      & $-$0.0003 & 0.012 & 0.952 & 0.048 \\
 &  & Stan-Breslow & $-$0.0004 & 0.012 & 0.952 & 0.047 \\
\bottomrule
\end{tabular}
\begin{tablenotes}
\small
\item Values are based on 100 simulated datasets. Bias, RMSE, empirical
coverage probability, and average interval width (AW) are aggregated over all
regression coefficients. For $p=20$, this includes the four nonzero coefficients
and the 16 coefficients whose true values are zero.
\end{tablenotes}
\end{threeparttable}
\end{table}
Table~\ref{tab:perf_hmc_fixed} summarizes
the estimation performance across all
regression coefficients.
The three methods yielded similar bias,
RMSE, empirical coverage probabilities,
and average interval widths across the
eight fixed-effects settings.
RMSE and average interval width generally
decreased as the sample size increased.
Thus, the large difference in ESS per second
between Stan-Breslow and the two Gibbs samplers
was not accompanied by a similarly large
difference in estimation accuracy
in these simulation settings.

% -----------------------------------------------------------------------------
% Shared frailty: estimation performance
% -----------------------------------------------------------------------------

\begin{table}[htbp]
\centering
\caption{Estimation performance of PL-Cox, GPL-Cox, and Stan-Breslow in the shared-frailty setting.}
\label{tab:perf_hmc_frailty}
\begin{tabular}{rrllrrrr}
\hline
$n$ & $G$ & Parameter & Method & Bias & RMSE & Coverage & AW \\
\hline
300 & 30 & Regression coef. & PL-Cox & 0.0024 & 0.081 & 0.948 & 0.313 \\
300 & 30 & Regression coef. & GPL-Cox & 0.0031 & 0.081 & 0.943 & 0.312 \\
300 & 30 & Regression coef. & Stan-Breslow & 0.0019 & 0.081 & 0.945 & 0.310 \\
300 & 30 & Frailty variance & PL-Cox & 0.0764 & 0.384 & 0.930 & 1.515 \\
300 & 30 & Frailty variance & GPL-Cox & 0.0877 & 0.381 & 0.940 & 1.493 \\
300 & 30 & Frailty variance & Stan-Breslow & 0.0589 & 0.372 & 0.920 & 1.481 \\
1000 & 100 & Regression coef. & PL-Cox & $-$0.0027 & 0.044 & 0.935 & 0.170 \\
1000 & 100 & Regression coef. & GPL-Cox & $-$0.0023 & 0.044 & 0.938 & 0.169 \\
1000 & 100 & Regression coef. & Stan-Breslow & $-$0.0030 & 0.044 & 0.930 & 0.168 \\
1000 & 100 & Frailty variance & PL-Cox & 0.0091 & 0.197 & 0.940 & 0.743 \\
1000 & 100 & Frailty variance & GPL-Cox & 0.0155 & 0.192 & 0.950 & 0.723 \\
1000 & 100 & Frailty variance & Stan-Breslow & $-$0.0057 & 0.192 & 0.940 & 0.729 \\
5000 & 500 & Regression coef. & PL-Cox & $-$0.0002 & 0.020 & 0.950 & 0.076 \\
5000 & 500 & Regression coef. & GPL-Cox & 0.0002 & 0.019 & 0.948 & 0.075 \\
5000 & 500 & Regression coef. & Stan-Breslow & $-$0.0006 & 0.020 & 0.943 & 0.075 \\
5000 & 500 & Frailty variance & PL-Cox & 0.0014 & 0.080 & 0.950 & 0.324 \\
5000 & 500 & Frailty variance & GPL-Cox & 0.0075 & 0.076 & 0.970 & 0.312 \\
5000 & 500 & Frailty variance & Stan-Breslow & $-$0.0130 & 0.079 & 0.920 & 0.318 \\
10000 & 1000 & Regression coef. & PL-Cox & $-$0.0012 & 0.014 & 0.955 & 0.053 \\
10000 & 1000 & Regression coef. & GPL-Cox & $-$0.0006 & 0.013 & 0.963 & 0.053 \\
10000 & 1000 & Regression coef. & Stan-Breslow & $-$0.0015 & 0.014 & 0.955 & 0.053 \\
10000 & 1000 & Frailty variance & PL-Cox & $-$0.0109 & 0.055 & 0.960 & 0.226 \\
10000 & 1000 & Frailty variance & GPL-Cox & $-$0.0001 & 0.052 & 0.970 & 0.218 \\
10000 & 1000 & Frailty variance & Stan-Breslow & $-$0.0247 & 0.059 & 0.930 & 0.222 \\
\hline
\end{tabular}
\par\vspace{2pt}
\begin{minipage}{0.98\linewidth}
\footnotesize Based on 100 simulated datasets per setting. Regression-coefficient metrics are pooled over the four slopes; frailty-variance metrics refer to $\sigma_u^2$. AW is the mean 95\% credible interval width.
\end{minipage}
\end{table}

\clearpage

\subsection{Operating characteristics under alternative sample sizes and prior specifications}

\subsubsection{Sensitivity to sample size}

To examine the operating characteristics of the methods in smaller samples, we conducted an additional simulation study with sample sizes $n \in \{30,50,100\}$. 
We considered four representative settings from the main simulation study: continuous survival times under the constant baseline hazard with either no rounding ($\Delta=0$) or rounding width $\Delta=0.1$, and discrete survival times under the constant baseline hazard with observation-grid widths $u=1$ and $u=7$. 
Apart from sample size, the data-generating mechanisms followed those of the corresponding main simulation settings.
For PL-Cox and GPL-Cox, the additional small-sample analyses used \texttt{standardize=FALSE}, whereas the $n=300$ reference results used \texttt{standardize=TRUE}.
Thus, comparisons with $n=300$ should be interpreted as descriptive rather than as varying sample size alone.
For each additional setting, 1,000 independently generated datasets were analyzed using Breslow, Efron, Exact, PL-Cox, GPL-Cox, Ren, and Tamano. 
The corresponding results for $n=300$, based on 10,000 simulated datasets from the main simulation study, are reproduced for comparison in Tables~\ref{tab:sample_size_C0}--\ref{tab:sample_size_D7}. 
Following the main simulation study, results are reported for $\beta_4$.
For the additional $n=30,50,100$ settings, the Tamano method was evaluated using the authors' official Python implementation of GS4Cox. 
The $n=300$ reference values are inherited from the main simulation study and were not recomputed as part of this sensitivity analysis.

Across the four settings, RMSE decreased substantially as the sample size increased for all methods. 
The proposed methods showed the same general pattern, although interval coverage was below the nominal 95\% level in some of the small-sample settings. 
This was particularly evident for the more strongly coarsened discrete-time setting with $u=7$: at $n=30$, the coverage probabilities were 89.8\% for PL-Cox and 89.6\% for GPL-Cox, increasing to 94.4\% and 94.0\%, respectively, at $n=100$.
Similar improvements in estimation accuracy were observed in the other tied-event settings.

GPL-Cox showed a different pattern in the continuous-time setting without rounding. 
Although its RMSE decreased with increasing sample size, its coverage probability remained below the nominal level, including at $n=300$. 
Thus, the additional simulations indicate that the proposed methods become substantially more accurate as the sample size increases, while the interval performance of GPL-Cox can remain sensitive to the data-generating setting when tied events are absent.

\begin{table}[!htbp]
\centering
\caption{Sample-size sensitivity for $\beta_4$ under continuous survival times under the constant baseline hazard with no rounding ($\Delta=0$). Reported metrics are empirical bias (Bias), the Monte Carlo standard deviation (SD) of the point estimates, root mean squared error (RMSE), coverage probability of the 95\% interval (CP), and average interval width (AW).}
\label{tab:sample_size_C0}
\begin{tabular}{r l r r r r r}
\toprule
$n$ & Method & Bias & SD & RMSE & CP & AW \\
\midrule
30 & Breslow & 0.055 & 0.331 & 0.336 & 93.20 & 1.110 \\
 & Efron & 0.055 & 0.331 & 0.336 & 93.20 & 1.110 \\
 & Exact & 0.055 & 0.331 & 0.336 & 93.20 & 1.110 \\
 & PL-Cox & 0.074 & 0.353 & 0.360 & 90.00 & 1.142 \\
  &  GPL-Cox  & 0.057 & 0.347 & 0.352 & 87.50 & 1.062 \\
 & Ren & 0.116 & 0.394 & 0.411 & 84.80 & 1.143 \\
 & Tamano & 0.050 & 0.326 & 0.330 & 92.70 & 1.106 \\
\hline
50 & Breslow & 0.030 & 0.216 & 0.218 & 93.30 & 0.782 \\
 & Efron & 0.030 & 0.216 & 0.218 & 93.30 & 0.782 \\
 & Exact & 0.030 & 0.216 & 0.218 & 93.30 & 0.782 \\
 & PL-Cox & 0.040 & 0.222 & 0.226 & 93.20 & 0.787 \\
  &  GPL-Cox  & 0.027 & 0.220 & 0.222 & 89.80 & 0.737 \\
 & Ren & 0.062 & 0.237 & 0.245 & 90.40 & 0.785 \\
 & Tamano & 0.029 & 0.215 & 0.216 & 93.00 & 0.780 \\
\hline
100 & Breslow & 0.011 & 0.145 & 0.145 & 92.70 & 0.518 \\
 & Efron & 0.011 & 0.145 & 0.145 & 92.70 & 0.518 \\
 & Exact & 0.011 & 0.145 & 0.145 & 92.70 & 0.518 \\
 & PL-Cox & 0.014 & 0.147 & 0.148 & 91.50 & 0.517 \\
  &  GPL-Cox  & 0.004 & 0.147 & 0.147 & 88.40 & 0.479 \\
 & Ren & 0.026 & 0.152 & 0.154 & 90.60 & 0.515 \\
 & Tamano & 0.011 & 0.144 & 0.145 & 92.60 & 0.516 \\
\hline
300 & Breslow & 0.005 & 0.075 & 0.075 & 94.52 & 0.286 \\
 & Efron & 0.005 & 0.075 & 0.075 & 94.52 & 0.286 \\
 & Exact & 0.005 & 0.075 & 0.075 & 94.52 & 0.286 \\
 & PL-Cox & 0.007 & 0.076 & 0.076 & 93.80 & 0.286 \\
 & GPL-Cox & $-$0.011 & 0.075 & 0.076 & 88.37 & 0.252 \\
 & Ren & 0.012 & 0.077 & 0.078 & 92.76 & 0.283 \\
 & Tamano & 0.005 & 0.075 & 0.075 & 94.46 & 0.286 \\
\bottomrule
\end{tabular}

\begin{minipage}{0.96\textwidth}
\textit{Note:} Results for $n=30,50,100$ are based on 1,000 simulated datasets. The results for $n=300$, based on 10,000 simulated datasets, are reproduced from the main simulation study for reference.
\end{minipage}
\end{table}

\begin{table}[!htbp]
\centering
\caption{Sample-size sensitivity for $\beta_4$ under continuous survival times under the constant baseline hazard with rounding width $\Delta=0.1$. Reported metrics are empirical bias (Bias), the Monte Carlo standard deviation (SD) of the point estimates, root mean squared error (RMSE), coverage probability of the 95\% interval (CP), and average interval width (AW).}
\label{tab:sample_size_C01}
\begin{tabular}{r l r r r r r}
\toprule
$n$ & Method & Bias & SD & RMSE & CP & AW \\
\midrule
30 & Breslow & 0.052 & 0.329 & 0.333 & 93.20 & 1.108 \\
 & Efron & 0.055 & 0.331 & 0.335 & 93.20 & 1.109 \\
 & Exact & 0.057 & 0.334 & 0.338 & 92.90 & 1.117 \\
 & PL-Cox & 0.070 & 0.352 & 0.359 & 90.60 & 1.137 \\
  &  GPL-Cox  & 0.063 & 0.345 & 0.350 & 89.90 & 1.114 \\
 & Ren & 0.119 & 0.397 & 0.414 & 85.40 & 1.140 \\
 & Tamano & 0.047 & 0.324 & 0.327 & 92.50 & 1.104 \\
\hline
50 & Breslow & 0.028 & 0.214 & 0.216 & 93.50 & 0.781 \\
 & Efron & 0.030 & 0.216 & 0.218 & 93.40 & 0.782 \\
 & Exact & 0.032 & 0.217 & 0.219 & 93.30 & 0.786 \\
 & PL-Cox & 0.037 & 0.220 & 0.223 & 92.30 & 0.789 \\
  &  GPL-Cox  & 0.038 & 0.221 & 0.225 & 92.90 & 0.777 \\
 & Ren & 0.065 & 0.240 & 0.249 & 90.30 & 0.788 \\
 & Tamano & 0.026 & 0.213 & 0.215 & 93.40 & 0.779 \\
\hline
100 & Breslow & 0.009 & 0.144 & 0.144 & 93.20 & 0.517 \\
 & Efron & 0.011 & 0.145 & 0.145 & 93.00 & 0.517 \\
 & Exact & 0.013 & 0.146 & 0.146 & 93.00 & 0.521 \\
 & PL-Cox & 0.013 & 0.147 & 0.147 & 92.00 & 0.518 \\
  &  GPL-Cox  & 0.012 & 0.146 & 0.146 & 92.20 & 0.516 \\
 & Ren & 0.027 & 0.153 & 0.155 & 90.70 & 0.514 \\
 & Tamano & 0.009 & 0.144 & 0.144 & 93.10 & 0.516 \\
\hline
300 & Breslow & 0.003 & 0.074 & 0.074 & 94.67 & 0.286 \\
 & Efron & 0.005 & 0.075 & 0.075 & 94.48 & 0.286 \\
 & Exact & 0.006 & 0.075 & 0.075 & 94.45 & 0.288 \\
 & PL-Cox & 0.005 & 0.075 & 0.075 & 93.86 & 0.285 \\
 & GPL-Cox & 0.006 & 0.075 & 0.075 & 93.77 & 0.284 \\
 & Ren & 0.013 & 0.077 & 0.078 & 92.67 & 0.283 \\
 & Tamano & 0.003 & 0.074 & 0.074 & 94.70 & 0.286 \\
\bottomrule
\end{tabular}

\begin{minipage}{0.96\textwidth}
\textit{Note:} Results for $n=30,50,100$ are based on 1,000 simulated datasets. The results for $n=300$, based on 10,000 simulated datasets, are reproduced from the main simulation study for reference.
\end{minipage}
\end{table}

\begin{table}[!htbp]
\centering
\caption{Sample-size sensitivity for $\beta_4$ under discrete survival times under the constant baseline hazard with observation-grid width $u=1$. Reported metrics are empirical bias (Bias), the Monte Carlo standard deviation (SD) of the point estimates, root mean squared error (RMSE), coverage probability of the 95\% interval (CP), and average interval width (AW).}
\label{tab:sample_size_D1}
\begin{tabular}{r l r r r r r}
\toprule
$n$ & Method & Bias & SD & RMSE & CP & AW \\
\midrule
30 & Breslow & 0.029 & 0.356 & 0.357 & 93.50 & 1.212 \\
 & Efron & 0.030 & 0.358 & 0.359 & 93.60 & 1.213 \\
 & Exact & 0.032 & 0.359 & 0.361 & 93.50 & 1.219 \\
 & PL-Cox & 0.050 & 0.386 & 0.389 & 91.30 & 1.251 \\
  &  GPL-Cox  & 0.041 & 0.374 & 0.376 & 89.60 & 1.194 \\
 & Ren & 0.096 & 0.448 & 0.458 & 87.80 & 1.258 \\
 & Tamano & 0.023 & 0.350 & 0.351 & 93.30 & 1.207 \\
\hline
50 & Breslow & 0.026 & 0.240 & 0.242 & 93.80 & 0.856 \\
 & Efron & 0.027 & 0.242 & 0.243 & 93.70 & 0.857 \\
 & Exact & 0.028 & 0.243 & 0.244 & 93.80 & 0.861 \\
 & PL-Cox & 0.035 & 0.248 & 0.251 & 93.10 & 0.863 \\
  &  GPL-Cox  & 0.035 & 0.248 & 0.251 & 92.30 & 0.854 \\
 & Ren & 0.063 & 0.271 & 0.278 & 90.60 & 0.867 \\
 & Tamano & 0.023 & 0.239 & 0.240 & 93.70 & 0.853 \\
\hline
100 & Breslow & 0.011 & 0.152 & 0.153 & 94.60 & 0.565 \\
 & Efron & 0.012 & 0.153 & 0.153 & 94.50 & 0.565 \\
 & Exact & 0.013 & 0.154 & 0.154 & 94.40 & 0.568 \\
 & PL-Cox & 0.014 & 0.157 & 0.157 & 92.70 & 0.565 \\
  &  GPL-Cox  & 0.014 & 0.156 & 0.156 & 92.50 & 0.561 \\
 & Ren & 0.027 & 0.162 & 0.164 & 92.20 & 0.560 \\
 & Tamano & 0.010 & 0.152 & 0.152 & 94.10 & 0.564 \\
\hline
300 & Breslow & 0.002 & 0.081 & 0.081 & 94.72 & 0.313 \\
 & Efron & 0.004 & 0.081 & 0.082 & 94.61 & 0.313 \\
 & Exact & 0.005 & 0.082 & 0.082 & 94.63 & 0.314 \\
 & PL-Cox & 0.004 & 0.082 & 0.082 & 93.92 & 0.311 \\
 & GPL-Cox & 0.005 & 0.082 & 0.082 & 93.68 & 0.310 \\
 & Ren & 0.010 & 0.084 & 0.084 & 93.34 & 0.309 \\
 & Tamano & 0.002 & 0.081 & 0.081 & 94.59 & 0.312 \\
\bottomrule
\end{tabular}

\begin{minipage}{0.96\textwidth}
\textit{Note:} Results for $n=30,50,100$ are based on 1,000 simulated datasets. The results for $n=300$, based on 10,000 simulated datasets, are reproduced from the main simulation study for reference.
\end{minipage}
\end{table}

\begin{table}[!htbp]
\centering
\caption{Sample-size sensitivity for $\beta_4$ under discrete survival times under the constant baseline hazard with observation-grid width $u=7$. Reported metrics are empirical bias (Bias), the Monte Carlo standard deviation (SD) of the point estimates, root mean squared error (RMSE), coverage probability of the 95\% interval (CP), and average interval width (AW).}
\label{tab:sample_size_D7}
\begin{tabular}{r l r r r r r}
\toprule
$n$ & Method & Bias & SD & RMSE & CP & AW \\
\midrule
30 & Breslow & 0.060 & 0.374 & 0.379 & 93.40 & 1.241 \\
 & Efron & 0.072 & 0.387 & 0.393 & 92.20 & 1.247 \\
 & Exact & 0.086 & 0.403 & 0.412 & 92.30 & 1.292 \\
 & PL-Cox & 0.085 & 0.411 & 0.419 & 89.80 & 1.278 \\
  &  GPL-Cox  & 0.109 & 0.425 & 0.438 & 89.60 & 1.320 \\
 & Ren & 0.160 & 0.496 & 0.521 & 83.58 & 1.310 \\
 & Tamano & 0.052 & 0.366 & 0.369 & 92.90 & 1.233 \\
\hline
50 & Breslow & 0.034 & 0.239 & 0.241 & 94.30 & 0.861 \\
 & Efron & 0.045 & 0.247 & 0.251 & 93.60 & 0.865 \\
 & Exact & 0.055 & 0.256 & 0.261 & 94.00 & 0.892 \\
 & PL-Cox & 0.043 & 0.251 & 0.255 & 92.80 & 0.873 \\
  &  GPL-Cox  & 0.062 & 0.260 & 0.268 & 93.20 & 0.895 \\
 & Ren & 0.087 & 0.280 & 0.293 & 88.50 & 0.875 \\
 & Tamano & 0.032 & 0.238 & 0.240 & 93.60 & 0.859 \\
\hline
100 & Breslow & 0.006 & 0.148 & 0.148 & 95.00 & 0.571 \\
 & Efron & 0.015 & 0.153 & 0.154 & 94.70 & 0.572 \\
 & Exact & 0.023 & 0.157 & 0.159 & 95.00 & 0.589 \\
 & PL-Cox & 0.010 & 0.152 & 0.152 & 94.40 & 0.569 \\
  &  GPL-Cox  & 0.025 & 0.157 & 0.159 & 94.00 & 0.586 \\
 & Ren & 0.032 & 0.165 & 0.168 & 92.40 & 0.571 \\
 & Tamano & 0.005 & 0.148 & 0.148 & 95.20 & 0.569 \\
\hline
300 & Breslow & $-$0.005 & 0.080 & 0.080 & 94.87 & 0.315 \\
 & Efron & 0.004 & 0.082 & 0.082 & 94.42 & 0.315 \\
 & Exact & 0.012 & 0.085 & 0.086 & 94.47 & 0.325 \\
 & PL-Cox & $-$0.003 & 0.081 & 0.081 & 94.15 & 0.313 \\
 & GPL-Cox & 0.012 & 0.084 & 0.085 & 94.22 & 0.321 \\
 & Ren & 0.013 & 0.086 & 0.087 & 92.61 & 0.312 \\
 & Tamano & $-$0.005 & 0.080 & 0.080 & 94.82 & 0.314 \\
\bottomrule
\end{tabular}

\begin{minipage}{0.96\textwidth}
\textit{Note:} Results for $n=30,50,100$ are based on 1,000 simulated datasets. The results for $n=300$, based on 10,000 simulated datasets, are reproduced from the main simulation study for reference. For Ren at $n=30$, one numerical failure occurred; the reported summaries are based on the 999 successful fits.
\end{minipage}
\end{table}

\subsubsection{Sensitivity to the prior variance}

We examined the sensitivity of PL-Cox and GPL-Cox to the prior
variance by fitting independent $\text{N}(0,v)$ priors with
$v\in\{10,100,1000\}$ at $n=300$, using the four representative
settings defined in Section~D.3.1. For $v=10$ and $v=1000$, we
analyzed 1,000 simulated datasets per setting, reusing the same
datasets across the two variance choices. The default $v=100$ rows
are reproduced from the 10,000-replicate main simulation, and thus
have a different Monte Carlo sample size. 
The analyses with $v=10$ and $v=1000$ used identical datasets and implementation settings.
The reference results with $v=100$ were obtained from the main simulation study, for which covariates were standardized before fitting.
Thus, comparisons with the reference results are descriptive rather than strictly paired comparisons.
We report $\beta_4$ in Table~\ref{tab:prior_sensitivity}.

Across the variance choices, the bias, RMSE, and average credible
interval width changed little. Differences in empirical coverage
were modest relative to the Monte Carlo uncertainty of the
1,000-replicate sensitivity experiments. In the continuous-time
setting without rounding ($\Delta=0$), GPL-Cox exhibited coverage
below the nominal 95\% level under all three priors; changing the
prior variance did not remove this undercoverage. These findings
indicate limited sensitivity to prior variance over the examined
range, rather than uniformly satisfactory calibration in every
scenario.

\begin{table}[!htbp]
\centering
\caption{Sensitivity of PL-Cox and GPL-Cox to the prior variance for $\beta_4$ ($n=300$).}
\label{tab:prior_sensitivity}
\begin{tabular}{l r l r r r r r}
\toprule
Setting & $v$ & Method & Bias & SD & RMSE & CP & AW \\
\midrule
Continuous, $\Delta=0$ & 10 & PL-Cox & 0.007 & 0.076 & 0.077 & 94.70 & 0.287 \\
 &  & GPL-Cox & $-$0.009 & 0.075 & 0.076 & 89.30 & 0.255 \\
 & 100 & PL-Cox & 0.007 & 0.076 & 0.076 & 93.80 & 0.286 \\
 &  & GPL-Cox & $-$0.011 & 0.075 & 0.076 & 88.37 & 0.252 \\
 & 1000 & PL-Cox & 0.007 & 0.076 & 0.077 & 94.80 & 0.286 \\
 &  & GPL-Cox & $-$0.011 & 0.077 & 0.078 & 88.10 & 0.253 \\
\midrule
Continuous, $\Delta=0.1$ & 10 & PL-Cox & 0.005 & 0.076 & 0.076 & 94.80 & 0.286 \\
 &  & GPL-Cox & 0.005 & 0.075 & 0.075 & 94.70 & 0.284 \\
 & 100 & PL-Cox & 0.005 & 0.075 & 0.075 & 93.86 & 0.285 \\
 &  & GPL-Cox & 0.006 & 0.075 & 0.075 & 93.77 & 0.284 \\
 & 1000 & PL-Cox & 0.005 & 0.076 & 0.076 & 94.10 & 0.284 \\
 &  & GPL-Cox & 0.006 & 0.076 & 0.076 & 94.60 & 0.285 \\
\midrule
Discrete, $u=1$ & 10 & PL-Cox & 0.004 & 0.079 & 0.079 & 95.20 & 0.311 \\
 &  & GPL-Cox & 0.003 & 0.079 & 0.079 & 95.10 & 0.308 \\
 & 100 & PL-Cox & 0.004 & 0.082 & 0.082 & 93.92 & 0.311 \\
 &  & GPL-Cox & 0.005 & 0.082 & 0.082 & 93.68 & 0.310 \\
 & 1000 & PL-Cox & 0.003 & 0.080 & 0.080 & 94.90 & 0.311 \\
 &  & GPL-Cox & 0.004 & 0.079 & 0.079 & 94.40 & 0.309 \\
\midrule
Discrete, $u=7$ & 10 & PL-Cox & 0.000 & 0.080 & 0.080 & 95.00 & 0.314 \\
 &  & GPL-Cox & 0.014 & 0.083 & 0.085 & 94.10 & 0.322 \\
 & 100 & PL-Cox & $-$0.003 & 0.081 & 0.081 & 94.15 & 0.313 \\
 &  & GPL-Cox & 0.012 & 0.084 & 0.085 & 94.22 & 0.321 \\
 & 1000 & PL-Cox & 0.000 & 0.080 & 0.080 & 95.40 & 0.314 \\
 &  & GPL-Cox & 0.014 & 0.084 & 0.085 & 94.40 & 0.321 \\
\bottomrule
\end{tabular}
\begin{minipage}{0.96\textwidth}
\scriptsize
\textit{Note:} Independent $\text{N}(0,v)$ priors were assigned to the regression coefficients (and the GPL-Cox intercept). 
CP denotes the empirical coverage probability (\%) of the 95\% credible interval, and AW is the average width of the 95\% credible interval. 
The $v=10$ and $v=1000$ rows are based on 1,000 common simulated datasets per setting and 1,000 successful fits per method; both prior choices were applied to the same datasets. 
The $v=100$ reference rows are taken from the main simulation study (10,000 datasets), rather than recomputed in experiment 007. No fitting errors were recorded in the additional runs.
\end{minipage}
\end{table}

\clearpage

\section{Detailed simulation settings and full results}
\subsection{Simulation settings\label{app:sim_settings}}
\paragraph{Continuous survival time scenarios.}
We generated survival data from a proportional hazards model with a Weibull baseline hazard.
For each simulated dataset, the covariates consisted of four independent variables generated from the standard normal distribution.
The regression coefficients were set to $\bld{\beta}^\top_{\text{true}}=(0.10,0.05,-0.15,0.30)$.
Right censoring was introduced independently of the failure time.

Event times $T_i^\ast$ were generated from a Weibull distribution with shape parameter $a$ and scale parameter $b_i$,
\[
T_i^\ast \sim \mathrm{Weibull}(a, b_i), \quad 
b_i = b \exp\left(- \frac{\boldsymbol{x}_i^\top \boldsymbol{\beta}}{a}\right),
\]
so that the proportional hazards structure is satisfied.
Three baseline hazard scenarios were considered:
\begin{itemize}
\item Scenario 1: $a = 1.0$, $b = 10$ (constant hazard),
\item Scenario 2: $a = 0.7$, $b = 8.0$ (decreasing hazard),
\item Scenario 3: $a = 2.0$, $b = 12.0$ (increasing hazard).
\end{itemize}
Independent censoring times were generated as $C_i \sim \mathrm{Unif}(0.5, 30)$, and the observed data were defined as $T_i = \min(T_i^\ast, C_i), \delta_i = \mathds{1}(T_i^\ast \le C_i)$.

To investigate the effect of tied event times, the observed survival times were coarsened using a grid of width $\Delta$.
Specifically, for $\Delta > 0$, the observed time was defined as $T_i^{(\Delta)} = \Delta \cdot \mathrm{round}\left( \frac{T_i}{\Delta} \right)$.
When $\Delta = 0$, no coarsening was applied and the data correspond to continuous event times.
We considered the following coarsening levels: $\Delta \in \{0 \ (\text{no rounding}),\ 0.01,\ 0.1,\ 0.5 \}$.
Any observed time rounded to zero was replaced by 0.05
to avoid zero follow-up times in subsequent model fitting.

\paragraph{Discrete survival time scenarios.}
We generated survival data from a discrete-time survival model with a logistic hazard formulation.
For each simulated dataset, the covariates consisted of four independent variables generated from the standard normal distribution.
The regression coefficients were set to $\bld{\beta}^\top_{\text{true}}=(0.10,0.05,-0.15,0.30)$.

Let $t = 1, \dots, T_{\max}$ denote discrete time points with $T_{\max}=300$.
Conditional on being at risk at time $t$, the event probability was defined as
\[
\Pr(T_i^\ast = t \mid T_i^\ast \ge t, \boldsymbol{x}_i)
= \mathrm{logit}^{-1}\left( \alpha_t + \boldsymbol{x}_i^\top \boldsymbol{\beta} \right),
\]
where $\{\alpha_t\}$ specifies the baseline hazard over time.

We considered the following baseline hazard scenarios:
\begin{itemize}
\item Scenario 1: constant hazard, $\alpha_t = \alpha_0$,
\item Scenario 2: decreasing hazard, $\alpha_t = \alpha_0 + 1.2 \left(1 - \frac{t-1}{T_{\max}-1}\right)$,
\item Scenario 3: increasing hazard, $\alpha_t = \alpha_0 + 1.2 \frac{t-1}{T_{\max}-1}$,
\end{itemize}
The intercept parameter was set to $\alpha_0=-5.0$ to control the overall event rate.

Event times were generated sequentially using Bernoulli trials at each time point until failure or censoring occurred.
Right censoring times were independently generated from a discrete uniform distribution over $\{1,\dots,T_{\max}\}$.

The generated times were coarsened to a grid of width $u$.
Let $\widetilde{T}_i$ and $\delta_i^{(0)}$ denote the observed time and event indicator before coarsening.
The coarsened event and censoring times were defined as
\[
T_i^{\text{event}}
=
\min\left\{
u\left\lceil \frac{\widetilde{T}_i}{u}\right\rceil,
T_{\max}
\right\},
\qquad
C_i^{\text{obs}}
=
\min\left\{
u\left\lfloor \frac{C_i}{u}\right\rfloor,
T_{\max}
\right\}.
\]
The observed time and event indicator after coarsening were
\[
T_i^{\text{obs}}
=
\min\left(T_i^{\text{event}},C_i^{\text{obs}}\right),
\qquad
\delta_i^{(u)}
=
\delta_i^{(0)}
\mathds{1}\left(
T_i^{\text{event}}\le C_i^{\text{obs}}
\right).
\]
We considered the following coarsening levels: $u \in \{1, 7, 14, 28\}$.
Thus, $u=1$ retains the originally generated discrete event times, which may already contain ties, whereas larger values of $u$ group multiple discrete time points into wider observation intervals.

\paragraph{Continuous survival time with non-proportional hazard scenarios.}
We generated survival data from a non-proportional hazards model, which violates the proportional hazards assumption.
For each simulated dataset, the covariates consisted of four independent variables generated from the standard normal distribution.
The regression coefficients were set to $\bld{\beta}^\top_{\text{true}}=(0.10,0.05,-0.15,0.30)$.
Right censoring was introduced independently of the failure time.

Event times $T_i^\ast$ were generated from a log-normal
accelerated failure time model,
\[
\log(T_i^\ast)
\sim
\mathrm{N}\left(
\log\mu-\boldsymbol{x}_i^\top\boldsymbol{\beta},
\sigma^2
\right).
\]
Four event-time distribution settings were considered:
\begin{itemize}
\item Scenario 1: $\mu = 60$, $\sigma = 0.6$ (baseline),
\item Scenario 2: $\mu = 35$, $\sigma = 0.6$ (earlier events),
\item Scenario 3: $\mu = 85$, $\sigma = 0.6$ (later events),
\item Scenario 4: $\mu = 60$, $\sigma = 1.0$ (heavier tail).
\end{itemize}
Independent censoring times were generated as $C_i \sim \mathrm{Unif}(0, 300)$, and the observed data were defined as $T_i = \min(T_i^\ast, C_i), \delta_i = \mathds{1}(T_i^\ast \le C_i)$.

To investigate the effect of tied event times, the observed survival times were coarsened using the same procedure as in the discrete-time setting, where event times were rounded up and censoring times were rounded down to the observation grid.

\paragraph{Analyses.}
We used the \texttt{BayesPLCox} package, which was developed for the implementation of PL-Cox and GPL-Cox models.
The Efron, Breslow, and Exact methods were implemented using the \texttt{coxph} function in the R package \texttt{survival}.
The Ren method was implemented using the R code provided in the supplementary materials of \citet{renCoxPolyaGammaAlgorithmFlexible2025}.
For the Tamano method, we used the authors' official Python implementation of GS4Cox, available at \url{https://github.com/shutech2001/GS4Cox}.
Posterior inference was carried out using Markov chain Monte Carlo with 3,000 iterations, discarding the first 1,000 iterations as burn-in.

\paragraph{Performance measures and others.}
We calculated the bias of the posterior mean of $\bld{\beta}$.
In addition, we calculated the Monte Carlo standard deviation of the point estimates, the root mean square error, coverage probability of 95\% posterior credible interval, and average length.
In the non-proportional hazard scenarios, given that inference under the Cox model was conducted under model misspecification, we simulated data for 500{,}000 individuals from the same data-generating mechanism without rounding.
We then applied the Efron method and considered the resulting point estimate as a pseudo-true parameter corresponding to the Cox model.
Data generation and the R-based analyses were performed using \texttt{R} version 4.4.2 for Linux, while the Tamano method was implemented in Python.
Computations were carried out on the supercomputer system of the Institute of Statistical Mathematics, Tokyo, Japan.
We generated $10{,}000$ replicated datasets for each scenario.

\newpage

\subsection{Full results}

The tables below report results for $\beta_4$. 
Complete simulation results for all four regression coefficients are provided in the accompanying Supplementary Excel file (\texttt{Simulation\_results\_all\_coefficients.xlsx}).

\subsubsection{Continuous survival time scenarios.}

\begin{table}[htbp]
\centering
\caption{Simulation results for the regression coefficient $\beta_4$ under continuous survival times generated from a weibull model with decreasing hazard ($n=300$). The rounding levels indicate the coarsening width applied to the observed times, where ``None'' corresponds to no rounding. Reported metrics are empirical bias (Bias), the Monte Carlo standard deviation (SD) of the point estimates, root mean squared error (RMSE), coverage probability of the 95\% interval (CP), and average interval width (AW).}
\begin{tabular}{ccrcccc}
\hline
Sce & Method & \multicolumn{1}{c}{Bias} & SD & RMSE & CP & AW \\
\hline
None & Breslow & 0.005 & 0.072 & 0.072 & 94.81 & 0.280 \\
  & Efron & 0.005 & 0.072 & 0.072 & 94.81 & 0.280 \\
  & Exact & 0.005 & 0.072 & 0.072 & 94.81 & 0.280 \\
  & PL-Cox & 0.007 & 0.073 & 0.073 & 94.19 & 0.279 \\
  & GPL-Cox & $-$0.011 & 0.073 & 0.074 & 88.62 & 0.245 \\
  & Ren & 0.013 & 0.074 & 0.075 & 93.51 & 0.276 \\
  & Tamano & 0.005 & 0.072 & 0.072 & 94.66 & 0.279 \\
\hline
0.01 & Breslow & 0.004 & 0.072 & 0.072 & 94.87 & 0.280 \\
  & Efron & 0.005 & 0.072 & 0.072 & 94.83 & 0.280 \\
  & Exact & 0.005 & 0.072 & 0.072 & 94.84 & 0.280 \\
  & PL-Cox & 0.006 & 0.073 & 0.073 & 94.37 & 0.279 \\
  & GPL-Cox & 0.012 & 0.076 & 0.076 & 91.55 & 0.271 \\
  & Ren & 0.012 & 0.074 & 0.075 & 93.19 & 0.276 \\
  & Tamano & 0.004 & 0.072 & 0.072 & 94.64 & 0.279 \\
\hline
0.1 & Breslow & 0.002 & 0.071 & 0.071 & 95.06 & 0.279 \\
  & Efron & 0.005 & 0.072 & 0.072 & 94.77 & 0.280 \\
  & Exact & 0.007 & 0.072 & 0.073 & 94.75 & 0.282 \\
  & PL-Cox & 0.004 & 0.072 & 0.072 & 94.36 & 0.279 \\
  & GPL-Cox & 0.035 & 0.080 & 0.087 & 89.75 & 0.281 \\
  & Ren & 0.013 & 0.074 & 0.075 & 93.30 & 0.277 \\
  & Tamano & 0.002 & 0.071 & 0.071 & 94.97 & 0.279 \\
\hline
0.5 & Breslow & $-$0.007 & 0.068 & 0.069 & 95.65 & 0.279 \\
  & Efron & 0.004 & 0.072 & 0.072 & 94.84 & 0.279 \\
  & Exact & 0.016 & 0.075 & 0.077 & 94.40 & 0.291 \\
  & PL-Cox & $-$0.006 & 0.069 & 0.070 & 95.05 & 0.279 \\
  & GPL-Cox & 0.056 & 0.084 & 0.101 & 85.35 & 0.289 \\
  & Ren & 0.017 & 0.075 & 0.077 & 92.78 & 0.276 \\
  & Tamano & $-$0.007 & 0.068 & 0.069 & 95.62 & 0.278 \\
\hline
\end{tabular}
\par\smallskip{\footnotesize Note: Summaries are based on successful fits only.
Ren had numerical failures in 10 of the 40 settings; the numbers of successful fits ranged from 9,906 to 10,000. Setting-specific counts are provided in the accompanying Supplementary Excel file.}
\end{table}

\begin{table}[htbp]
\centering
\caption{Simulation results for the regression coefficient $\beta_4$ under continuous survival times generated from a weibull model with increasing  hazard ($n=300$). The rounding levels indicate the coarsening width applied to the observed times, where ``None'' corresponds to no rounding. Reported metrics are empirical bias (Bias), the Monte Carlo standard deviation (SD) of the point estimates, root mean squared error (RMSE), coverage probability of the 95\% interval (CP), and average interval width (AW).}
\begin{tabular}{ccrcccc}
\hline
Sce & Method & \multicolumn{1}{c}{Bias} & SD & RMSE & CP & AW \\
\hline
None & Breslow & 0.004 & 0.078 & 0.078 & 94.63 & 0.297 \\
  & Efron & 0.004 & 0.078 & 0.078 & 94.63 & 0.297 \\
  & Exact & 0.004 & 0.078 & 0.078 & 94.63 & 0.297 \\
  & PL-Cox & 0.005 & 0.078 & 0.079 & 94.00 & 0.297 \\
  & GPL-Cox & $-$0.012 & 0.078 & 0.079 & 88.41 & 0.261 \\
  & Ren & 0.011 & 0.080 & 0.081 & 92.89 & 0.292 \\
  & Tamano & 0.004 & 0.078 & 0.078 & 94.57 & 0.296 \\
\hline
0.01 & Breslow & 0.004 & 0.078 & 0.078 & 94.64 & 0.296 \\
  & Efron & 0.004 & 0.078 & 0.078 & 94.65 & 0.296 \\
  & Exact & 0.004 & 0.078 & 0.078 & 94.64 & 0.297 \\
  & PL-Cox & 0.005 & 0.078 & 0.079 & 94.00 & 0.296 \\
  & GPL-Cox & $-$0.005 & 0.078 & 0.078 & 90.97 & 0.277 \\
  & Ren & 0.011 & 0.080 & 0.081 & 92.98 & 0.293 \\
  & Tamano & 0.004 & 0.078 & 0.078 & 94.54 & 0.296 \\
\hline
0.1 & Breslow & 0.002 & 0.077 & 0.077 & 94.82 & 0.296 \\
  & Efron & 0.004 & 0.078 & 0.078 & 94.65 & 0.296 \\
  & Exact & 0.006 & 0.078 & 0.078 & 94.68 & 0.298 \\
  & PL-Cox & 0.003 & 0.078 & 0.078 & 94.02 & 0.296 \\
  & GPL-Cox & $-$0.028 & 0.069 & 0.075 & 94.61 & 0.292 \\
  & Ren & 0.012 & 0.080 & 0.081 & 92.78 & 0.292 \\
  & Tamano & 0.002 & 0.077 & 0.077 & 94.66 & 0.296 \\
\hline
0.5 & Breslow & $-$0.006 & 0.075 & 0.075 & 95.25 & 0.296 \\
  & Efron & 0.004 & 0.078 & 0.078 & 94.73 & 0.296 \\
  & Exact & 0.013 & 0.080 & 0.081 & 94.55 & 0.306 \\
  & PL-Cox & $-$0.004 & 0.076 & 0.076 & 94.70 & 0.295 \\
  & GPL-Cox & $-$0.074 & 0.059 & 0.095 & 88.56 & 0.294 \\
  & Ren & 0.016 & 0.081 & 0.082 & 92.60 & 0.293 \\
  & Tamano & $-$0.006 & 0.075 & 0.075 & 95.12 & 0.295 \\
\hline
\end{tabular}
\par\smallskip{\footnotesize Note: Summaries are based on successful fits only.
Ren had numerical failures in 10 of the 40 settings; the numbers of successful fits ranged from 9,906 to 10,000. Setting-specific counts are provided in the accompanying Supplementary Excel file.}
\end{table}

\newpage

\subsubsection{Discrete survival time scenarios.}

\begin{table}[htbp]
\centering
\caption{Simulation results for the regression coefficient $\beta_4$ under discrete survival times generated from a logistic hazard model with decreasing hazard ($n=300$). The coarsening levels indicate the observation grid width applied to the event and censoring times. Reported metrics are empirical bias (Bias), the Monte Carlo standard deviation (SD) of the point estimates, root mean squared error (RMSE), coverage probability of the 95\% interval (CP), and average interval width (AW).}
\begin{tabular}{ccrcccc}
\hline
Sce & Method & \multicolumn{1}{c}{Bias} & SD & RMSE & CP & AW \\
\hline
1 & Breslow & $-$0.003 & 0.068 & 0.068 & 95.26 & 0.265 \\
  & Efron & 0.001 & 0.069 & 0.069 & 95.12 & 0.265 \\
  & Exact & 0.004 & 0.070 & 0.070 & 95.13 & 0.268 \\
  & PL-Cox & $-$0.001 & 0.069 & 0.069 & 94.70 & 0.265 \\
  & GPL-Cox & 0.013 & 0.071 & 0.073 & 93.59 & 0.265 \\
  & Ren & 0.012 & 0.072 & 0.073 & 93.06 & 0.262 \\
  & Tamano & $-$0.003 & 0.068 & 0.068 & 95.12 & 0.264 \\
\hline
7 & Breslow & $-$0.022 & 0.063 & 0.067 & 95.50 & 0.266 \\
  & Efron & $-$0.000 & 0.069 & 0.069 & 95.16 & 0.267 \\
  & Exact & 0.024 & 0.075 & 0.078 & 94.26 & 0.290 \\
  & PL-Cox & $-$0.021 & 0.064 & 0.067 & 95.06 & 0.265 \\
  & GPL-Cox & 0.047 & 0.079 & 0.092 & 88.24 & 0.285 \\
  & Ren & 0.019 & 0.074 & 0.076 & 91.76 & 0.264 \\
  & Tamano & $-$0.022 & 0.063 & 0.067 & 95.39 & 0.265 \\
\hline
14 & Breslow & $-$0.043 & 0.059 & 0.073 & 93.44 & 0.267 \\
  & Efron & $-$0.002 & 0.069 & 0.069 & 95.12 & 0.268 \\
  & Exact & 0.047 & 0.082 & 0.094 & 91.64 & 0.316 \\
  & PL-Cox & $-$0.042 & 0.060 & 0.073 & 93.13 & 0.267 \\
  & GPL-Cox & 0.077 & 0.087 & 0.116 & 81.59 & 0.311 \\
  & Ren & 0.025 & 0.076 & 0.080 & 90.63 & 0.266 \\
  & Tamano & $-$0.043 & 0.059 & 0.073 & 93.53 & 0.266 \\
\hline
28 & Breslow & $-$0.080 & 0.052 & 0.096 & 84.17 & 0.269 \\
  & Efron & $-$0.011 & 0.068 & 0.069 & 95.04 & 0.272 \\
  & Exact & 0.094 & 0.096 & 0.135 & 84.42 & 0.372 \\
  & PL-Cox & $-$0.080 & 0.052 & 0.096 & 83.25 & 0.269 \\
  & GPL-Cox & 0.130 & 0.103 & 0.166 & 69.81 & 0.368 \\
  & Ren & 0.033 & 0.082 & 0.088 & 88.11 & 0.271 \\
  & Tamano & $-$0.080 & 0.052 & 0.096 & 84.04 & 0.269 \\
\hline
\end{tabular}
\par\smallskip{\footnotesize Note: Summaries are based on successful fits only.
Ren had numerical failures in 10 of the 40 settings; the numbers of successful fits ranged from 9,906 to 10,000. Setting-specific counts are provided in the accompanying Supplementary Excel file.}
\end{table}

\begin{table}[htbp]
\centering
\caption{Simulation results for the regression coefficient $\beta_4$ under discrete survival times generated from a logistic hazard model with increasing hazard ($n=300$). The coarsening levels indicate the observation grid width applied to the event and censoring times. Reported metrics are empirical bias (Bias), the Monte Carlo standard deviation (SD) of the point estimates, root mean squared error (RMSE), coverage probability of the 95\% interval (CP), and average interval width (AW).}
\begin{tabular}{ccrcccc}
\hline
Sce & Method & \multicolumn{1}{c}{Bias} & SD & RMSE & CP & AW \\
\hline
1 & Breslow & 0.003 & 0.076 & 0.076 & 94.87 & 0.294 \\
  & Efron & 0.004 & 0.076 & 0.077 & 94.81 & 0.294 \\
  & Exact & 0.006 & 0.077 & 0.077 & 94.82 & 0.296 \\
  & PL-Cox & 0.005 & 0.077 & 0.077 & 94.11 & 0.294 \\
  & GPL-Cox & $-$0.003 & 0.075 & 0.075 & 94.49 & 0.291 \\
  & Ren & 0.011 & 0.079 & 0.080 & 93.21 & 0.290 \\
  & Tamano & 0.003 & 0.076 & 0.076 & 94.74 & 0.293 \\
\hline
7 & Breslow & $-$0.006 & 0.074 & 0.074 & 95.38 & 0.296 \\
  & Efron & 0.004 & 0.077 & 0.077 & 94.84 & 0.296 \\
  & Exact & 0.015 & 0.080 & 0.081 & 94.43 & 0.308 \\
  & PL-Cox & $-$0.004 & 0.075 & 0.075 & 94.77 & 0.295 \\
  & GPL-Cox & $-$0.012 & 0.073 & 0.074 & 95.90 & 0.302 \\
  & Ren & 0.015 & 0.080 & 0.082 & 92.72 & 0.293 \\
  & Tamano & $-$0.006 & 0.074 & 0.074 & 95.37 & 0.295 \\
\hline
14 & Breslow & $-$0.017 & 0.072 & 0.074 & 95.60 & 0.298 \\
  & Efron & 0.004 & 0.077 & 0.077 & 94.91 & 0.299 \\
  & Exact & 0.025 & 0.083 & 0.087 & 94.19 & 0.322 \\
  & PL-Cox & $-$0.015 & 0.073 & 0.074 & 95.27 & 0.297 \\
  & GPL-Cox & $-$0.006 & 0.075 & 0.075 & 96.42 & 0.315 \\
  & Ren & 0.018 & 0.082 & 0.084 & 92.39 & 0.296 \\
  & Tamano & $-$0.017 & 0.072 & 0.074 & 95.51 & 0.297 \\
\hline
28 & Breslow & $-$0.036 & 0.068 & 0.077 & 95.00 & 0.302 \\
  & Efron & 0.001 & 0.078 & 0.078 & 94.89 & 0.304 \\
  & Exact & 0.047 & 0.091 & 0.102 & 92.56 & 0.353 \\
  & PL-Cox & $-$0.035 & 0.069 & 0.077 & 94.47 & 0.301 \\
  & GPL-Cox & 0.013 & 0.082 & 0.083 & 96.29 & 0.346 \\
  & Ren & 0.023 & 0.086 & 0.089 & 90.99 & 0.301 \\
  & Tamano & $-$0.036 & 0.068 & 0.077 & 94.89 & 0.302 \\
\hline
\end{tabular}
\par\smallskip{\footnotesize Note: Summaries are based on successful fits only.
Ren had numerical failures in 10 of the 40 settings; the numbers of successful fits ranged from 9,906 to 10,000. Setting-specific counts are provided in the accompanying Supplementary Excel file.}
\end{table}

\newpage

\subsubsection{Continuous survival time with non-proportional hazard scenarios}

\begin{table}[htbp]
\centering
\caption{Simulation results for the regression coefficient $\beta_4$ under continuous survival times generated from non-proportional hazard model with baseline scenario ($n=300$). The coarsening levels indicate the observation grid width applied to the event and censoring times. Reported metrics are empirical bias (Bias), the Monte Carlo standard deviation (SD) of the point estimates, root mean squared error (RMSE), coverage probability of the 95\% interval (CP), and average interval width (AW).}
\begin{tabular}{ccrcccc}
\hline
Sce & Method & \multicolumn{1}{c}{Bias} & SD & RMSE & CP & AW \\
\hline
1 & Breslow & 0.011 & 0.081 & 0.082 & 92.61 & 0.289 \\
  & Efron & 0.016 & 0.082 & 0.083 & 92.02 & 0.289 \\
  & Exact & 0.022 & 0.083 & 0.086 & 91.87 & 0.293 \\
  & PL-Cox & 0.015 & 0.081 & 0.083 & 91.86 & 0.289 \\
  & GPL-Cox & $-$0.048 & 0.068 & 0.083 & 90.41 & 0.281 \\
  & Ren & 0.026 & 0.085 & 0.089 & 89.16 & 0.284 \\
  & Tamano & 0.010 & 0.081 & 0.081 & 92.53 & 0.288 \\
\hline
7 & Breslow & $-$0.016 & 0.075 & 0.077 & 93.64 & 0.289 \\
  & Efron & 0.014 & 0.082 & 0.083 & 92.11 & 0.290 \\
  & Exact & 0.058 & 0.089 & 0.106 & 87.90 & 0.318 \\
  & PL-Cox & $-$0.013 & 0.076 & 0.077 & 93.47 & 0.289 \\
  & GPL-Cox & $-$0.097 & 0.061 & 0.115 & 77.89 & 0.293 \\
  & Ren & 0.033 & 0.088 & 0.094 & 87.82 & 0.286 \\
  & Tamano & $-$0.017 & 0.075 & 0.077 & 93.50 & 0.288 \\
\hline
14 & Breslow & $-$0.047 & 0.070 & 0.084 & 90.79 & 0.289 \\
  & Efron & 0.009 & 0.082 & 0.082 & 92.55 & 0.292 \\
  & Exact & 0.099 & 0.096 & 0.138 & 80.05 & 0.348 \\
  & PL-Cox & $-$0.044 & 0.071 & 0.083 & 90.73 & 0.289 \\
  & GPL-Cox & $-$0.111 & 0.058 & 0.125 & 76.13 & 0.308 \\
  & Ren & 0.038 & 0.090 & 0.098 & 86.21 & 0.288 \\
  & Tamano & $-$0.047 & 0.070 & 0.084 & 90.63 & 0.289 \\
\hline
28 & Breslow & $-$0.098 & 0.062 & 0.116 & 76.08 & 0.291 \\
  & Efron & $-$0.007 & 0.080 & 0.080 & 93.19 & 0.294 \\
  & Exact & 0.169 & 0.109 & 0.201 & 63.41 & 0.403 \\
  & PL-Cox & $-$0.096 & 0.063 & 0.115 & 76.42 & 0.291 \\
  & GPL-Cox & $-$0.057 & 0.069 & 0.090 & 95.24 & 0.351 \\
  & Ren & 0.039 & 0.095 & 0.103 & 84.83 & 0.292 \\
  & Tamano & $-$0.099 & 0.062 & 0.116 & 75.92 & 0.290 \\
\hline
\end{tabular}
\par\smallskip{\footnotesize Note: Summaries are based on successful fits only.
Ren had numerical failures in 10 of the 40 settings; the numbers of successful fits ranged from 9,906 to 10,000. Setting-specific counts are provided in the accompanying Supplementary Excel file.}
\end{table}

\begin{table}[htbp]
\centering
\caption{Simulation results for the regression coefficient $\beta_4$ under continuous survival times generated from non-proportional hazard model with earlier events scenario ($n=300$). The coarsening levels indicate the observation grid width applied to the event and censoring times. Reported metrics are empirical bias (Bias), the Monte Carlo standard deviation (SD) of the point estimates, root mean squared error (RMSE), coverage probability of the 95\% interval (CP), and average interval width (AW).}
\begin{tabular}{ccrcccc}
\hline
Sce & Method & \multicolumn{1}{c}{Bias} & SD & RMSE & CP & AW \\
\hline
1 & Breslow & 0.008 & 0.076 & 0.076 & 92.78 & 0.271 \\
  & Efron & 0.015 & 0.077 & 0.079 & 91.79 & 0.271 \\
  & Exact & 0.026 & 0.079 & 0.083 & 91.43 & 0.277 \\
  & PL-Cox & 0.011 & 0.076 & 0.077 & 92.56 & 0.272 \\
  & GPL-Cox & $-$0.072 & 0.060 & 0.093 & 82.90 & 0.263 \\
  & Ren & 0.026 & 0.081 & 0.085 & 88.48 & 0.267 \\
  & Tamano & 0.007 & 0.075 & 0.076 & 92.70 & 0.270 \\
\hline
7 & Breslow & $-$0.038 & 0.067 & 0.077 & 91.68 & 0.270 \\
  & Efron & 0.010 & 0.077 & 0.077 & 92.32 & 0.272 \\
  & Exact & 0.083 & 0.088 & 0.121 & 81.43 & 0.316 \\
  & PL-Cox & $-$0.035 & 0.067 & 0.076 & 91.84 & 0.270 \\
  & GPL-Cox & $-$0.118 & 0.053 & 0.129 & 65.65 & 0.281 \\
  & Ren & 0.036 & 0.084 & 0.092 & 86.20 & 0.268 \\
  & Tamano & $-$0.038 & 0.067 & 0.077 & 91.53 & 0.269 \\
\hline
14 & Breslow & $-$0.084 & 0.059 & 0.103 & 79.15 & 0.269 \\
  & Efron & $-$0.003 & 0.075 & 0.075 & 93.16 & 0.272 \\
  & Exact & 0.144 & 0.098 & 0.174 & 65.40 & 0.358 \\
  & PL-Cox & $-$0.082 & 0.060 & 0.102 & 79.24 & 0.269 \\
  & GPL-Cox & $-$0.084 & 0.058 & 0.102 & 88.78 & 0.312 \\
  & Ren & 0.040 & 0.088 & 0.096 & 84.50 & 0.269 \\
  & Tamano & $-$0.084 & 0.059 & 0.103 & 79.05 & 0.268 \\
\hline
28 & Breslow & $-$0.163 & 0.048 & 0.170 & 27.63 & 0.268 \\
  & Efron & $-$0.043 & 0.069 & 0.081 & 89.74 & 0.271 \\
  & Exact & 0.227 & 0.116 & 0.255 & 47.56 & 0.438 \\
  & PL-Cox & $-$0.162 & 0.049 & 0.170 & 28.30 & 0.267 \\
  & GPL-Cox & 0.029 & 0.080 & 0.085 & 97.86 & 0.390 \\
  & Ren & 0.029 & 0.094 & 0.099 & 83.11 & 0.270 \\
  & Tamano & $-$0.163 & 0.048 & 0.170 & 27.31 & 0.267 \\
\hline
\end{tabular}
\par\smallskip{\footnotesize Note: Summaries are based on successful fits only.
Ren had numerical failures in 10 of the 40 settings; the numbers of successful fits ranged from 9,906 to 10,000. Setting-specific counts are provided in the accompanying Supplementary Excel file.}
\end{table}

\begin{table}[htbp]
\centering
\caption{Simulation results for the regression coefficient $\beta_4$ under continuous survival times generated from non-proportional hazard model with later events scenario ($n=300$). The coarsening levels indicate the observation grid width applied to the event and censoring times. Reported metrics are empirical bias (Bias), the Monte Carlo standard deviation (SD) of the point estimates, root mean squared error (RMSE), coverage probability of the 95\% interval (CP), and average interval width (AW).}
\begin{tabular}{ccrcccc}
\hline
Sce & Method & \multicolumn{1}{c}{Bias} & SD & RMSE & CP & AW \\
\hline
1 & Breslow & 0.012 & 0.087 & 0.088 & 92.36 & 0.309 \\
  & Efron & 0.015 & 0.088 & 0.090 & 91.96 & 0.309 \\
  & Exact & 0.020 & 0.089 & 0.091 & 91.94 & 0.312 \\
  & PL-Cox & 0.017 & 0.088 & 0.090 & 91.69 & 0.309 \\
  & GPL-Cox & $-$0.033 & 0.077 & 0.084 & 92.10 & 0.301 \\
  & Ren & 0.025 & 0.091 & 0.095 & 89.69 & 0.304 \\
  & Tamano & 0.011 & 0.087 & 0.088 & 92.33 & 0.308 \\
\hline
7 & Breslow & $-$0.008 & 0.083 & 0.084 & 93.48 & 0.310 \\
  & Efron & 0.015 & 0.089 & 0.090 & 92.12 & 0.311 \\
  & Exact & 0.047 & 0.094 & 0.105 & 89.83 & 0.332 \\
  & PL-Cox & $-$0.003 & 0.085 & 0.085 & 92.70 & 0.310 \\
  & GPL-Cox & $-$0.092 & 0.067 & 0.114 & 81.44 & 0.310 \\
  & Ren & 0.031 & 0.093 & 0.098 & 88.78 & 0.307 \\
  & Tamano & $-$0.008 & 0.083 & 0.084 & 93.43 & 0.309 \\
\hline
14 & Breslow & $-$0.030 & 0.079 & 0.085 & 92.85 & 0.311 \\
  & Efron & 0.013 & 0.089 & 0.090 & 92.28 & 0.313 \\
  & Exact & 0.078 & 0.099 & 0.126 & 85.32 & 0.356 \\
  & PL-Cox & $-$0.026 & 0.080 & 0.084 & 92.91 & 0.311 \\
  & GPL-Cox & $-$0.101 & 0.069 & 0.122 & 79.52 & 0.323 \\
  & Ren & 0.036 & 0.096 & 0.102 & 87.26 & 0.309 \\
  & Tamano & $-$0.031 & 0.079 & 0.085 & 92.81 & 0.310 \\
\hline
28 & Breslow & $-$0.071 & 0.072 & 0.101 & 87.34 & 0.314 \\
  & Efron & 0.003 & 0.087 & 0.087 & 93.24 & 0.318 \\
  & Exact & 0.136 & 0.110 & 0.174 & 74.65 & 0.403 \\
  & PL-Cox & $-$0.068 & 0.073 & 0.100 & 87.20 & 0.314 \\
  & GPL-Cox & $-$0.083 & 0.070 & 0.109 & 90.72 & 0.354 \\
  & Ren & 0.040 & 0.100 & 0.108 & 85.76 & 0.314 \\
  & Tamano & $-$0.072 & 0.072 & 0.101 & 87.23 & 0.313 \\
\hline
\end{tabular}
\par\smallskip{\footnotesize Note: Summaries are based on successful fits only.
Ren had numerical failures in 10 of the 40 settings; the numbers of successful fits ranged from 9,906 to 10,000. Setting-specific counts are provided in the accompanying Supplementary Excel file.}
\end{table}

\begin{table}[htbp]
\centering
\caption{Simulation results for the regression coefficient $\beta_4$ under continuous survival times generated from non-proportional hazard model with heavier tail scenario ($n=300$). The coarsening levels indicate the observation grid width applied to the event and censoring times. Reported metrics are empirical bias (Bias), the Monte Carlo standard deviation (SD) of the point estimates, root mean squared error (RMSE), coverage probability of the 95\% interval (CP), and average interval width (AW).}
\begin{tabular}{ccrcccc}
\hline
Sce & Method & \multicolumn{1}{c}{Bias} & SD & RMSE & CP & AW \\
\hline
1 & Breslow & 0.004 & 0.077 & 0.077 & 93.73 & 0.284 \\
  & Efron & 0.006 & 0.077 & 0.078 & 93.55 & 0.284 \\
  & Exact & 0.008 & 0.078 & 0.078 & 93.53 & 0.286 \\
  & PL-Cox & 0.007 & 0.078 & 0.078 & 93.23 & 0.284 \\
  & GPL-Cox & $-$0.001 & 0.075 & 0.075 & 93.74 & 0.282 \\
  & Ren & 0.013 & 0.080 & 0.081 & 91.55 & 0.281 \\
  & Tamano & 0.004 & 0.077 & 0.077 & 93.50 & 0.283 \\
\hline
7 & Breslow & $-$0.008 & 0.074 & 0.074 & 94.60 & 0.285 \\
  & Efron & 0.006 & 0.078 & 0.078 & 93.59 & 0.286 \\
  & Exact & 0.022 & 0.082 & 0.085 & 92.86 & 0.301 \\
  & PL-Cox & $-$0.005 & 0.075 & 0.075 & 94.12 & 0.286 \\
  & GPL-Cox & 0.000 & 0.076 & 0.075 & 95.08 & 0.295 \\
  & Ren & 0.016 & 0.082 & 0.084 & 90.80 & 0.283 \\
  & Tamano & $-$0.008 & 0.074 & 0.074 & 94.46 & 0.285 \\
\hline
14 & Breslow & $-$0.021 & 0.070 & 0.074 & 94.67 & 0.287 \\
  & Efron & 0.005 & 0.078 & 0.078 & 93.54 & 0.288 \\
  & Exact & 0.037 & 0.086 & 0.094 & 91.77 & 0.319 \\
  & PL-Cox & $-$0.019 & 0.071 & 0.074 & 94.43 & 0.286 \\
  & GPL-Cox & 0.017 & 0.080 & 0.082 & 94.39 & 0.312 \\
  & Ren & 0.019 & 0.084 & 0.087 & 90.07 & 0.285 \\
  & Tamano & $-$0.021 & 0.070 & 0.074 & 94.72 & 0.286 \\
\hline
28 & Breslow & $-$0.048 & 0.065 & 0.081 & 93.03 & 0.290 \\
  & Efron & $-$0.002 & 0.078 & 0.078 & 94.28 & 0.292 \\
  & Exact & 0.062 & 0.095 & 0.113 & 89.35 & 0.354 \\
  & PL-Cox & $-$0.047 & 0.065 & 0.081 & 92.53 & 0.290 \\
  & GPL-Cox & 0.051 & 0.092 & 0.105 & 90.39 & 0.349 \\
  & Ren & 0.021 & 0.088 & 0.091 & 89.10 & 0.290 \\
  & Tamano & $-$0.048 & 0.065 & 0.081 & 93.00 & 0.290 \\
\hline
\end{tabular}
\par\smallskip{\footnotesize Note: Summaries are based on successful fits only.
Ren had numerical failures in 10 of the 40 settings; the numbers of successful fits ranged from 9,906 to 10,000. Setting-specific counts are provided in the accompanying Supplementary Excel file.}
\end{table}

\newpage

\section{Additional results for real data analysis}

\subsection{Vital Data}

The VITAL application uses \texttt{BayesPLCox} version 0.5.0 and complete-case data comprising 23,972 individuals and 1,507 cancer events.
Original follow-up times were additionally rounded to 0.05, 0.10, and 0.25 years.
The covariates included vitamin D and marine n-3 fatty acid
treatment assignments, age, sex, race/ethnicity, body mass
index, current smoking, antihypertensive medication use,
statin use, and diabetes status.
For PL-Cox and GPL-Cox, independent $\mathrm{N}(0,100)$ priors were assigned to the regression coefficients, and the covariates were standardized internally.
Posterior coefficient draws were transformed back to the original covariate scale for reporting.
The Tamano method was implemented using the authors' Python code, with four chains of 8,000 iterations each, including 3,000 burn-in iterations.
We assess the four-chain PL-Cox and GPL-Cox runs using the prespecified criteria $\widehat R\leq1.01$, Bulk ESS $\geq400$, and Tail ESS $\geq400$ simultaneously for all monitored regression coefficients (including the GPL intercept).
The PL-Cox common-budget runs used $\delta=100$ without a fixed intercept or Metropolis--Hastings correction.
Both common-budget methods used 3,000 burn-in and 5,000 retained iterations per chain; neither met all convergence criteria under that budget.
The reported common-budget ESS/sec values are therefore descriptive and should not be interpreted as the efficiency of stationary posterior sampling.

\begin{table}[htbp]
\centering
\caption{Convergence diagnostics for the four-chain
PL-Cox and GPL-Cox analyses of VITAL.
Retained iterations per chain differ across
methods and coarsening conditions.
Pass requires $\widehat R\leq1.01$ and
both Bulk and Tail ESS values $\geq400$.
\label{tab:vital_longrun_diagnostics}}
\setlength{\tabcolsep}{4pt}
\begin{tabular}{llrrrc}
\hline
Coarsening & Method
& Max $\widehat R$
& Min Bulk ESS
& Min Tail ESS
& Pass \\
\hline
Original & PL-Cox
& 1.0054 & 787.7 & 1421.7 & Yes \\
Original & GPL-Cox
& 1.1387 & 19.7 & 11.3 & No \\
0.05 & PL-Cox
& 1.0058 & 1466.0 & 2068.4 & Yes \\
0.05 & GPL-Cox
& 1.0085 & 654.0 & 1196.0 & Yes \\
0.10 & PL-Cox
& 1.0033 & 1480.9 & 2432.4 & Yes \\
0.10 & GPL-Cox
& 1.0042 & 1302.1 & 2117.8 & Yes \\
0.25 & PL-Cox
& 1.0066 & 694.5 & 1126.4 & Yes \\
0.25 & GPL-Cox
& 1.0018 & 2339.7 & 4767.3 & Yes \\
\hline
\end{tabular}
\end{table}

The PL-Cox runs satisfied all three convergence criteria under all four coarsening conditions.
The GPL-Cox runs also satisfied the criteria for the 0.05-, 0.10-, and 0.25-year rounding conditions.
In contrast, GPL-Cox exhibited substantial mixing problems for the original follow-up times, with a maximum $\widehat R$ of 1.1387 and a minimum Bulk ESS of 19.7.
An extended run retaining 200,000 iterations per chain did not resolve these problems.

For consistency across the four coarsening conditions, the coefficient plots display all three methods using identical graphical symbols.
However, the GPL-Cox estimates and nominal credible intervals for the original follow-up times are presented only as descriptive summaries of nonconverged samples and are not used for posterior inference.
The convergence diagnostics are reported in Tables~\ref{tab:vital_longrun_diagnostics} and~\ref{tab:vital_200k_diagnostics}.
The plots use a common exponentiated-coefficient scale, although the GPL-Cox coefficients do not have the conventional Cox hazard-ratio interpretation.

\begin{table}[htbp]
\centering
\caption{Extended GPL-Cox convergence diagnostics for original VITAL follow-up times, across four chains (15 native coefficients and one derived quantity). The reported ESS minima include the derived centered-intercept quantity.\label{tab:vital_200k_diagnostics}}
\begin{tabular}{lrrrr}
\hline
Retained iteration window & Max $\widehat R$ & Min Bulk ESS & Min Tail ESS & Pass (out of 16) \\
\hline
First 40,000 & 1.1452 & 21.4 & 34.6 & 0 \\
First 80,000 & 1.0627 & 44.4 & 60.7 & 0 \\
First 120,000 & 1.1033 & 38.0 & 102.3 & 0 \\
All 200,000 & 1.0346 & 114.0 & 265.1 & 0 \\
Last 100,000 & 1.0611 & 62.7 & 146.5 & 0 \\
\hline
\end{tabular}
\end{table}

\begin{figure}[H]
\centering
\includegraphics[width=0.98\linewidth,height=0.78\textheight,keepaspectratio]{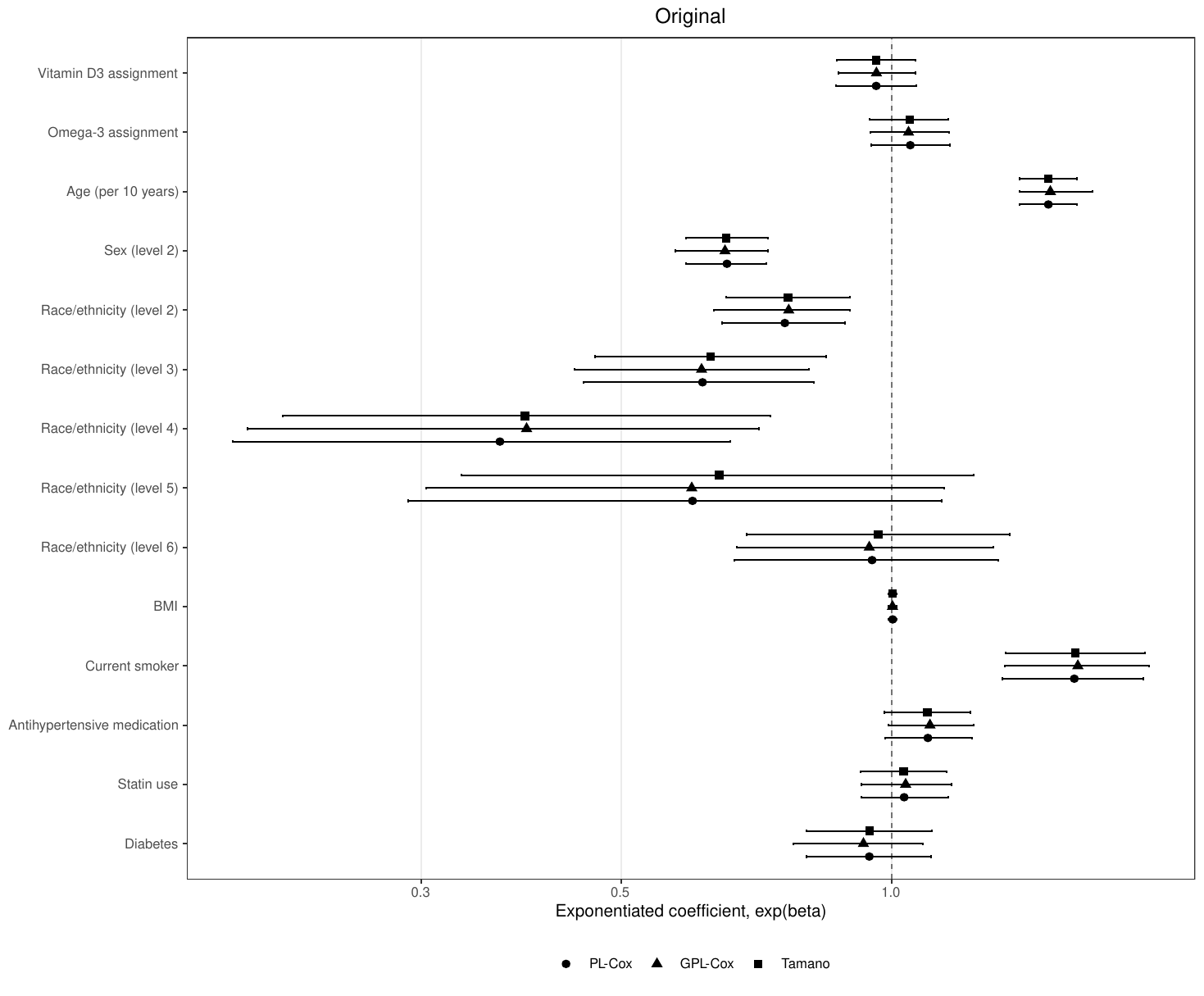}
\caption{VITAL coefficient summaries for Original on the exponentiated-coefficient scale. PL-Cox and Tamano estimates are Cox hazard ratios; GPL-Cox estimates are exponentiated latent-ranking coefficients, not conventional Cox hazard ratios. GPL-Cox Original is displayed for graphical consistency only: its point estimate and nominal 95\% interval come from nonconverged MCMC draws and must not be used for inference.\label{fig:vital_combined_full_original}}
\end{figure}

\begin{figure}[H]
\centering
\includegraphics[width=0.98\linewidth,height=0.78\textheight,keepaspectratio]{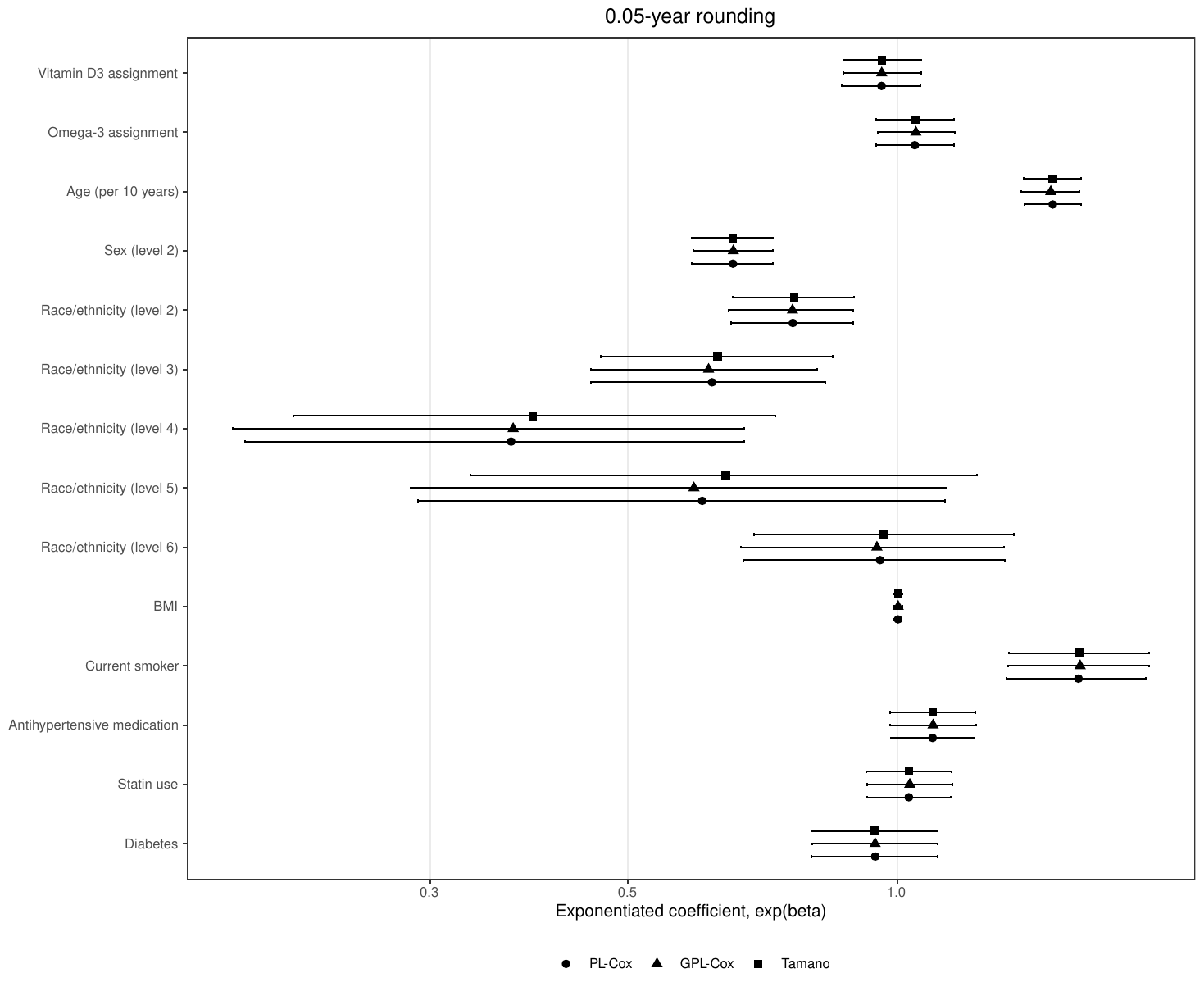}
\caption{VITAL coefficient summaries for 0.05-year rounding: posterior point estimates and 95\% credible intervals on the exponentiated-coefficient scale. PL-Cox and Tamano estimates are Cox hazard ratios; GPL-Cox estimates are exponentiated latent ranking coefficients and do not have the conventional Cox hazard-ratio interpretation. PL-Cox satisfied the prespecified convergence criteria for this condition.\label{fig:vital_combined_full_005}}
\end{figure}

\begin{figure}[H]
\centering
\includegraphics[width=0.98\linewidth,height=0.78\textheight,keepaspectratio]{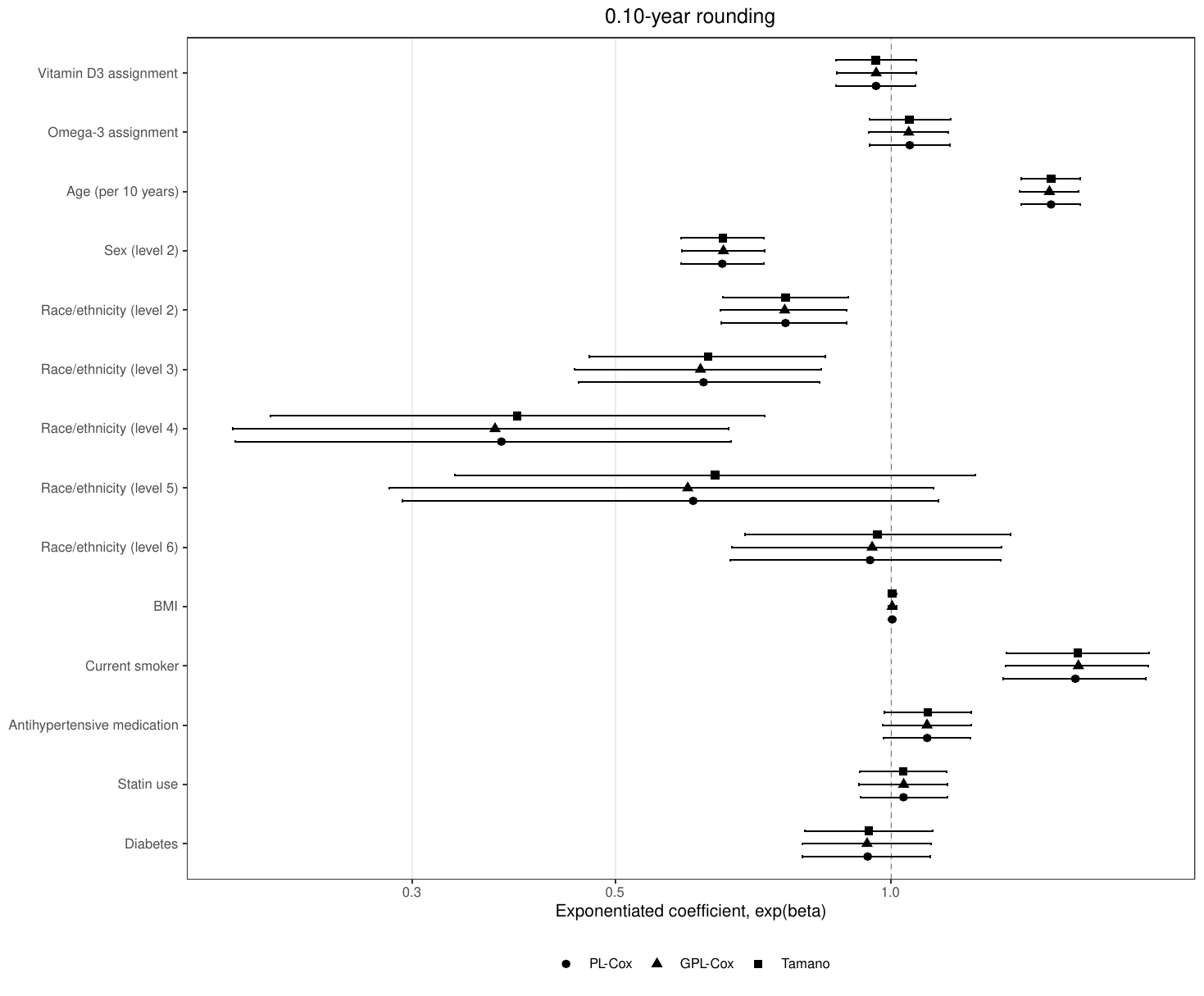}
\caption{VITAL coefficient summaries for 0.10-year rounding: posterior point estimates and 95\% credible intervals on the exponentiated-coefficient scale. PL-Cox and Tamano estimates are Cox hazard ratios; GPL-Cox estimates are exponentiated latent ranking coefficients and do not have the conventional Cox hazard-ratio interpretation. PL-Cox satisfied the prespecified
convergence criteria for this condition.\label{fig:vital_combined_full_010}}
\end{figure}

\begin{figure}[H]
\centering
\includegraphics[width=0.98\linewidth,height=0.78\textheight,keepaspectratio]{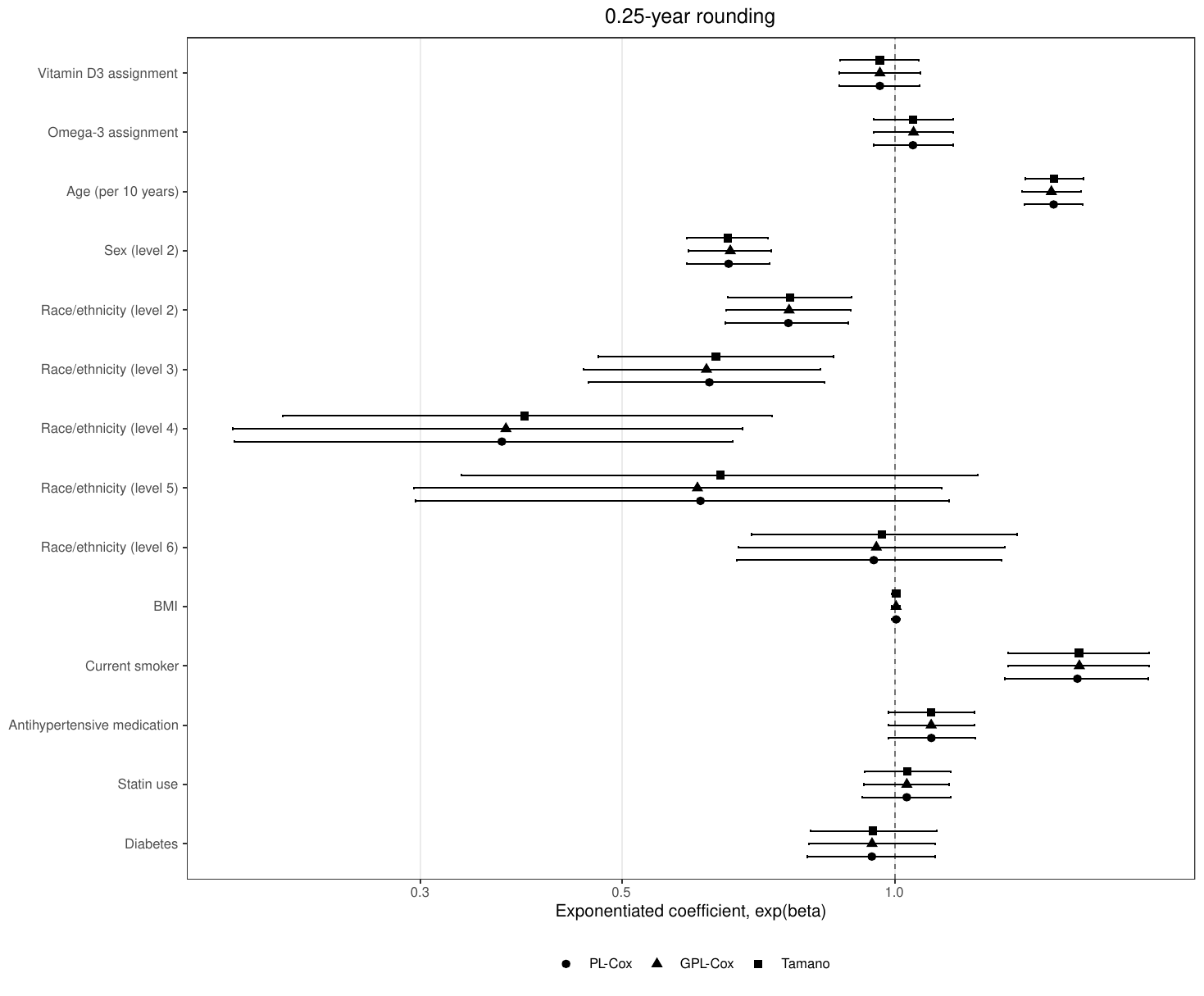}
\caption{VITAL coefficient summaries for 0.25-year rounding: posterior point estimates and 95\% credible intervals on the exponentiated-coefficient scale. PL-Cox and Tamano estimates are Cox hazard ratios; GPL-Cox estimates are exponentiated latent ranking coefficients and do not have the conventional Cox hazard-ratio interpretation.\label{fig:vital_combined_full_025}}
\end{figure}

\subsection{Readmissions Data\label{supp:readmissions}}

The readmissions analysis comprised 403 subjects, 861 observation rows, and 458 observed events.
PL-Cox and GPL-Cox were implemented using \texttt{BayesPLCox} version 0.5.0 and compared with Ren and Stan-Breslow.
For PL-Cox, we used $\delta=100$ without a fixed intercept or Metropolis--Hastings correction.
For PL-Cox, GPL-Cox, and Stan-Breslow, we assigned independent $\mathrm{N}(0,100)$ priors to the regression coefficients and an $\mathrm{IG}(0.001,0.001)$ prior to the shared-frailty variance.
The fixed-effect covariates were standardized during posterior computation, and the coefficient draws were transformed back to the original scale for reporting.
Ren was fitted using its original implementation with \texttt{prior\_sigma=1000}; thus, the prior specifications were not identical across all methods.

All methods were fitted using four chains. 
We distinguish common-budget runs from longer runs used for posterior inference.
Table~\ref{tab:readmissions_convergence} summarizes the convergence diagnostics for the regression coefficients and frailty variance, including the GPL-Cox intercept where applicable.
All four methods satisfied the prespecified convergence criteria in their respective longer runs, whereas the common-budget runs did not meet all criteria.
The common-budget results are therefore reported only for descriptive computational comparisons.

\begin{table}[htbp]
\centering
\small
\caption{MCMC convergence diagnostics for the readmissions shared-frailty analysis.
The common-budget runs use 3,000 burn-in and 5,000 retained iterations per chain.
The longer runs retain 20,000, 80,000, 40,000, and 5,000 iterations per chain for PL-Cox, GPL-Cox, Ren, and Stan-Breslow, respectively.
Pass requires $\widehat R\leq1.01$ and
both Bulk and Tail ESS values $\geq400$.
\label{tab:readmissions_convergence}}
\setlength{\tabcolsep}{4pt}
\begin{tabular}{llrrrc}
\hline
Run & Method
& Max $\widehat R$
& Min Bulk ESS
& Min Tail ESS
& Pass \\
\hline
Common & PL-Cox
& 1.0371 & 113.9 & 134.5 & No \\
Common & GPL-Cox
& 1.0420 & 104.8 & 234.2 & No \\
Common & Ren
& 1.0381 & 102.2 & 324.5 & No \\
\hline
Long & PL-Cox
& 1.0050 & 483.3 & 1214.9 & Yes \\
Long & GPL-Cox
& 1.0049 & 1237.9 & 3002.7 & Yes \\
Long & Ren
& 1.0045 & 868.5 & 2281.1 & Yes \\
Long & Stan-Breslow
& 1.0002 & 6935.8 & 12729.8 & Yes \\
\hline
\end{tabular}
\end{table}

The coefficient plot in Figure~\ref{fig:readmissions_forest} compares all four methods using exponentiated coefficient summaries and 95\% credible intervals. 
As in the VITAL analysis, the GPL-Cox entries are not conventional Cox hazard ratios. 
The corresponding frailty-variance intervals are shown in Figure~\ref{fig:readmissions_variance}.

\begin{figure}[H]
\centering
\includegraphics[width=0.93\linewidth,height=0.77\textheight,keepaspectratio]{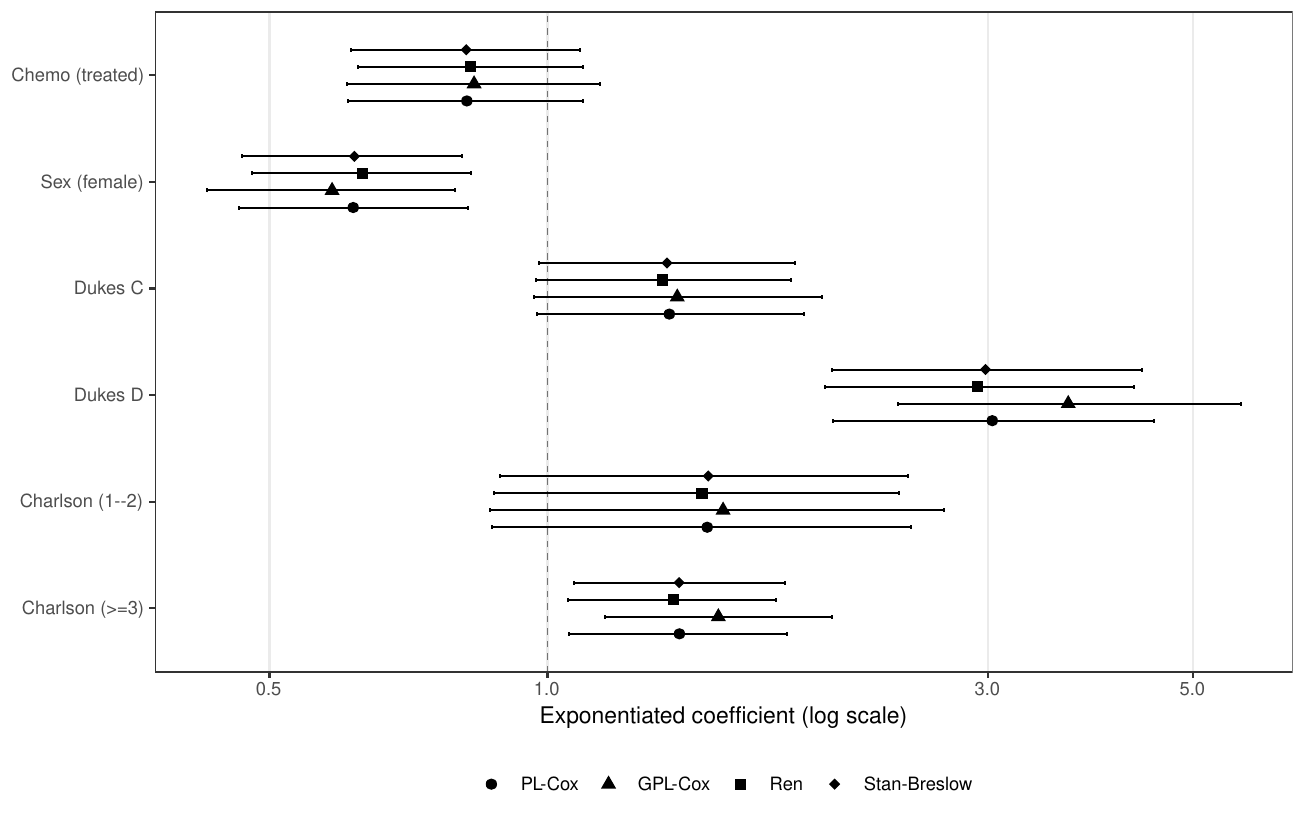}
\caption{Exponentiated regression-coefficient summaries for PL-Cox, GPL-Cox, Ren, and Stan-Breslow on a common plotting axis.
GPL-Cox exponentiated coefficients are not conventional Cox hazard ratios.}
\label{fig:readmissions_forest}
\end{figure}

\begin{figure}[H]
\centering
\includegraphics[width=0.85\linewidth,height=0.77\textheight,keepaspectratio]{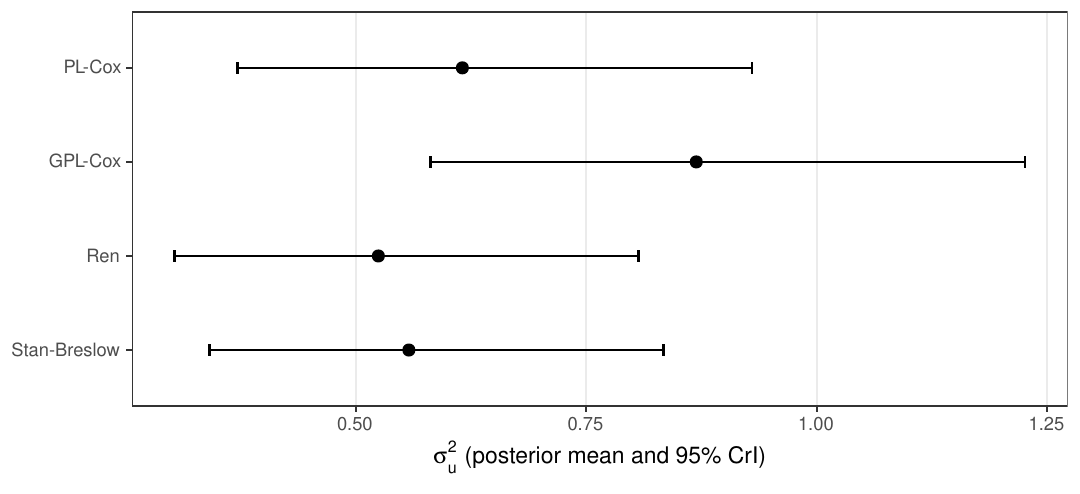}
\caption{Posterior mean and 95\% credible interval for the shared-frailty variance across the four methods.}
\label{fig:readmissions_variance}
\end{figure}

\begin{figure}[H]
\centering
\includegraphics[width=0.95\textwidth,height=0.66\textheight,keepaspectratio]{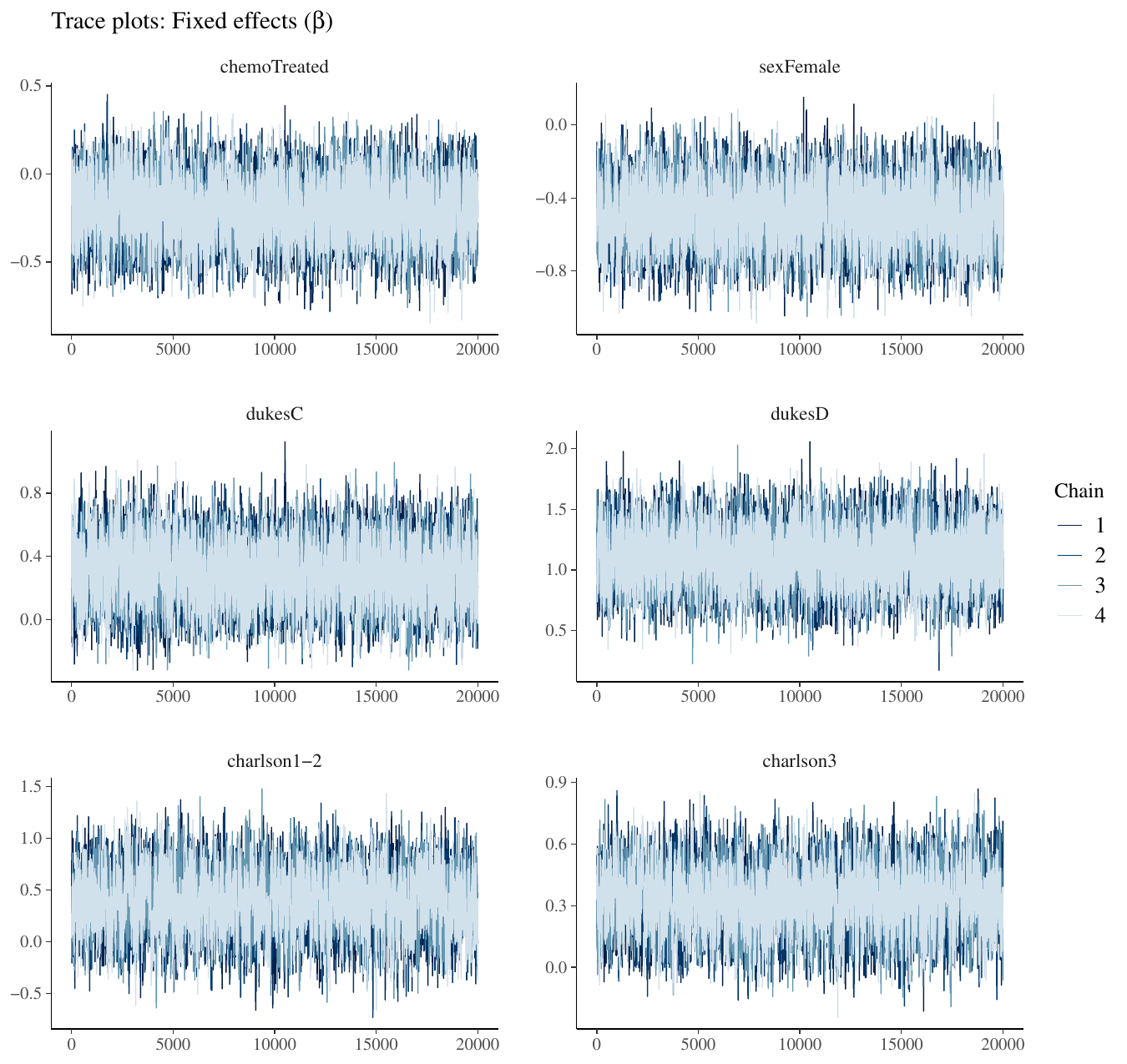}
\vspace{2mm}
\includegraphics[width=0.95\textwidth,height=0.18\textheight,keepaspectratio]{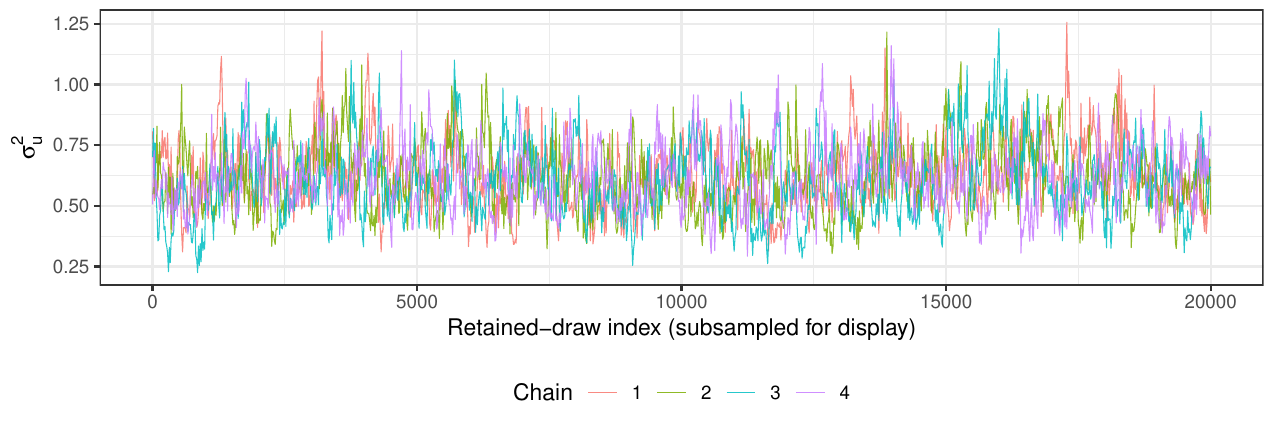}
\caption{Trace plots of the regression coefficients and frailty variance
for the PL-Cox model based on the extended MCMC run.}
\label{fig:readmissions_trace_pl}
\end{figure}

\begin{figure}[H]
\centering
\includegraphics[width=0.95\textwidth,height=0.66\textheight,keepaspectratio]{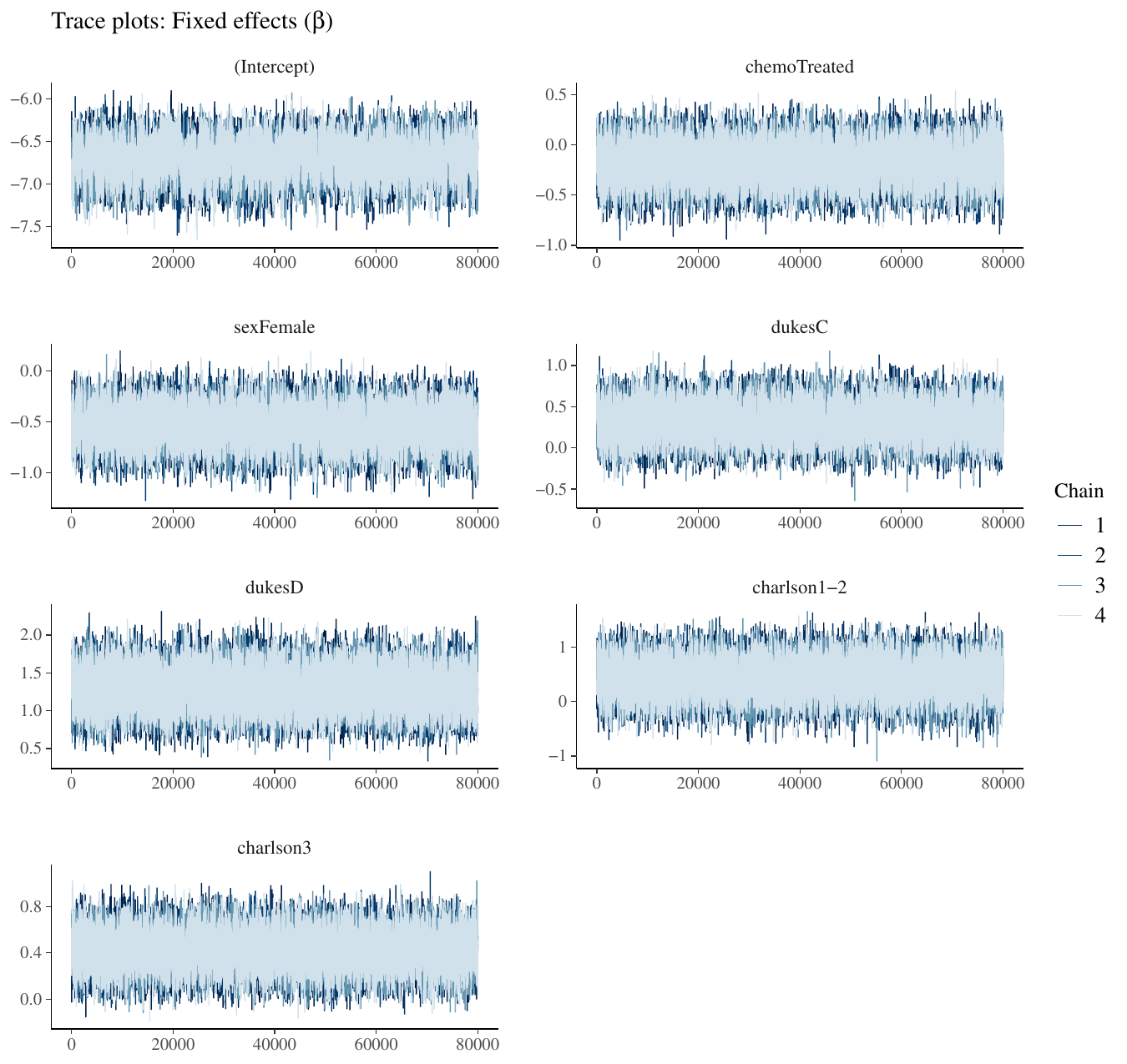}
\vspace{2mm}
\includegraphics[width=0.95\textwidth,height=0.18\textheight,keepaspectratio]{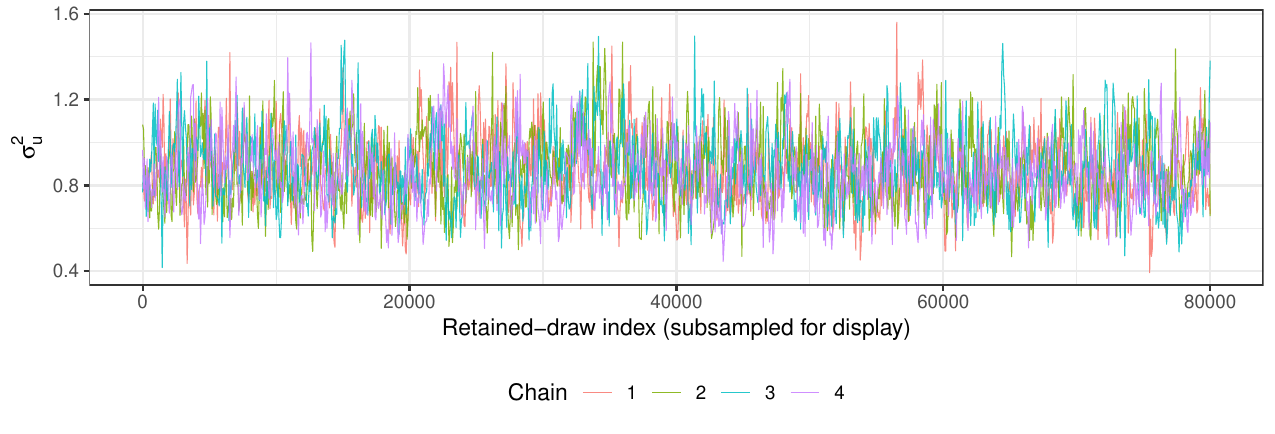}
\caption{Trace plots of the regression coefficients and frailty variance
for the GPL-Cox model based on the extended MCMC run.}
\label{fig:readmissions_trace_gpl}
\end{figure}

\begin{figure}[H]
\centering
\includegraphics[width=0.93\linewidth,height=0.77\textheight,keepaspectratio]{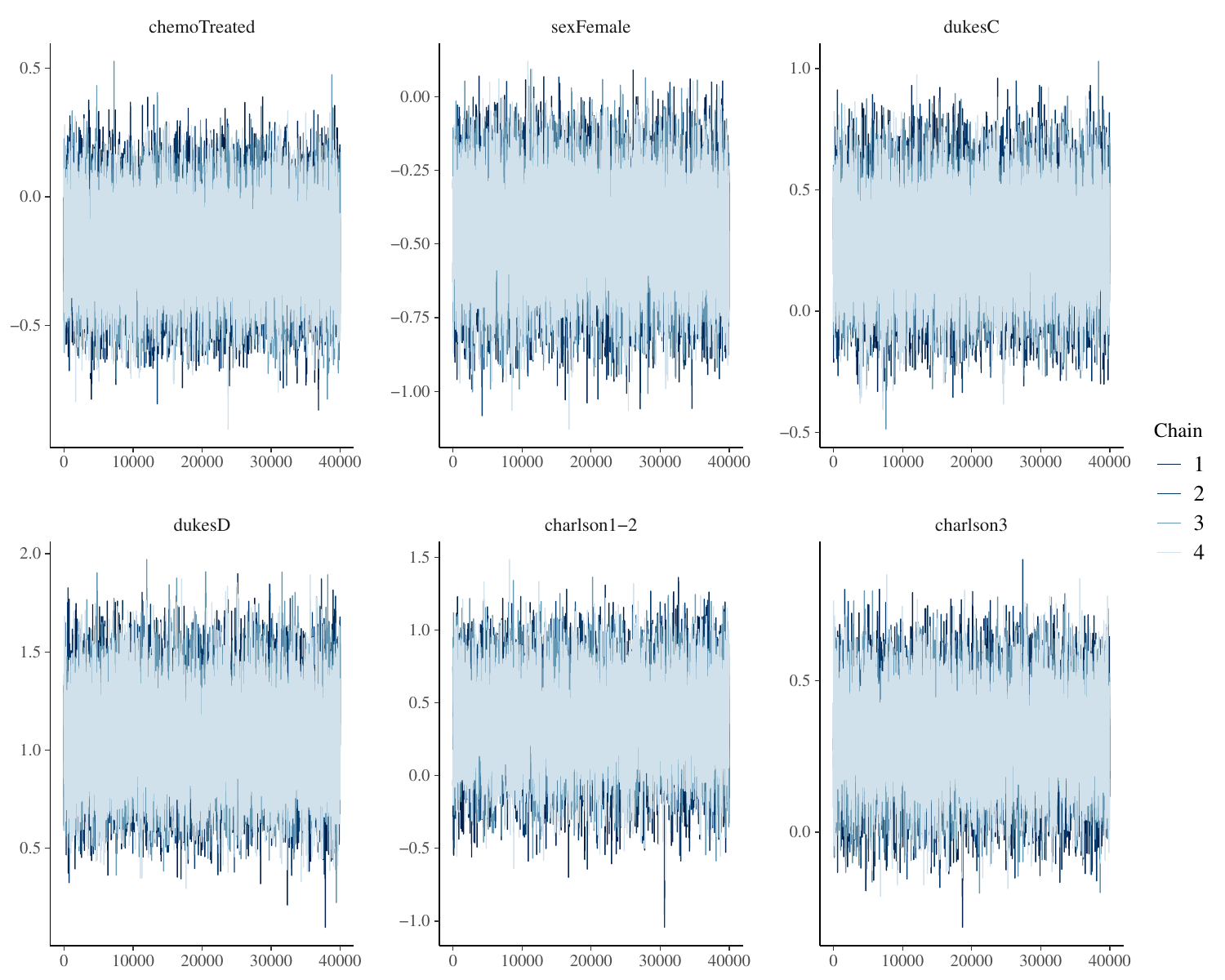}
\caption{Four-chain trace plots for Ren, including frailty variance, from the long run.}
\label{fig:readmissions_trace_ren}
\end{figure}

\begin{figure}[H]
\centering
\includegraphics[width=0.98\linewidth,height=0.77\textheight,keepaspectratio]{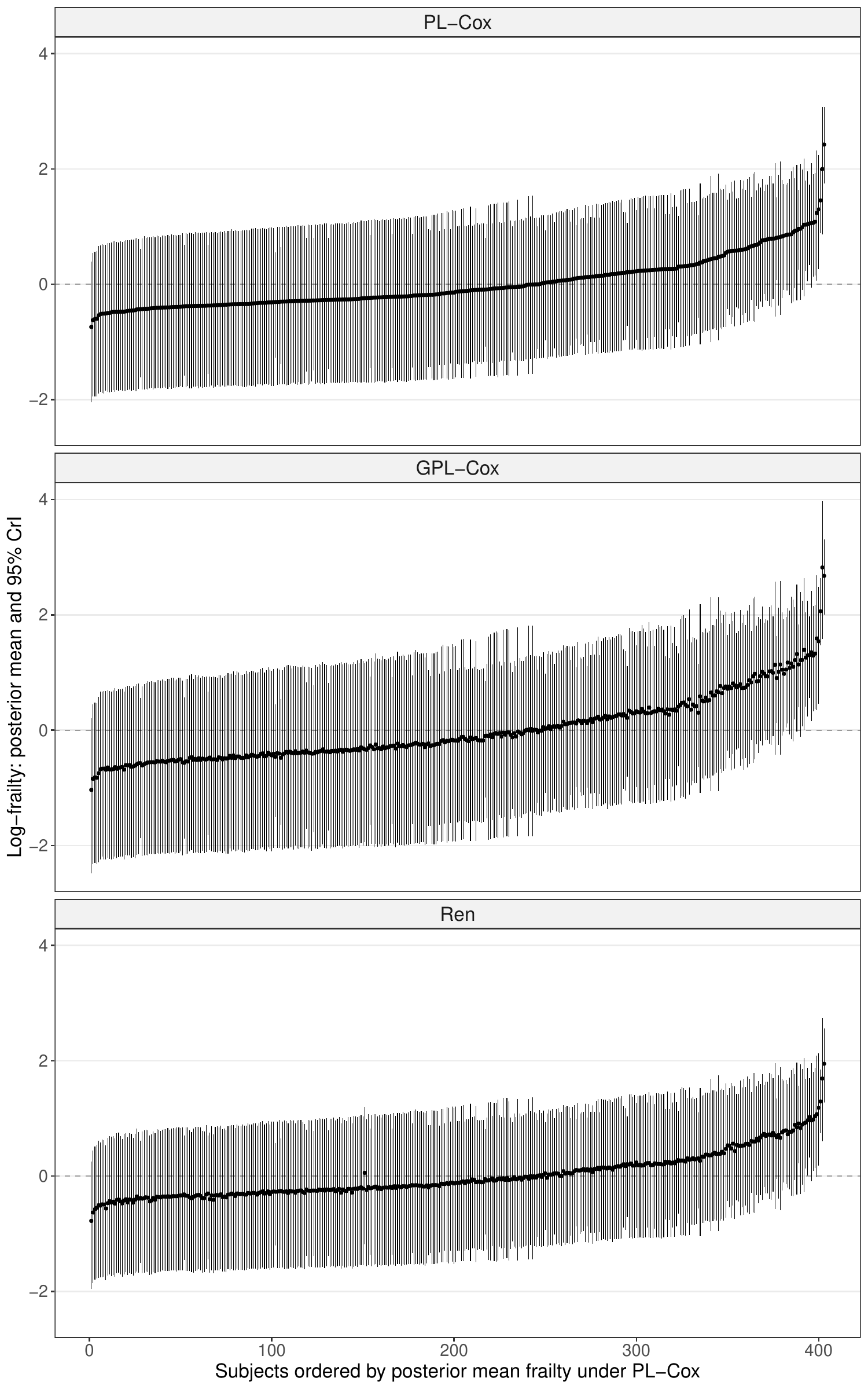}
\caption{Subject-level posterior log-frailty means and 95\% credible intervals for all three methods, with subjects ordered by their PL-Cox posterior means.}
\label{fig:readmissions_caterpillar}
\end{figure}

\newpage

\subsection{Right Heart Catheterization Data\label{rhc}}

To apply the model in a practical setting, we analyzed the right heart catheterization (RHC) dataset derived from the Study to Understand Prognoses and Preferences for Outcomes and Risks of Treatments (SUPPORT), which was originally conducted to evaluate the effect of right heart catheterization during the initial care of critically ill patients in the intensive care unit (ICU) on survival time up to 30 days \citep{connorsEffectivenessRightHeart1996}. 
In this study, RHC was performed in 2{,}184 patients within the first 24 hours of ICU stay, while 3{,}551 patients were managed without RHC. 
Therefore, the dataset used in our analysis consisted of 5{,}735 patients.
The outcome of interest was the time to death, measured in days from study entry.
Among the 1{,}918 observed deaths, only 29 distinct event times were recorded, implying that all events occurred in tied blocks.
The largest tie block contained 189 deaths.
Figure~\ref{fig:rhc_tie_hist} displays the distribution of the number of deaths occurring at each event time, illustrating the presence of several extremely large tie blocks throughout the follow-up period.

Due to the presence of extremely large tie blocks, the Exact method could not be used to obtain parameter estimates.
Therefore, this dataset provides a useful setting for examining the behavior of different approaches to handling tied event times.
We fitted the Cox model using the Breslow and Efron, PL-Cox, and GPL-Cox methods.
For comparison, we also applied the method of \citet{tamanoEfficientGibbsSampling2025} (Tamano).
For PL-Cox and GPL-Cox, independent $\mathrm{N}(0,100)$ priors were assigned to the regression coefficients, and the covariates were standardized internally.
Both methods used four chains with 3,000 burn-in and 5,000 retained iterations per chain in the common-budget analysis.
For the extended PL-Cox analysis, 10,000 iterations were retained per chain after 3,000 burn-in iterations.
Table~\ref{tab:rhc_results} reports exponentiated coefficient estimates and 95\% intervals. 
The results from Efron, Breslow, PL-Cox, and Tamano can be interpreted as Cox hazard ratios. 
GPL-Cox estimates are reported on their native model scale, so exponentiating the coefficients does not produce conventional Cox hazard ratios. 
Point estimates from the Cox-based methods were similar, while computation times differed substantially. 
The PL-Cox long run was used for its final coefficient estimates because the common-budget run did not meet all prespecified MCMC diagnostics; the equal-budget results were retained for comparing computational efficiency. 
In particular, PL-Cox was faster than GPL-Cox for the four-chain common-budget run, whereas the publicly available Python implementation of Tamano required substantially more time.

The PL-Cox common-budget run did not satisfy the prespecified convergence criteria (maximum $\widehat R = 1.016$, minimum bulk ESS = 346.4).
We therefore used the extended PL-Cox run for coefficient inference, which satisfied all criteria (maximum $\widehat R = 1.009$, minimum bulk ESS = 634.5, minimum tail ESS = 1479.8).

For computational comparisons, we retained the separate common-budget results. 
Accordingly, the reported computation times and ESS per second should not be interpreted as the computational efficiency of the extended PL-Cox run.

\begin{table}[!htbp]
\centering
\caption{Exponentiated coefficient estimates for the RHC data, expressed per indicated covariate increment.}
\label{tab:rhc_results}
\scriptsize
\resizebox{\linewidth}{!}{%
\begin{tabular}{lccccc}
\hline
 & Efron (CI) & Breslow (CI) & PL-Cox (CrI) & GPL-Cox (CrI) & Tamano (CrI) \\
\hline
RHC & 1.21 [1.10, 1.33] & 1.21 [1.10, 1.32] & 1.20 [1.09, 1.32] & 1.23 [1.12, 1.35] & 1.21 [1.10, 1.32] \\
Age ($+10$ years) & 1.11 [1.08, 1.15] & 1.11 [1.08, 1.15] & 1.11 [1.08, 1.14] & 1.12 [1.09, 1.16] & 1.11 [1.08, 1.15] \\
Sex (female) & 0.98 [0.90, 1.08] & 0.98 [0.90, 1.08] & 0.99 [0.90, 1.08] & 0.98 [0.89, 1.07] & 0.98 [0.90, 1.08] \\
Mean BP ($+10$ mmHg) & 0.96 [0.95, 0.97] & 0.96 [0.95, 0.98] & 0.96 [0.95, 0.98] & 0.96 [0.95, 0.97] & 0.96 [0.95, 0.98] \\
WBC ($+10$ recorded units) & 1.03 [1.00, 1.07] & 1.03 [1.00, 1.07] & 1.03 [1.00, 1.07] & 1.04 [1.00, 1.07] & 1.03 [1.00, 1.07] \\
Heart rate ($+10$ bpm) & 1.02 [1.01, 1.03] & 1.02 [1.01, 1.03] & 1.02 [1.01, 1.03] & 1.02 [1.01, 1.03] & 1.02 [1.01, 1.03] \\
Resp. rate ($+10$ breaths/min) & 0.98 [0.95, 1.02] & 0.98 [0.95, 1.02] & 0.98 [0.95, 1.02] & 0.98 [0.95, 1.02] & 0.98 [0.95, 1.02] \\
Creatinine & 1.03 [1.01, 1.05] & 1.03 [1.01, 1.05] & 1.03 [1.01, 1.05] & 1.03 [1.01, 1.05] & 1.03 [1.01, 1.05] \\
Temperature & 0.98 [0.95, 1.01] & 0.98 [0.96, 1.01] & 0.98 [0.96, 1.01] & 0.98 [0.95, 1.00] & 0.98 [0.95, 1.01] \\
\hline
\end{tabular}}
\begin{minipage}{\linewidth}
\footnotesize
\emph{Notes:} CI denotes a Wald 95\% confidence interval (Efron and Breslow); CrI denotes a nominal 95\% posterior credible interval (Bayesian methods). Continuous covariate effects are reported for the increments indicated in the row labels; WBC units refer to the recorded data scale. PL-Cox coefficient estimates are from the converged extended run, while computation times elsewhere are from the common-budget run. GPL-Cox values are exponentiated latent-ranking coefficients, not conventional Cox hazard ratios.
\end{minipage}
\end{table}

\begin{figure}[H]
\centering
\includegraphics[width=0.9\linewidth]{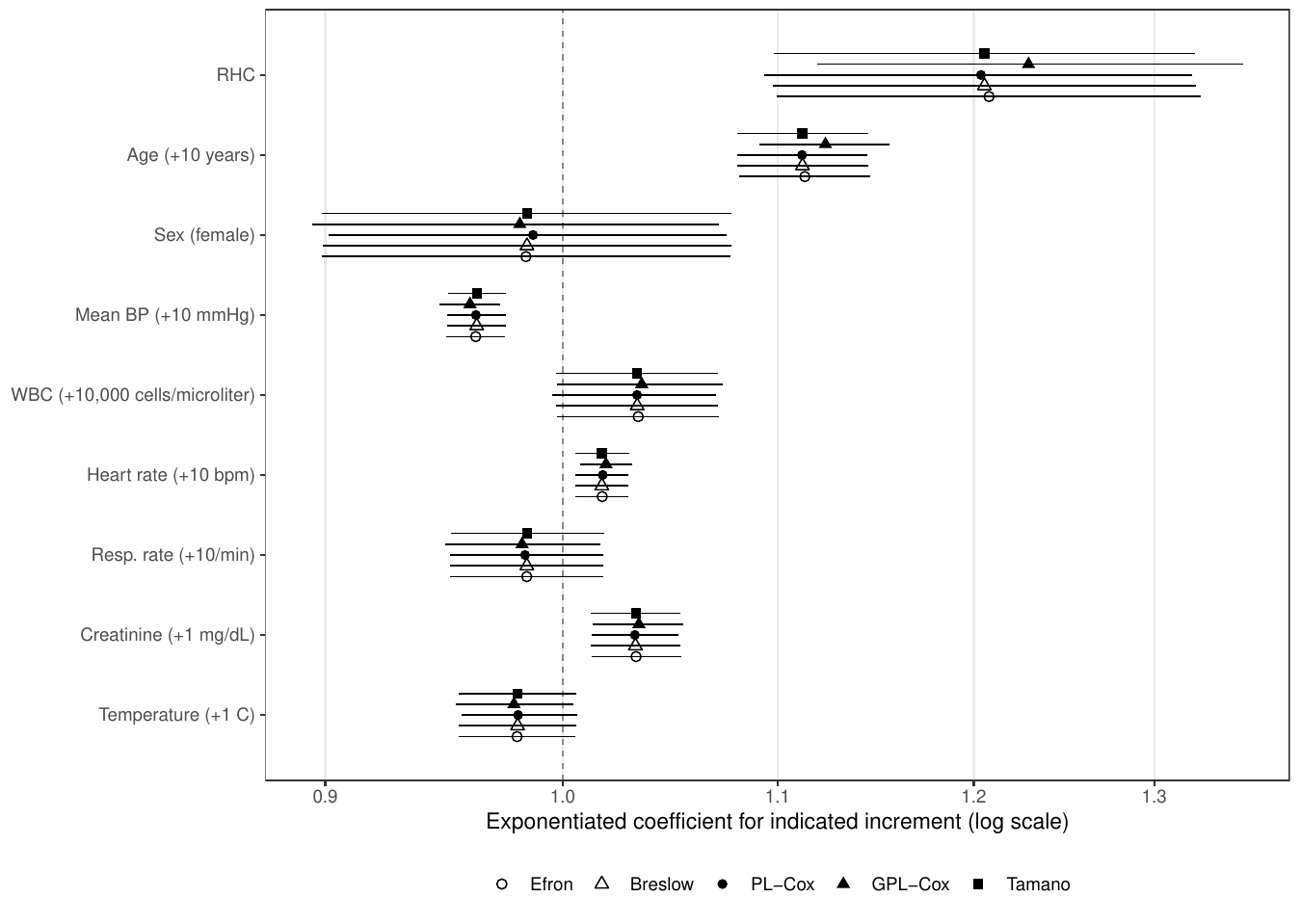}
\caption{Exponentiated coefficient estimates for the RHC data. 
Points and horizontal lines show point estimates and 95\% intervals, respectively.
Efron and Breslow use Wald confidence intervals, whereas the Bayesian methods use credible intervals.
PL-Cox estimates are from the extended run. 
GPL-Cox exponentiated coefficients are not conventional Cox hazard ratios.}
\label{fig:rhc_forest_beta}
\end{figure}

\begin{figure}[H]
	\centering
	\includegraphics[width=0.7\linewidth]{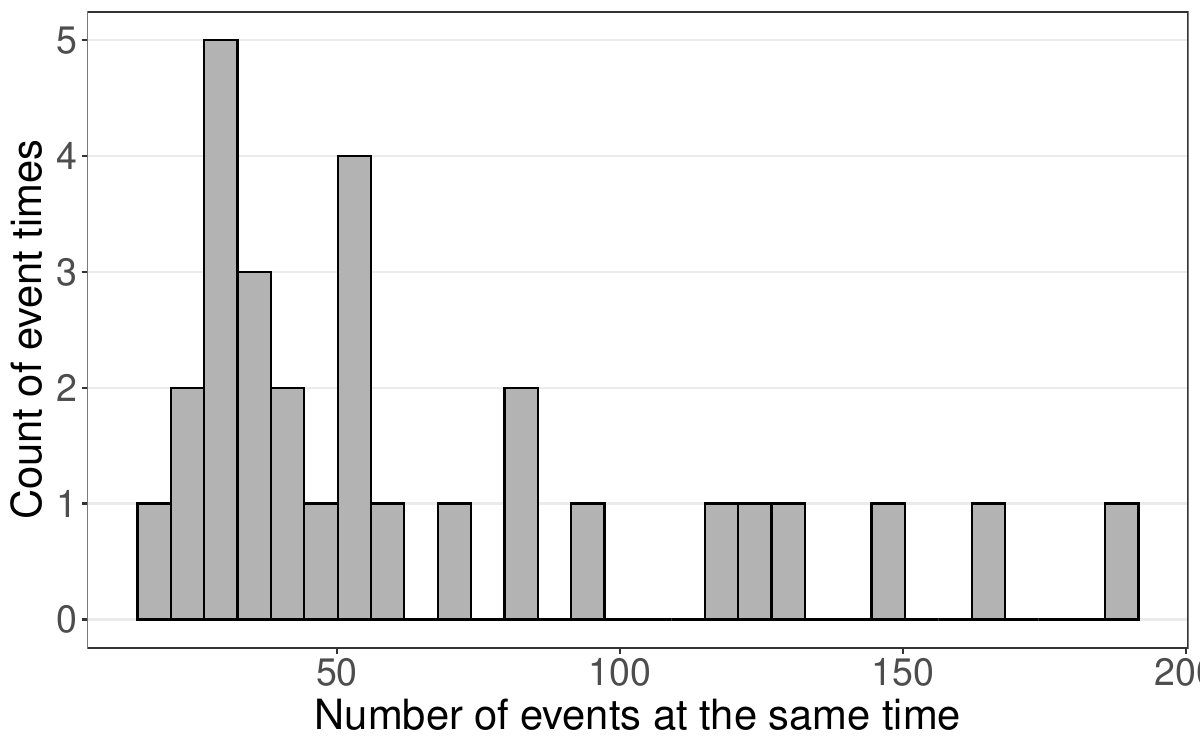}
	\caption{Distribution of tie block sizes.}
	\label{fig:rhc_tie_hist}
\end{figure}

\begin{table}[!htbp]
\centering
\caption{MCMC diagnostics and computational
performance for the RHC analysis.}
\label{tab:rhc_computation}
\small
\begin{tabular}{lrrrr}
\hline
Method & Max $\widehat R$
& Min bulk ESS & Time (s) & Min bulk ESS/s \\
\hline
PL-Cox       & 1.016 & 346.4  & 49.3    & 7.030 \\
GPL-Cox      & 1.005 & 1091.8 & 391.8   & 2.787 \\
Tamano       & 1.000 & 18000  & 51120.2 & 0.352 \\
\hline
\end{tabular}
\begin{minipage}{\linewidth}
\footnotesize
\emph{Notes:} The reported values are from the formal computational runs, rather than the extended PL-Cox run used for coefficient inference. 
The PL-Cox run did not satisfy the prespecified convergence criteria.
The values are descriptive and should be interpreted in light of the differing sampling implementations and run settings.
\end{minipage}
\end{table}

\begin{figure}[H]
	\centering
	\includegraphics[width=0.9\linewidth]{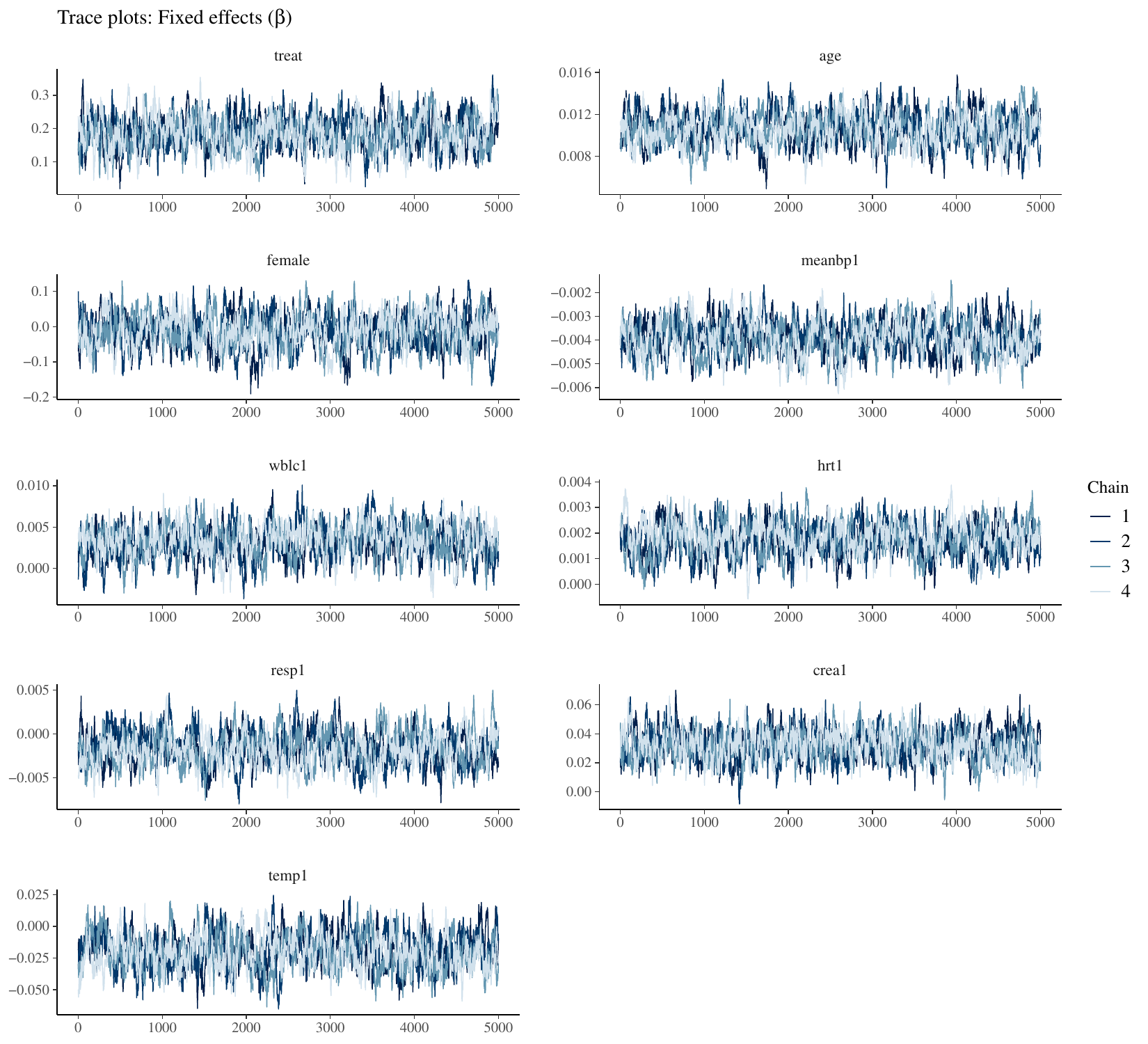}
    \caption{Four-chain trace plots for PL-Cox from the common-budget run, which did not satisfy the prespecified convergence criteria.}
	\label{fig:rhc_trace_pl}
\end{figure}

\begin{figure}[H]
	\centering
	\includegraphics[width=0.9\linewidth]{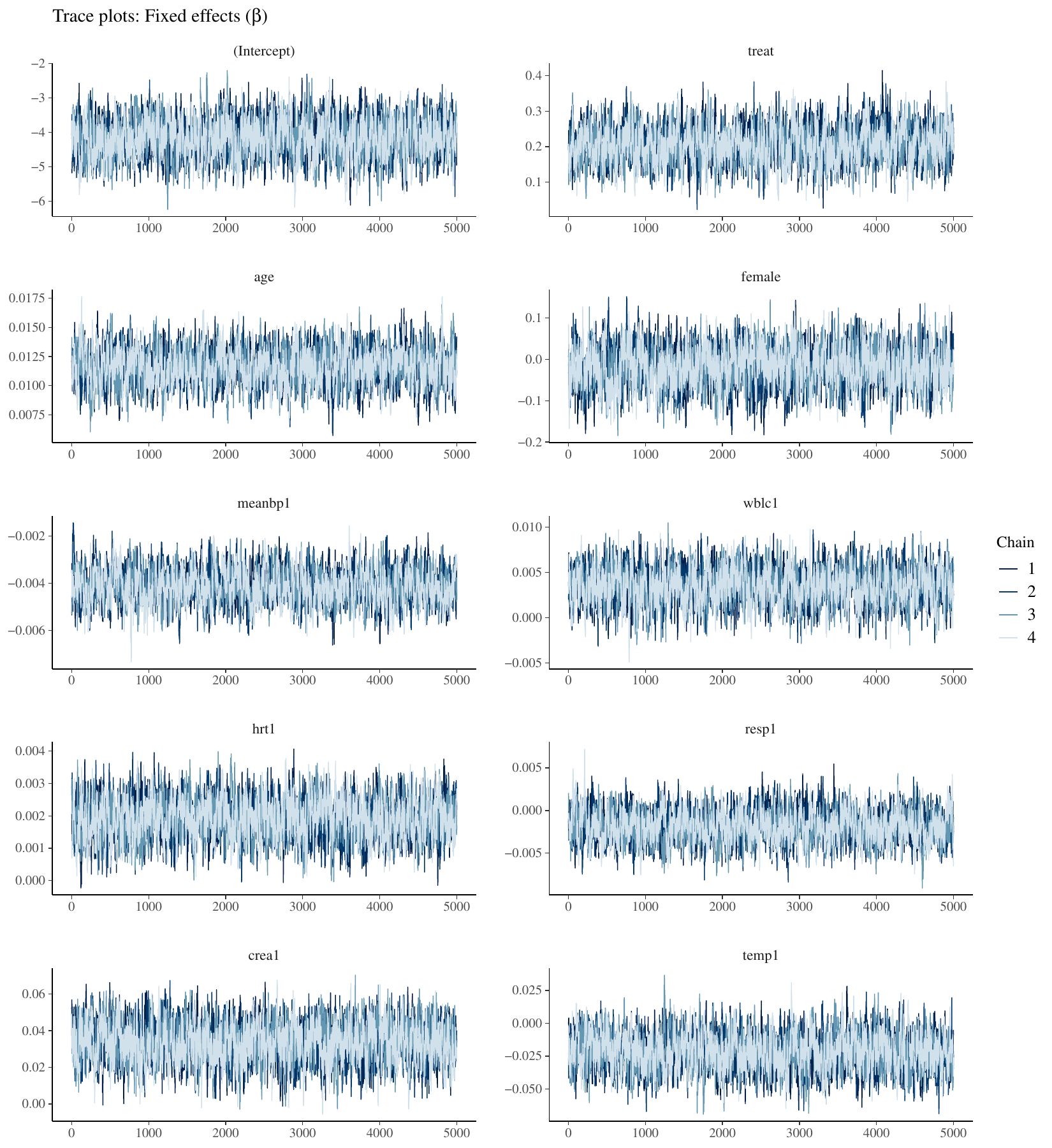}
    \caption{Four-chain trace plots for GPL-Cox from the formal run.}
	\label{fig:rhc_trace_gpl}
\end{figure}

% %---------------------------------------------%
% %                Reference                    %
% %---------------------------------------------%
% \bibliographystyle{biom}
% \bibliography{library}
% 
% \newpage

% \end{thebibliography}
% \end{document}

\end{document}